\documentclass[oneside]{amsart}
\usepackage{textcomp}
\usepackage[latin9]{inputenc}
\usepackage{color}
\usepackage{booktabs}
\usepackage{mathrsfs}
\usepackage{mathtools}
\usepackage{amstext}
\usepackage{amsthm}
\usepackage{amssymb}
\usepackage{graphicx}

\makeatletter

\providecommand{\tabularnewline}{\\}

\numberwithin{equation}{section}
\numberwithin{figure}{section}
\theoremstyle{plain}
\newtheorem{thm}{\protect\theoremname}
\theoremstyle{plain}
\newtheorem{prop}[thm]{\protect\propositionname}
\theoremstyle{definition}
\newtheorem{defn}[thm]{\protect\definitionname}
\theoremstyle{plain}
\newtheorem{lem}[thm]{\protect\lemmaname}
\theoremstyle{remark}
\newtheorem{rem}[thm]{\protect\remarkname}
\theoremstyle{plain}
\newtheorem{cor}[thm]{\protect\corollaryname}

\numberwithin{thm}{section}

\@ifundefined{date}{}{\date{}}
\ifdefined\showcaptionsetup
 \PassOptionsToPackage{caption=false}{subfig}
\fi
\usepackage{subfig}
\makeatother

\providecommand{\corollaryname}{Corollary}
\providecommand{\definitionname}{Definition}
\providecommand{\lemmaname}{Lemma}
\providecommand{\propositionname}{Proposition}
\providecommand{\remarkname}{Remark}
\providecommand{\theoremname}{Theorem}

\begin{document}
\title[Fractional Diffusion on Graphs: Laplacian Semigroups and Memory]{Fractional Diffusion on Graphs: Superposition of Laplacian Semigroups
and Memory}
\author{Nikita Deniskin and Ernesto Estrada{*}}
\address{Scuola Normale Superiore, Pisa, Italy. E--mail: nikita.deniskin@sns.it;
Institute for Cross-Disciplinary Physics and Complex Systems (IFISC),
CSIC-UIB, Palma de Mallorca, Spain. E-mail: estrada@ifisc.uib-csic.es}
\begin{abstract}
Subdiffusion on graphs is often modeled by time-fractional diffusion equations, yet its structural and dynamical consequences remain unclear. We show that subdiffusive transport on graphs is a memory-driven process generated by a random time change that compresses operational time, produces long-tailed waiting times, and breaks Markovianity while preserving linearity and mass conservation. We prove that Mittag-Leffler graph dynamics admit an exact convex, mass-preserving representation as a superposition of classical heat semigroups evaluated at rescaled times, revealing fractional diffusion as ordinary diffusion acting across multiple intrinsic time scales. This framework uncovers heterogeneous, vertex-dependent memory effects and induces transport biases absent in classical diffusion, including algebraic relaxation, degree-dependent waiting times, and early-time asymmetries between sources and neighbors. These features define a subdiffusive geometry on graphs enabling particles to locally discover global shortest paths while favoring high-degree regions. Finally, we show that time-fractional diffusion arises as a singular limit of multi-rate diffusion.

\medskip{}

MSC: 26A33; 15A16; 35R11; 47D06; 45D05

\medskip{}

Keywords: Subdiffusion on graphs; Time-fractional diffusion; Mittag-Leffler
functions; Memory kernels; Sum-of-exponentials; Shortest paths 
\end{abstract}

\maketitle
\tableofcontents{}

\section{Introduction}

Graphs $G=(V,E)$---also referred to as networks---provide a natural
mathematical framework to represent a wide variety of complex systems
arising in molecular, ecological, technological, and social contexts
\cite{estrada2012structure}. In this representation, the set of vertices
$V$ typically corresponds to the entities of the system, while the
set of edges $E$ encodes their interactions. A fundamental mechanism
governing the transport of mass, energy, or information across such
networks is diffusion \cite{lopez2016overview,masuda2017random}.
However, due to the complexity of the environments in which many real-world
networks are embedded, transport processes often deviate from classical
diffusive behavior \cite{nicolaides2010anomalous,metzler2000random,sokolov2012models,medina2022nontrivial}. 

Subdiffusion, characterized by a slower-than-linear growth of the
mean squared displacement, is especially prevalent in complex systems.
For instance, subdiffusive dynamics are ubiquitous in the crowded
interior of biological cells \cite{weiss2004anomalous,gupta2016protein,grimaldo2019dynamics,basak2019understanding},
where millions of macromolecules interact, forming intricate networks
of biochemical processes. Similarly, crowding effects caused by vehicle
density and driving fluctuations in urban transportation systems have
been shown to induce subdiffusive traffic states \cite{combinido2012crowding}.
Subdiffusion has also been observed in information transmission processes
on online social networks, such as Twitter (now X) and Digg \cite{foroozani2019anomalous}.

A standard mathematical approach to model such crowded and heterogeneous
systems is to replace the classical time derivative in the diffusion
equation with a fractional-time derivative \cite{evangelista2018fractional,sokolov2002fractional}.
When anomalous diffusion takes place on a graph $G$, this leads to
the following fractional diffusion equation: 
\begin{equation}
\left\{ \begin{aligned}D_{t}^{\alpha}u_{\theta}(t)+\theta\,L\,u_{\theta}(t) & =0,\\
u_{\theta}(0) & =u_{0},
\end{aligned}
\right.\label{eq:Main FDE}
\end{equation}
where $\theta>0$ denotes the diffusivity and $L$ is the graph Laplacian
operator acting on functions defined on the vertex set. Specifically,
$L$ is the linear mapping from $\mathbb{C}(V)$ into itself given
by 

\begin{equation}
(Lf)(v)\coloneqq\sum_{(v,w)\in E}\bigl(f(v)-f(w)\bigr),\qquad f\in\mathbb{C}(V),
\end{equation}
which is known as the graph Laplacian. The operator $D_{t}^{\alpha}$
denotes the Caputo time-fractional derivative of order $0<\alpha<1$
\cite{caputo1966linear}, defined by 

\begin{equation}
D_{t}^{\alpha}u(t)=\frac{1}{\Gamma(1-\alpha)}\int_{0}^{t}\frac{u'(\tau)}{(t-\tau)^{\alpha}}\,d\tau,
\end{equation}
where $u'(\tau)$ is the first derivative of $u$ evaluated
at time $\tau$. Despite their broad relevance in continuous settings,
time-fractional diffusion models have been only sparsely explored
on graphs and networks. To date, applications have been largely restricted
to epidemiological modeling {\cite{abadias2020fractional,d2025fractional}}
and to the fractional diffusion on the human proteome, proposed as
an alternative framework to account for the multi-organ damage associated
with SARS-CoV-2 infection {\cite{estrada2020fractional}}. An
exception is the use of \eqref{eq:Main FDE} in an engineering context
for searching consensus of autonomous systems, where it receive the
name of ``fractional order consensus'' \cite{sun2011convergence,yan2024consensus,sun2024group,huang2024distributed,lu2012consensus,cao2009distributed}.
A rapidly growing line of research explores the integration of fractional
derivatives into learning algorithms, leading to the development of
fractional-order neural networks and related architectures, with emerging
applications in machine learning and artificial intelligence \cite{pang2019fpinns,joshi2023survey}.

Apart from crowding and excluded-volume effects, which reduce mobility
by limiting the available space for motion, several additional mechanisms
can give rise to subdiffusion in complex systems. This happen, for
instance in ecological networks, where structural disorder and heterogeneity---manifested
through irregular geometries, bottlenecks, and hierarchical or fractal
structures---constrain transport pathways and significantly slow
down spatial exploration \cite{kim2024cover}. These features often
induce trapping events and lead to broad distributions of residence
times. Another characteristic of complex systems is the existence
of memory. Temporal memory effects have been observed, for instance,
in transcription regulator--DNA interactions in live bacterial cells
\cite{jung2024memory}, apart from the large existing evidence accumulated
on single-cell experiments. As memory refers to the phenomenon where
past events influence a system\textquoteright s current and future
states or behaviors, the evolution of these complex systems at a given
time depends on its entire past history rather than solely on its
instantaneous state.

In such non-Markovian settings, classical diffusion equations fail
to provide an adequate description. Instead, diffusion models with
memory kernels naturally arise, capturing the nonlocal-in-time response
induced by trapping, heterogeneity, and temporal correlations \cite{sandev2018generalized,gurtin1968general,Ponce2021,saif2023inverse,trimper2004memory,toan2022nonclassical}.
Once again, in the case of graphs, this leads to the generalized diffusion
equation 

\begin{equation}
\left\{ \begin{aligned}\int_{0}^{t}\gamma(t-s)\,\partial_{s}u_{\theta}(s)\,ds+\theta\,L\,u_{\theta}(t) & =0,\\
u_{\theta}(0) & =u_{0},
\end{aligned}
\right.
\end{equation}
where $\gamma(t)$ denotes a memory kernel. Such formulations provide
a unifying framework for describing subdiffusive transport and establish
a direct connection between microscopic mechanisms and macroscopic
anomalous diffusion. Related non-Markovian formulations on networks
arise naturally when vertex-to-vertex motion is modeled as a non-Poisson
continuous-time random walk. In this setting, the node-state evolution
is governed by a generalized master equation in which the graph Laplacian
(or transition operator) appears entwined in time with a kernel determined
by the waiting-time distribution, and memory effects are known to
substantially alter diffusion and mixing properties on (temporal)
networks \cite{hoffmann2013random,lambiotte2015effect,scholtes2014causality}.

While time-fractional diffusion equations and memory-kernel formulations
are widely used to model subdiffusive transport in continuous settings,
their role in the context of networks and graph-based diffusion has
so far been explored mainly at a phenomenological level. Existing
studies typically emphasize anomalous scaling, spectral properties,
or long-time relaxation, but often treat memory as a uniform slowing-down
mechanism acting on otherwise classical graph diffusion \cite{kosztolowicz2005measuring,kepten2015guidelines,gallos2007scaling}.
As a result, comparatively little is known about how non-Markovian
effects interact with the discrete geometry of a graph, how they influence
transport pathways, or how they modify vertex-level and path-level
behavior beyond global decay rates.

The present work is motivated by the need to clarify how subdiffusion
reshapes diffusion on graphs at the structural level. By exploiting
the subordination principle and a mass-preserving sum-of-exponentials
representation of the Mittag-Leffler operator, we connect fractional
diffusion on graphs to superposition of classical heat processes acting
at different internal times. This perspective allows us to study,
in a unified way, memory effects, vertex-dependent waiting times,
effective distances, and path selection in subdiffusive dynamics.
In doing so, we show that memory does not merely slow down diffusion
uniformly, but induces heterogeneous temporal behavior across vertices
and leads to well-defined subdiffusive distances and paths that differ
from their classical diffusive counterparts. These results provide
a concrete link between fractional models, memory kernels, and graph-based
transport, and offer a more detailed understanding of what subdiffusion
means when the underlying space is a network.

\subsection{Subdiffusion: physical mechanisms and mathematical models}

A central hallmark of anomalous diffusion is the deviation of the
mean squared displacement (MSD) from the linear-in-time growth predicted
by Fickian diffusion. Instead, one observes a power-law scaling \cite{sokolov2012models,bouchaud1990anomalous,metzler2000random}
\begin{equation}
\langle x^{2}(t)\rangle\sim t^{\alpha},\qquad0<\alpha<1,
\end{equation}
which defines subdiffusive behavior. As emphasized in \cite{sokolov2012models},
such subdiffusion is not a single phenomenon but rather a collective
outcome of different physical mechanisms, each leading to distinct
stochastic and mathematical descriptions.

One prominent mechanism underlying subdiffusion is trapping \cite{sokolov2002fractional,sokolov2012models,evangelista2018fractional}.
In crowded or energetically disordered environments, particles may
experience long waiting times between successive displacements due
to transient binding or deep potential wells. This situation is naturally
described by the continuous-time random walk (CTRW) framework, where
the waiting times between steps are independent random variables drawn
from a heavy-tailed distribution 
\begin{equation}
\psi(\tau)\sim\tau^{-1-\alpha},\qquad0<\alpha<1,
\end{equation}
implying a diverging mean waiting time.

In the scaling limit, CTRW dynamics lead to a fractional diffusion
equation (FDE) for the probability density $p(x,t)$. Sokolov \cite{sokolov2012models}
stresses that the use of such fractional equations is physically justified
only when the underlying trapping assumptions are valid. CTRW models
generally exhibit aging, weak ergodicity breaking, and large trajectory-to-trajectory
fluctuations.

Another class of subdiffusive systems arises from transport in labyrinthine
or fractal structures, such as percolation clusters or tortuous channel
networks \cite{polanowski2014simulation,sokolov2012models,meinecke2017multiscale,fanelli2010diffusion}.
In these systems, anomalous diffusion originates from geometric constraints
rather than trapping times. The particle explores a space with no
translational invariance and a broad distribution of path lengths.

A paradigmatic example is diffusion on critical percolation clusters
or related fractal media, for which the MSD again follows a subdiffusive
power law. While fractional diffusion equations may reproduce the
probability density in unbounded domains, Sokolov \cite{sokolov2012models}
emphasizes that they generally fail to capture important properties
such as first-passage times or confined-domain behavior. Consequently,
geometric subdiffusion and trapping-induced subdiffusion can yield
identical PDFs while corresponding to fundamentally different physical
processes.

Subdiffusion may also emerge in viscoelastic environments \cite{Goychuk2018,chauhan2024quantifying,goychuk2021fingerprints},
where the tagged particle is embedded in a complex interacting medium,
such as a polymer network or the cytoskeleton. In this case, the dynamics
is not governed by trapping or geometry but by long-range temporal
correlations in the particle's motion.

These systems are often described by fractional Brownian motion (fBm),
a Gaussian process characterized by correlated increments, or equivalently
by generalized Langevin equations (GLEs) with memory kernels, 
\begin{equation}
m\dot{v}(t)=-\int_{0}^{t}G(t-t')\,v(t')\,dt'+\xi(t),
\end{equation}
where the friction kernel typically follows a power law $G(t)\sim t^{-\beta}$,
and $\xi(t)$ is a correlated noise term. Depending on the noise properties,
the resulting MSD scales subdiffusively. Unlike CTRW-based models,
fBm and GLE dynamics are ergodic and do not exhibit aging, a distinction
that is crucial for interpreting experimental data.

Another way of modeling subdiffusion is via models based on diffusion
equations with time-dependent diffusion coefficients \cite{saxton2001anomalous,lim2002self},
\begin{equation}
\frac{\partial p(x,t)}{\partial t}=D(t)\,\frac{\partial^{2}p(x,t)}{\partial x^{2}},\qquad D(t)\sim t^{\alpha-1}.
\end{equation}
Although such models reproduce the same MSD scaling as subdiffusive
processes, they are primarily phenomenological fitting tools (see
\cite{sokolov2012models}). Despite yielding the same probability
density as fractional Brownian motion, they lack its correlation structure
and are more closely related to mean-field descriptions of CTRW dynamics.

In closing, real systems often exhibit subdiffusion of mixed origin,
combining trapping, geometric constraints, and viscoelastic effects.
In such cases, different models may predict the same MSD or PDF while
differing fundamentally in their aging properties, ergodicity, and
trajectory statistics \cite{sokolov2012models}. 

\subsection{Contributions of the paper}
\begin{itemize}
\item \textbf{Subdiffusion on graphs is a memory-driven process.}
Time-fractional diffusion emerges from a random time change that compresses
operational time and induces long-tailed waiting times, breaking Markovianity
while preserving linearity and mass conservation.

\item \textbf{Subdiffusion is a superposition of diffusions across multiple time scales.} By using the subordination relation (\S\ref{sec:prelim}), we approximate the Mittag-Leffler operator by a finite sum of exponential operators (\S\ref{sec:soe}). In other words, subdiffusive dynamics admit a representation as a convex and mass-preserving mixture  of classical diffusions.

\item \textbf{Sum-of-exponentials (SOE) approximation bounds.} The error of the discrete SOE approximation decays geometrically (Theorem~\ref{thm:geo}). Integral tail bounds and criteria for the selection of the quadrature window are presented in \S\ref{sec:tail-bounds}. 

\item \textbf{New subdiffusive metric on the graph.} We present a geometrization of the graph based on subdiffusive and diffusive dynamics, analyzing shortest paths between two fixed vertices for those different metrics. Our experiments demonstrate that while diffusive paths explore broader regions of the graph, the subdiffusive regime preserves path history even over extended timescales. (\S\ref{sec:shortest-paths})

\item \textbf{Subdiffusion discovers and remembers topological shortest paths.} We prove that in the small-time limit the subdiffusive shortest paths coincide exactly with the shortest paths in the original metric (Theorems~\ref{thm:Shortest-path-dominance} and~\ref{thm:Short-time-selection}). This alignment persists over longer timescales, with memory acting as an implicit reinforcement mechanism that suppresses exploration of suboptimal routes and stabilizes optimal trajectories (\S\ref{sec:shortest-paths}).

\item \textbf{Subdiffusive geometry selects paths through high-degree regions.} In the small-time limit, we prove that subdiffusive shortest paths favor those topological shortest paths that maximize the average edge degree, reversing classical diffusion's tendency to avoid hubs and revealing memory-assisted navigation (Theorem~\ref{thm:Shortest-path-dominance}).

\item \textbf{Memory is heterogeneous and vertex-dependent}. We prove that, although the system is governed by a single fractional order, different vertices (as well as the same vertex
at different times) receive larger contributions from remote past, recent past,
or present states, depending on the local temporal curvature of the
solution (\S\ref{sec:remote-early-memory}).

\item \textbf{Fractional diffusion induces structural biases absent in classical
diffusion.} These include early-time convexity at sources and concavity at neighbors,
algebraic (rather than exponential) relaxation, and degree-dependent
waiting-time effects fundamentally alter transport, trapping, and
path selection on networks (\S\ref{sec:waiting-time-errors} and \S\ref{sec:memory-emergence}).

\item \textbf{Fractional dynamics arise as a singular limit of multi-rate diffusion
with memory.} Time-fractional equations correspond to scale-free limits
of finite superposition of diffusions, equivalently describable via
operator-valued Volterra memory kernels (Proposition~\ref{lem:Caputo_from_Volterra} and Theorem~\ref{thm:SOE_to_fractional_resolvent}). This provides a common architecture linking SOE approximations, memory equations, and fractional calculus (\S\ref{sec:differential-equations}).
\end{itemize}

\section{Preliminaries }
\label{sec:prelim} 

\subsection{Notation and assumptions}

We work on $\mathbb{R}^{n}$ with the Euclidean norm $\|\cdot\|_{2}$
and induced operator norm $\|\cdot\|_{2}$. We consider $G=(V,E)$
to be a simple undirected graph on $n$ vertices with $A=(A_{ij})$ being its
adjacency matrix, and $D=\mathrm{diag}(d_{i})$ with $d_{i}=\sum_{j}A_{ij}$
the matrix of vertex degree. Let $\mathbb{C}(V)$ be the set of all
complex-valued functions on $V$ and let $\ell^{2}(V)$ be the Hilbert
space of square-summable functions on $V$ with inner product:
\[
\left\langle f,g\right\rangle =\sum_{v\in V}f\left(v\right)\overline{g\left(v\right)},\qquad f,g\in\ell^{2}\left(V\right).
\]

When the graph is finite the Laplacian operator is realized by the
so-called graph Laplacian matrix $L$. $L=D-A$ is real symmetric,
positive semidefinite, and $L\mathbf{1}=0$; if $G$ is connected
then $\ker(L)=\mathrm{span}\{\mathbf{1}\}$. The spectrum of $L$
is $0=\lambda_{1}\le\lambda_{2}\le\cdots\le\lambda_{n}=:\lambda_{\max}$;
let $L=V\Lambda V^{\top}$ be an orthogonal diagonalization. For a
bounded Borel function $f:[0,\infty)\to\mathbb{R}$ we use the spectral
calculus $f(L):=V\,f(\Lambda)\,V^{\top}$. 

If $\{(\lambda_{n},\phi_{n})\}_{n=1}^{N}$ are the eigenpairs of $L$,
then the solution of the abstract Cauchy problem of \ref{eq:Main FDE}
is given by
\begin{equation}
\begin{split}u(t) & =E_{\alpha}(-t^{\alpha}L)\,u_{0}\\
 & =\sum_{n=1}^{N}E_{\alpha}(-\lambda_{n}t^{\alpha})\,\langle u_{0},\phi_{n}\rangle\,\phi_{n},
\end{split}
\label{eq:spectral-sol}
\end{equation}
where $E_{\alpha}(M)=E_{\alpha,1}(M)$ is the Mittag-Leffler matrix
function of $M,$ which has the following power-series definition:

\begin{equation}
E_{\alpha}(M)=\sum_{k=0}^{\infty}\dfrac{M^{k}}{\Gamma\left(\alpha k+1\right)},
\end{equation}
with $\Gamma\left(\cdot\right)$ being the Euler gamma function.

\subsection{Subordination identity and the M--Wright density}

\label{sec:subordination-identity} For $0<\alpha<1$, 
\begin{equation}
E_{\alpha}\!\left(-t^{\alpha}L\right)\;=\;\int_{0}^{\infty}M_{\alpha}(\theta)\,e^{-\theta\,t^{\alpha}\,L}\,d\theta,\label{eq:subordination}
\end{equation}
where $M_{\alpha}(\theta)\ge0$ is the M--Wright (Mainardi) density,
$\int_{0}^{\infty}M_{\alpha}(\theta)\,d\theta=1$, and a convenient
series form is \cite{mainardi2010m}
\begin{equation}
M_{\alpha}(\theta)\;=\;\sum_{k=0}^{\infty}\frac{(-\theta)^{k}}{k!\,\Gamma\!\big(1-\alpha(k+1)\big)}\quad(\theta\ge0).\label{eq:Mainardi}
\end{equation}

The integral in \eqref{eq:subordination} is a \emph{Bochner integral}
of the operator-valued map $s\mapsto e^{-sL}$ against probability
measures. Formally, $\left(\int e^{-sL}\,d\mu(s)\right)u_{0}$ is
the norm-limit of Riemann sums acting on $u_{0}$.

\subsection{Complete monotonicity and Bernstein's theorem.}

The next facts make precise why the fractional propagator is a \emph{positive,
mass-preserving} contraction.
\begin{prop}
Fix $0<\alpha<1$. The scalar function $\lambda\rightarrow E_{\alpha}(-t^{\alpha}\lambda)$
is completely monotone on $[0,\infty)$. There is a probability density
$M_{\alpha}(\theta)$ on $(0,\infty)$ such that 
\[
E_{\alpha}(-t^{\alpha}\lambda)=\int_{0}^{\infty}e^{-\theta t^{\alpha}\lambda}\,M_{\alpha}(\theta)\,d\theta,
\]
and in the operator case (Bochner integral): 
\[
E_{\alpha}(-t^{\alpha}L)=\int_{0}^{\infty}e^{-\theta t^{\alpha}L}\,M_{\alpha}(\theta)\,d\theta.
\]
Furthermore, $E_{\alpha}(-t^{\alpha}L)$ is self-adjoint, positive,
and $\|E_{\alpha}(-t^{\alpha}L)\|_{2}=1$. Moreover, if $L\mathbf{1}=0$
then $E_{\alpha}(-t^{\alpha}L)\mathbf{1}=\mathbf{1}$. 
\end{prop}

\begin{proof}
The Mittag-Leffler function is completely monotone.
By Bernstein's Theorem \cite{SchillingSongVondracekBernstein}, any
completely monotone function can be represented as a mixture of exponentials:
\[
E_{\alpha}(-y)=\int_{0}^{\infty}e^{-y\theta}M_{\alpha}(\theta)\,d\theta.
\]
For the Mittag-Leffler function, the density is given by the Mainardi
function $M_{\alpha}(\theta)$. For a thorough exposition of these
results, see \cite{GorenfloKilbasMainardiRogosin,MainardiBook,SchillingSongVondracekBernstein,MainardiMuraPagniniSurvey}.
By change of variables $y=t^{\alpha}\lambda$ we get the desired scalar
identity. Functional calculus yields the operator identity: 
\[
\begin{aligned}E_{\alpha}(L) & =VE_{\alpha}(\Lambda)V^{\top}=V\,\left(\int\limits_{0}^{\infty}e^{-\theta t^{\alpha}\Lambda}M_{\alpha}(\theta)d\theta\right)\,V\\
 & =\int\limits_{0}^{\infty}\left(Ve^{-\theta t^{\alpha}\Lambda}V\right)M_{\alpha}(\theta)d\theta=\int\limits_{0}^{\infty}e^{-\theta t^{\alpha}L}M_{\alpha}(\theta)d\theta.
\end{aligned}
\]
We have positivity $E_{\alpha}(-t^{\alpha}\lambda)\geq0$ for $0<\alpha<1$.
From monotonicity we get $E_{\alpha}(-t^{\alpha}\lambda)\leq E_{\alpha}(0)=1$
for $\lambda\geq0$. Since $L$ is self-adjoint, the eigenvalues of
$L$ are non-negative and $\lambda_{1}=0$, so we have 
\[
\|E_{\alpha}(-t^{\alpha}L)\|_{2}=\max\limits_{\lambda\in\sigma(L)}|E_{\alpha}(-t^{\alpha}\lambda)|=E_{\alpha}(0)=1.
\]
The mass property follows from $e^{-\theta tL}\mathbf{1}=\mathbf{1}$
and $\int_{0}^{\infty}M_{\alpha}(\theta)d\theta=1$. 
\end{proof}

If $G$ is connected, then $\lim_{t\to\infty}E_{\alpha}(-t^{\alpha}L)u_{0}=\frac{1}{n}(\mathbf{1}^{\top}u_{0})\mathbf{1}$.
For a graph with $c>1$ components, the limit projects $u_{0}$ to
the vector that is constant on each component with the corresponding
component averages.

\subsection{Terminology: subordination and the random clock}

\label{sec:terminology}

Consider the heat semigroup on the graph $\{e^{-sL}\}_{s\ge0}$. Let
$(S_{\tau})_{\tau\ge0}$ be the $\alpha$-stable subordinator and
$E_{t}=\inf\{\tau>0:\,S_{\tau}>t\}$ the inverse subordinator. Here,
\emph{subordination} means that the time-fractional evolution is obtained
by running the baseline (time-linear) heat dynamics with the \emph{random
clock} $E_{t}$ in place of deterministic time, 
\begin{equation}
E_{\alpha}(-t^{\alpha}L)=\int_{0}^{\infty}e^{-sL}\,g_{\alpha}(s,t)\,ds=\int_{0}^{\infty}M_{\alpha}(\theta)\,e^{-\theta\,t\,L}\,d\theta,
\label{eq:subordination-prelim}
\end{equation}
with $g_{\alpha}(s,t)=t^{-\alpha}M_{\alpha}(s/t^{\alpha})$; see \cite{MeerschaertSikorskii,ApplebaumLevy,MainardiMuraPagniniSurvey}.
We do \emph{not} consider Bochner subordination by a general Bernstein
function $\phi$ (which would lead to space-fractional $\phi(L)$).
Equation \eqref{eq:subordination-prelim} means that the operator $E_{\alpha}(-t^{\alpha}L)$ is the expect value of the standard diffusion operator $e^{-sL}$, where the $g_{\alpha}(s,t)$ is the probability that the random clock displays time $s$ when the physical time that has passed is $t$. 

\subsection{The time--changed process $X_{E_{t}}$}

\label{sec:time-changed} Let $\{X_{s}\}_{s\ge0}$ be the continuous--time
Markov process on $V$ with generator $-L$: 
\[
\frac{d}{ds}\,p(s)=-L\,p(s),\qquad p(s)=\big(\mathbb{P}[X_{s}=j\,|\,X_{0}=i]\big)_{i,j}=e^{-sL}.
\]
From node $i$ the holding time is an exponential random variable with decay rate $d_{i}=\sum_{j}A_{ij}$. The next node is chosen to be $j$ with probability $A_{ij}/d_{i}$.

Let $E_{t}:=\inf\{\tau>0:\,S_{\tau}>t\}$ be the inverse $\alpha$--stable
subordinator ($0<\alpha<1$), independent of $X$. The \emph{time--changed
process} is $Y_{t}:=X_{E_{t}}$, i.e., the baseline diffusion observed
at the random operational time $E_{t}$.

Averaging over the random clock yields 
\[
\mathbb{P}[Y_{t}=j\,|\,Y_{0}=i]=\int_{0}^{\infty}\mathbb{P}[X_{s}=j\,|\,X_{0}=i]\;g_{\alpha}(s,t)\,ds=\big[E_{\alpha}(-t^{\alpha}L)\big]_{ij}.
\]
Hence for any initial $u_{0}$, $u(t)=E_{\alpha}(-t^{\alpha}L)\,u_{0}$
solves $\partial_{t}^{\alpha}u(t)=-Lu(t)$ (Caputo).

$Y_{t}$ is semi--Markov in physical time $t$: the holding time $T_i$
at node $i$ has survival $\mathbb{P}(T_i>t)=E_{\alpha}(-d_{i}t^{\alpha})$,
a heavy tail ($\sim t^{-\alpha}$), producing trapping and memory.

Since $L\mathbf{1}=0$ and the integral representation averages stochastic
kernels, $Y_{t}$ preserves total mass. On a connected graph, $\lim_{t\to\infty}Y_{t}$
has the same equilibrium as the baseline walk (componentwise averaging).

\section{Sum of Exponentials approximation}
\label{sec:soe}

Here we develop a practical, mass-preserving approximation
of the Mittag-Leffler function as a sum of exponential functions.
We start with the following fundamental definition.
\begin{defn}
The \emph{sum-of-exponentials (SOE)} approximation is 
\begin{equation}
F_{J}(t,L)=\sum_{j=1}^{J}a_{j}\,e^{-b_{j}t^{\alpha}L},
\end{equation}
with coefficients $a_{j},b_{j}$ that depend only on $\alpha$ and
not on $t$. 
\end{defn}

Our aim is to have 
\[
E_{\alpha}(-t^{\alpha}L)\;\approx\;\sum_{j=1}^{J}a_{j}\,e^{-b_{j}t^{\alpha}L},
\]
This approximation is constructed as a quadrature of \eqref{eq:subordination}
whose window is logarithmically scaled in $\theta$, furthermore we: 
\begin{enumerate}
\item give a log--trapezoidal construction with $a_{j}>0$, $\sum_{j}a_{j}=1$
(mass conservation); 
\item prove geometric convergence in $J$ for a fixed window; 
\item derive explicit, uniform tail bounds from $M_{\alpha}$ to select
$(\theta_{\min},\theta_{\max})$; 
\item tabulate ready-to-use window endpoints $(\theta_{\min},\theta_{\max})$
(and $(y_{\min},y_{\max})$) for typical $\alpha,\varepsilon$; 
\item provide practical error metrics (relative/probe, mass error) for a
posteriori assessment. 
\end{enumerate}

\subsection{SOE via log--trapezoidal quadrature}
\label{sec:soe-quadrature} Set $\theta=e^{y}$, $y\in\mathbb{R}$, and truncate
to $[y_{\min},y_{\max}]$. With a uniform grid $y_{j}=y_{\min}+(j-1)h$,
$h=\frac{y_{\max}-y_{\min}}{J-1}$, define 
\begin{equation}
b_{j}=e^{y_{j}},\qquad\tilde{a}_{j}=h\,M_{\alpha}(b_{j})\,b_{j},\qquad a_{j}=\frac{\tilde{a}_{j}}{\sum_{k=1}^{J}\tilde{a}_{k}}.\label{eq:ab-def}
\end{equation}
Then 
\begin{equation}
E_{\alpha}\!\left(-t^{\alpha}L\right)\;\approx\;\sum_{j=1}^{J}a_{j}\,e^{-\,b_{j}\,t^{\alpha}\,L},\qquad a_{j}>0,\;b_{j}>0,\;\sum_{j=1}^{J}a_{j}=1.\label{eq:SOE}
\end{equation}

\noindent\emph{Convention 1.} In all that follows we fix the nodes
$b_{j}$ once (they depend only on $\alpha$ and the chosen log--window)
and, for any $t>0$, we evaluate the SOE as $\sum_{j=1}^{J}a_{j}\,e^{-(t^{\alpha}\,b_{j})\,L}$.
Thus the coefficients $(a_{j},b_{j})$ do not depend on $t$ or on
$L$. We then have the following results.
\begin{prop}
\label{prop:mass} If $L\mathbf{1}=0$ and $\sum_{j}a_{j}=1$, then
for all $t\ge0$, $\mathbf{1}^{\top}\big(\sum_{j=1}^{J}a_{j}e^{-b_{j}tL}\big)u_{0}=\mathbf{1}^{\top}u_{0}.$ 
\end{prop}

\begin{proof}
Since $L\mathbf{1}=0$, $e^{-b_{j}t^{\alpha}L}\mathbf{1}=\mathbf{1}$
for each $j$. Thus $\big(\sum_{j}a_{j}e^{-b_{j}t^{\alpha}L}\big)\mathbf{1}=(\sum_{j}a_{j})\,\mathbf{1}=\mathbf{1}$,
and left-multiplying by $\mathbf{1}^{\top}$ gives $\mathbf{1}^{\top}(\cdot)u_{0}=\mathbf{1}^{\top}u_{0}$. 
\end{proof}
\begin{lem}
\label{lem:op-scalar} For self-adjoint $L$ and bounded Borel $f,g$,
$\|f(L)-g(L)\|_{2}=\max_{\lambda\in\sigma(L)}|f(\lambda)-g(\lambda)|.$
In particular, the SOE error is 
\[
\|E_{\alpha}(-t^{\alpha}L)-F_{J}(t,L)\|_{2}=\max_{\lambda\in[0,\lambda_{\max}]}\big|g(\lambda)-g_{J}(\lambda)\big|,
\]
with $g(\lambda)=E_{\alpha}(-t^{\alpha}\lambda)$ and $g_{J}(\lambda)=\sum_{j}a_{j}e^{-b_{j}t^{\alpha}\lambda}$. 
\end{lem}

\begin{proof}
Diagonalize $L=V\Lambda V^{\top}$. Then $g(L)-g_{J}(L)=V\,(g(\Lambda)-g_{J}(\Lambda))\,V^{\top}$
and $\|g(L)-g_{J}(L)\|_{2}=\|g(\Lambda)-g_{J}(\Lambda)\|_{2}=\max_{k}|g(\lambda_{k})-g_{J}(\lambda_{k})|.$ 
\end{proof}
\begin{thm}
\label{thm:geo} Let $f_{\lambda}(y):=M_{\alpha}(e^{y})\,e^{y}\,e^{-e^{y}t^{\alpha}\lambda}$
with $t>0$, $\lambda\in[0,\lambda_{\max}]$. For fixed $[y_{\min},y_{\max}]$
and step $h=(y_{\max}-y_{\min})/(J-1)$, 
\[
\sup_{\lambda\in[0,\lambda_{\max}]}\left|\int_{y_{\min}}^{y_{\max}}f_{\lambda}(y)\,dy-h\sum_{j=1}^{J}f_{\lambda}(y_{j})\right|\;\le\;C_{1}\,e^{-C_{2}/h},
\]
with constants $C_{1},C_{2}>0$ independent of $h$ and $\lambda$.
Consequently, 
\[
\|E_{\alpha}(-t^{\alpha}L)-\sum_{j=1}^{J}a_{j}e^{-b_{j}t^{\alpha}L}\|_{2}\;\le\;C_{1}\,e^{-C_{2}(J-1)/(y_{\max}-y_{\min})}.
\]
\end{thm}

\begin{proof}
(\emph{Analytic trapezoidal rule}) The map $y\mapsto f_{\lambda}(y)$
is analytic in a strip $S_{a}=\{|\mathrm{Im}\,y|<a\}$ and decays
rapidly as $\mathrm{Re}\,y\to\pm\infty$ (doubly-exponential decay
for $+\infty$, integrable for $-\infty$). For analytic, rapidly
decaying functions, the (periodized) trapezoidal rule error on a finite
interval is $\le Ce^{-2\pi a/h}$ via Poisson summation; see \cite[Thm.~2.1]{TrefethenWeideman2014}.
Uniform strip bounds in $\lambda\in[0,\lambda_{\max}]$ give uniform
$C_{1},C_{2}$. Lemma~\ref{lem:op-scalar} transfers the scalar quadrature
error to the operator norm bound. 
\end{proof}

\subsection{SOE discretization via subordination: quadrature weights, mass preservation,
and vertex survivals}

We sample uniformly in $y=\log\theta$ on $[y_{\min},y_{\max}]$ with
step 
\[
h_{y}=\frac{y_{\max}-y_{\min}}{J-1},\qquad y_{j}=y_{\min}+(j-1)h_{y},\qquad b_{j}=\theta_{j}=e^{y_{j}}.
\]
The \emph{raw} (log--trapezoid) quadrature weights for the subordination
integral are 
\begin{equation}
w_{j}^{\mathrm{raw}}\;=\;h_{y}\,M_{\alpha}(b_{j})\,b_{j}\qquad(j=1,\dots,J).\label{eq:wraw}
\end{equation}
They approximate $\int M_{\alpha}(\theta)\,(\cdot)\,d\theta$ on the
chosen window, with total mass 
\[
\mathrm{mass}_{\mathrm{win}}\;:=\;\sum_{j=1}^{J}w_{j}^{\mathrm{raw}}\;=\;\int_{\theta_{\min}}^{\theta_{\max}}M_{\alpha}(\theta)\,d\theta\;\le\;1.
\]
For the \emph{operator} we enforce mass preservation of the zero mode
by normalizing the raw weights, 
\begin{equation}
a_{j}\;=\;\frac{w_{j}^{\mathrm{raw}}}{\sum_{k=1}^{J}w_{k}^{\mathrm{raw}}},\qquad\sum_{j=1}^{J}a_{j}=1,\label{eq:a-from-wraw}
\end{equation}
and keep time only inside the exponentials: 
\begin{equation}
E_{\alpha}(-t^{\alpha}L)\;\approx\;\sum_{j=1}^{J}a_{j}\,e^{-(t^{\alpha}\,b_{j})\,L}.\label{eq:SOE-conv1}
\end{equation}
Here the nodes $b_{j}$ (internal times $\theta_{j}$) and the weights
depend on $\alpha$ and on the chosen window but are \emph{independent}
of $t$ and of $L$ (Convention~1). 
\begin{rem}
The normalization \eqref{eq:a-from-wraw} is essential to preserve
mass in the \emph{operator}-level approximation \eqref{eq:SOE-conv1}.
For scalar Laplace integrals against $M_{\alpha}$ one must use the
\emph{raw} weights. In particular, the vertex waiting-time survival
at node $i$, with degree $d_{i}=L_{ii}$, 
\[
S_{i}(t)\;=\;E_{\alpha}(-d_{i}t^{\alpha})\;=\;\int_{0}^{\infty}M_{\alpha}(\theta)\,e^{-d_{i}t^{\alpha}\,\theta}\,d\theta,
\]
is approximated by 
\begin{equation}
S_{i}(t)\;\approx\;\sum_{j=1}^{J}w_{j}^{\mathrm{raw}}\,e^{-(d_{i}t^{\alpha})\,b_{j}}.\label{eq:survival-raw}
\end{equation}
Equivalently, if only the normalized $a_{j}$ are available, 
\[
S_{i}(t)\;\approx\;\mathrm{mass}_{\mathrm{win}}\,\sum_{j=1}^{J}a_{j}\,e^{-(d_{i}t^{\alpha})\,b_{j}},\qquad\mathrm{mass}_{\mathrm{win}}=\sum_{j=1}^{J}w_{j}^{\mathrm{raw}}.
\]
\end{rem}

\medskip{}

\begin{rem}
Under subordination, $\theta$ is the operational (internal) time
of the heat semigroup. The SOE nodes $b_{j}=\theta_{j}$ are \emph{global
internal-time samples}: the same set $\{b_{j}\}_{j=1}^{J}$ simultaneously
approximates the spectral map $\lambda\mapsto E_{\alpha}(-(t^{\alpha})\lambda)$
(operator level) and all vertex survivals $S_{i}(t)$ via \eqref{eq:survival-raw}.
Accuracy is governed by the window $[\theta_{\min},\theta_{\max}]$:
in practice one chooses it using tail bounds so that the peak of $M_{\alpha}(\theta)\,e^{-q\theta}$
at $\theta\sim1/q$ with $q=d_{i}t^{\alpha}$ lies well inside the
window for all relevant degrees $d_{i}$ and times $t$. 
\end{rem}

\section{Tail bounds for the subordination integral and window selection}
\label{sec:tail-bounds}

Here we start by considering the subordination identity previously
defined and its scalar counterpart $F(\lambda)=\int_{0}^{\infty}M_{\alpha}(\theta)e^{-\theta t^{\alpha}\lambda}\,d\theta$
for $\lambda\ge0$. For window endpoints $0<\theta_{\min}<\theta_{\max}<\infty$
(equivalently, $y_{\min}=\log\theta_{\min}$ and $y_{\max}=\log\theta_{\max}$),
we define the truncated operator 
\[
F_{\mathrm{win}}(L)\;=\;\int_{\theta_{\min}}^{\theta_{\max}}M_{\alpha}(\theta)\,e^{-\theta t^{\alpha}L}\,d\theta,
\]
so that the truncation error splits as the \emph{left} and \emph{right}
tails: 
\[
\mathcal{T}_{L}(\lambda):=\int_{0}^{\theta_{\min}}M_{\alpha}(\theta)\,e^{-\theta t^{\alpha}\lambda}\,d\theta,\qquad\mathcal{T}_{R}(\lambda):=\int_{\theta_{\max}}^{\infty}M_{\alpha}(\theta)\,e^{-\theta t^{\alpha}\lambda}\,d\theta.
\]
By the spectral theorem (Lemma~\ref{lem:op-scalar}), the operator
truncation error satisfies 
\[
\|F(L)-F_{\mathrm{win}}(L)\|_{2}\;=\;\max_{\lambda\in\sigma(L)}\big(\mathcal{T}_{L}(\lambda)+\mathcal{T}_{R}(\lambda)\big)\;\le\;\sup_{\lambda\in[0,\lambda_{\max}]}\mathcal{T}_{L}(\lambda)\;+\;\sup_{\lambda\in[0,\lambda_{\max}]}\mathcal{T}_{R}(\lambda).
\]

\subsection{Left tail (small $\theta$)}

Near $\theta=0$, the M--Wright density satisfies $M_{\alpha}(\theta)=\frac{1}{\Gamma(1-\alpha)}+O(\theta)$
\cite{MainardiBook,GorenfloKilbasMainardiRogosin}. A simple bound
follows. 
\begin{prop}
\label{prop:left-tail} For all $t\ge0$, $\lambda\ge0$ and $0<\theta_{\min}\le1$,
\begin{equation}
\mathcal{T}_{L}(\lambda)\;\le\;\frac{1}{\Gamma(1-\alpha)}\begin{cases}
\dfrac{1-e^{-t^{\alpha}\lambda\theta_{\min}}}{t^{\alpha}\lambda}, & \lambda>0,\\[1ex]
\theta_{\min}, & \lambda=0.
\end{cases}\label{eq:left-tail}
\end{equation}
Consequently, uniformly in $\lambda\in[0,\lambda_{\max}]$, 
\[
\sup_{\lambda\in[0,\lambda_{\max}]}\mathcal{T}_{L}(\lambda)\;\le\;\frac{\theta_{\min}}{\Gamma(1-\alpha)}.
\]
\end{prop}

\begin{proof}
Use $M_{\alpha}(\theta)\le\frac{1}{\Gamma(1-\alpha)}$ for $\theta\in(0,1]$
and integrate $\int_{0}^{\theta_{\min}}e^{-t\lambda\theta}\,d\theta$
(with the obvious limit as $\lambda\downarrow0$). 
\end{proof}
We should remark that taking $\theta_{\min}=\varepsilon\,\Gamma(1-\alpha)$
forces the left-tail contribution below $\varepsilon$ for \emph{all}
eigenvalues (including the zero mode).

\subsection{Right tail (large $\theta$)}

As $\theta\to\infty$ the density has a \emph{stretched-exponential}
decay \cite[Ch.~4]{MainardiBook}, \cite[Ch.~2]{GorenfloKilbasMainardiRogosin}:
there exist $C_{\alpha},c_{\alpha}>0$ and $q_{\alpha}:=\frac{1}{1-\alpha}>1$
such that, for all sufficiently large $\theta$, 
\begin{equation}
M_{\alpha}(\theta)\;\le\;C_{\alpha}\,\theta^{p_{\alpha}}\,\exp\!\big(-c_{\alpha}\,\theta^{\,q_{\alpha}}\big),\qquad p_{\alpha}:=\frac{\alpha-2}{2(1-\alpha)}.\label{eq:Mwright-tail}
\end{equation}
This yields two complementary bounds. 
\begin{prop}
\label{prop:right-tail-general} Let \eqref{eq:Mwright-tail} hold
for all $\theta\ge\theta_{0}(\alpha)$. Then for any $t\ge0$, $\lambda\ge0$
and $\theta_{\max}\ge\theta_{0}(\alpha)$, 
\begin{equation}
\mathcal{T}_{R}(\lambda)\;\le\;\int_{\theta_{\max}}^{\infty}C_{\alpha}\theta^{p_{\alpha}}e^{-c_{\alpha}\theta^{q_{\alpha}}}\,d\theta\;\le\;\frac{C_{\alpha}}{q_{\alpha}c_{\alpha}}\;\theta_{\max}^{\,p_{\alpha}+1-q_{\alpha}}\;\exp\!\big(-c_{\alpha}\,\theta_{\max}^{q_{\alpha}}\big).\label{eq:right-tail-general}
\end{equation}
In particular, the bound is \emph{independent of $\lambda$} (and
thus controls the zero eigenvalue). 
\end{prop}

\begin{proof}
Drop the factor $e^{-t\lambda\theta}\le1$ and bound the tail of a
monotone stretched-exponential via the standard inequality $\int_{x}^{\infty}u^{p}e^{-cu^{q}}du\le\frac{1}{qc}\,x^{p+1-q}e^{-cx^{q}}$
for $x$ large. 
\end{proof}
\begin{prop}
\label{prop:right-tail-gap} If the graph is connected and we restrict
to $\mathbf{1}^{\perp}$ (thus $\lambda\ge\lambda_{2}>0$), then for
any $\theta_{\max}>0$, 
\begin{equation}
\sup_{\lambda\in[\lambda_{2},\lambda_{\max}]}\mathcal{T}_{R}(\lambda)\;\le\;e^{-t^{\alpha}\lambda_{2}\theta_{\max}}\int_{\theta_{\max}}^{\infty}M_{\alpha}(\theta)\,d\theta\;\le\;e^{-t^{\alpha}\lambda_{2}\theta_{\max}}.\label{eq:right-tail-gap}
\end{equation}
\end{prop}

\begin{proof}
Use $e^{-t\lambda\theta}\le e^{-t\lambda_{2}\theta}$ and $\int_{\theta_{\max}}^{\infty}M_{\alpha}(\theta)\,d\theta\le1$. 
\end{proof}
We now consider the practical choices based on the previous results: 
\begin{itemize}
\item \emph{All modes (including $\lambda=0$):} pick $\theta_{\max}$ so
that $c_{\alpha}\theta_{\max}^{q_{\alpha}}\ge\log(2/\varepsilon)$;
then by \eqref{eq:right-tail-general}, $\sup_{\lambda}\mathcal{T}_{R}(\lambda)\lesssim\tfrac{C_{\alpha}}{q_{\alpha}c_{\alpha}}\,\theta_{\max}^{p_{\alpha}+1-q_{\alpha}}\,e^{-c_{\alpha}\theta_{\max}^{q_{\alpha}}}\le\varepsilon/2$
for large enough $\theta_{\max}$. 
\item \emph{Mean-zero subspace:} set $\theta_{\max}\ge\frac{1}{t^{\alpha}\lambda_{2}}\log(2/\varepsilon)$
to guarantee $\sup_{\lambda\ge\lambda_{2}}\mathcal{T}_{R}(\lambda)\le\varepsilon/2$
by \eqref{eq:right-tail-gap}. 
\end{itemize}

\subsection{Putting it together: window rules}

Combining Propositions~\ref{prop:left-tail}--\ref{prop:right-tail-gap}
yields explicit choices for $(\theta_{\min},\theta_{\max})$ ensuring
a total tail below a target tolerance $\varepsilon$: 
\begin{cor}
\label{cor:window} Given $\varepsilon\in(0,1)$ and $0<\alpha<1$:

\textbf{General (all modes).} Choose 
\[
\theta_{\min}=\frac{\varepsilon}{2}\,\Gamma(1-\alpha),\qquad\theta_{\max}\ \text{s.t.}\quad c_{\alpha}\theta_{\max}^{q_{\alpha}}\ge\log\!\frac{2}{\varepsilon}.
\]
Then $\sup_{\lambda\in[0,\lambda_{\max}]}\big(\mathcal{T}_{L}(\lambda)+\mathcal{T}_{R}(\lambda)\big)\le\varepsilon$.

\textbf{Mean-zero subspace (connected graph).} If one estimates the
operator on $\mathbf{1}^{\perp}$, take 
\[
\theta_{\min}=\frac{\varepsilon}{2}\,\Gamma(1-\alpha),\qquad\theta_{\max}=\frac{1}{t\lambda_{2}}\,\log\!\frac{2}{\varepsilon}.
\]
\medspace{}\emph{If one intends to reuse the same nodes for all}
$t\in[t_{\min},t_{\max}]$ , replace $t$ by $t_{\min}$ in this formula
so that $(a_{j},b_{j})$ are $t$-independent. Then $\sup_{\lambda\in[\lambda_{2},\lambda_{\max}]}\big(\mathcal{T}_{L}(\lambda)+\mathcal{T}_{R}(\lambda)\big)\le\varepsilon$. 
\end{cor}

\begin{rem}
The log--trapezoidal SOE uses $y=\log\theta$ on $[y_{\min},y_{\max}]$
with $y_{\min}=\log\theta_{\min}$ and $y_{\max}=\log\theta_{\max}$.
After fixing the window by Corollary~\ref{cor:window}, increasing
$J$ then controls the \emph{discretization} error inside the window
with geometric rate (Theorem~\ref{thm:geo}). 
\end{rem}

\begin{rem}
The choice $\theta_{\max}$ via Proposition~\ref{prop:right-tail-gap}
with $\lambda_{2}$ replaced by $\lambda_{\max}$ gives a conservative
upper bound on the tail for \emph{high-frequency} modes and motivates
rules of the form $\theta_{\max}\sim C/(t\lambda)$. For reuse across
multiple times, take $t=t_{\min}$ of the range. Empirically, taking
$C\approx32$ makes $\exp(-\theta_{\max}t\lambda)$ fall near machine
precision $10^{-14}$ and works well when the focus is on modes away
from the zero eigenvalue; the rigorous alternative \eqref{eq:right-tail-general}
controls \emph{all} modes using only properties of $M_{\alpha}$. 
\end{rem}

\subsection{Constants in the stretched-exponential bound.}

Explicit asymptotics for $M_{\alpha}$ give $q_{\alpha}=\frac{1}{1-\alpha}$
and an exponent constant $c_{\alpha}=(1-\alpha)\,\alpha^{\alpha/(1-\alpha)}$.
The prefactor has a power $\theta^{p_{\alpha}}$ with $p_{\alpha}=\frac{\alpha-2}{2(1-\alpha)}$
and a multiplicative constant depending only on $\alpha$ (see \cite[Sec.~4.3]{MainardiBook}
and \cite[Sec.~2.3]{GorenfloKilbasMainardiRogosin}). Using these
in \eqref{eq:right-tail-general} yields a fully explicit $\theta_{\max}(\varepsilon,\alpha)$.

\section{Vertex waiting times and error metrics.}
\label{sec:waiting-time-errors}

In \emph{time-fractional} diffusion, physical time $t$ is linked
to the baseline (operational) time $s$ by a \emph{random clock} $E_{t}$
(the inverse $\alpha$-stable subordinator, $0<\alpha<1$). Its density
satisfies the scaling 
\begin{equation}
g_{\alpha}(s,t)\;=\;t^{-\alpha}\,M_{\alpha}\!\Big(\frac{s}{t^{\alpha}}\Big),\qquad s>0,\ t>0,\label{eq:clock-scaling}
\end{equation}
and one has the identical subordination forms 
\begin{equation}
E_{\alpha}(-t^{\alpha}L)\;=\;\int_{0}^{\infty}e^{-sL}\,g_{\alpha}(s,t)\,ds\;=\;\int_{0}^{\infty}M_{\alpha}(\theta)\,e^{-\theta\,t^{\alpha}\,L}\,d\theta.
\end{equation}
Let $T_{i}$ denote the (physical-time) waiting time before the next
jump when the process is at node $i$ at time $0$. Conditioning on
the random clock and averaging yields the \emph{survival function}
\begin{equation}
\underbrace{S_{i}(t)}_{\text{survival}}\;:=\;\mathbb{P}(T_{i}>t)\;=\;\int_{0}^{\infty}e^{-d_{i}s}\,g_{\alpha}(s,t)\,ds\;=\;\int_{0}^{\infty}M_{\alpha}(\theta)\,e^{-d_{i}t^{\alpha}\theta}\,d\theta\;=\;E_{\alpha}(-d_{i}\,t^{\alpha}),\qquad t\ge0,\label{eq:survival-exact}
\end{equation}
so each vertex exhibits a \emph{Mittag-Leffler} waiting-time law
with heavy tail $S_{i}(t)\sim\big(d_{i}\,\Gamma(1-\alpha)\big)^{-1}t^{-\alpha}$.

\subsection{Raw vs.\ normalized quadrature weights.} For numerical quadrature
on $\theta=e^{y}$ we use a uniform grid $y_{j}=y_{\min}+(j-1)h_{y}$
with $h_{y}=(y_{\max}-y_{\min})/(J-1)$ and set $b_{j}=\theta_{j}=e^{y_{j}}$.
The \emph{raw} log--trapezoid weights are 
\[
w_{j}^{\mathrm{raw}}\;=\;h_{y}\,M_{\alpha}(b_{j})\,b_{j},\qquad j=1,\dots,J,
\]
which approximate the measure $M_{\alpha}(\theta)\,d\theta$ on $[\theta_{\min},\theta_{\max}]$.
Their sum $\mathrm{mass}_{\mathrm{win}}:=\sum_{j}w_{j}^{\mathrm{raw}}=\int_{\theta_{\min}}^{\theta_{\max}}M_{\alpha}(\theta)\,d\theta\le1$
is the in-window mass. For the \emph{operator} we normalize to $a_{j}=w_{j}^{\mathrm{raw}}/\sum_{k}w_{k}^{\mathrm{raw}}$
so that $\sum_{j}a_{j}=1$ and obtain 
\[
E_{\alpha}(-t^{\alpha}L)\;\approx\;\sum_{j=1}^{J}a_{j}\,e^{-(t^{\alpha}\,b_{j})L}.
\]
For \emph{scalar} quantities like vertex survivals, probability density
functions and hazards, one must use the \emph{raw} weights directly
(equivalently, multiply the normalized sums by $\mathrm{mass}_{\mathrm{win}}$).

\subsection{Mixture interpretation for the random clock.} With the raw weights,
\[
g_{\alpha}(s,t)\,ds\;\approx\;\sum_{j=1}^{J}w_{j}^{\mathrm{raw}}\;\delta\big(s-t\,b_{j}\big)\,ds,
\]
so that, simultaneously for \emph{every} vertex $i$, 
\begin{align}
S_{i}(t) & \approx\sum_{j=1}^{J}w_{j}^{\mathrm{raw}}\,e^{-\,d_{i}\,t\,b_{j}},\label{eq:survival-soe}\\[0.25em]
f_{i}(t) & :=-\frac{d}{dt}S_{i}(t)\;\approx\;d_{i}\sum_{j=1}^{J}w_{j}^{\mathrm{raw}}\,b_{j}\,e^{-\,d_{i}\,t\,b_{j}},\label{eq:pdf-soe}\\[0.25em]
h_{i}(t) & :=\frac{f_{i}(t)}{S_{i}(t)}\;\approx\;d_{i}\,\frac{\sum_{j=1}^{J}w_{j}^{\mathrm{raw}}\,b_{j}\,e^{-\,d_{i}\,t\,b_{j}}}{\sum_{j=1}^{J}w_{j}^{\mathrm{raw}}\,e^{-\,d_{i}\,t\,b_{j}}}\;=\;d_{i}\,\frac{\sum_{j}a_{j}\,b_{j}\,e^{-d_{i}tb_{j}}}{\sum_{j}a_{j}\,e^{-d_{i}tb_{j}}}.\label{eq:hazard-soe}
\end{align}

Then the coefficients $w_{j}^{\mathrm{raw}}$ are the actual ``importance''
of every exponential term in defining the waiting time at every vertex
of the graph.

The SOE internal times $s_{j}(t)=t\,b_{j}$ are \emph{global}
operational-time samples drawn from the inverse-stable clock; they
are not per-vertex waits. Vertex-dependence enters only through the
degree $d_{i}$ (the baseline attempt rate). Mixing the global internal-time
samples with the node-dependent rates $d_{i}$ yields Mittag-Leffler
waiting times at each vertex. In this precise sense, the $\{b_{j}\}$
are the global internal-time samples whose mixtures simultaneously
approximate all vertex waiting-time distributions.

Because \eqref{eq:survival-soe} is a convex combination of decaying
exponentials in $t$, the approximate hazard $h_{i}(t)$ in \eqref{eq:hazard-soe}
is a \emph{decreasing} function of $t$ (``aging''), matching the
time-fractional renewal interpretation. Also, from $S_{i}(t)=E_{\alpha}(-d_{i}t^{\alpha})$
one has $S_{i}(t)\sim\big(d_{i}\,\Gamma(1-\alpha)\big)^{-1}\,t^{-\alpha}$
as $t\to\infty$, so the mean waiting time is infinite for every $i$
when $0<\alpha<1$.

\subsection{Error metrics}

Here we define some error metrics to check the accuracy of the previous
approximations. Given a probe vector $u_{0}$ and the SOE approximation
\[
u^{\mathrm{SOE}}(t)\;=\;F_{J}(t,L)\,u_{0}\;=\;\sum_{j=1}^{J}a_{j}\,e^{-b_{j}t^{\alpha}L}\,u_{0},
\]
we measure the error in the following ways.

\subsubsection{Relative error (per probe).}

\begin{equation}
\mathrm{relerr}(u_{0})\;=\;\frac{\big\|\,u^{\star}(t)-u^{\mathrm{SOE}}(t)\,\big\|_{2}}{\big\|\,u^{\star}(t)\,\big\|_{2}}.\label{eq:relerr}
\end{equation}
When multiple probes $\{u_{0}^{(k)}\}_{k=1}^{r}$ are used (e.g.,
random), we report $\max_{k}\mathrm{relerr}(u_{0}^{(k)})$.

\subsubsection{Mass conservation error.}

\begin{equation}
\mathrm{masserr}(u_{0})\;=\;\left|\,\mathbf{1}^{\top}u^{\star}(t)-\mathbf{1}^{\top}u^{\mathrm{SOE}}(t)\,\right|.\label{eq:masserr}
\end{equation}
Ideally $\mathrm{masserr}(u_{0})=0$; with floating-point arithmetic
it is typically at the level of machine precision due to $\sum_{j}a_{j}=1$
and $L\mathbf{1}=0$.

\subsubsection{Scalar  error.}
The scalar error, intended on the spectrum of the matrix-functions, is
\begin{equation}
 \max\limits_{\lambda\in[0,\lambda_{\max}]} \left| E_{\alpha}(-t^{\alpha}\lambda) - F_{J}(t,\lambda)\right|.
\end{equation}
By Theorem~\ref{thm:geo} and  uniformly for $\lambda\in[0,\lambda_{\max}]$, 
the scalar error decays \emph{geometrically} with $J$ once the window captures the effective support of the integrand;

The raw error as a function of $J$ need not be strictly decreasing
(typical oscillations of spectrally convergent trapezoidal rules),
but its envelope decays until it reaches the desired accuracy \cite{TrefethenWeideman2014}.

In all cases the SOE inherits the correct long-time limit and conservation
properties by construction (Proposition~\ref{prop:mass}). Accuracy
improves with $J$ and with a window adapted to the relevant scale
$t, \lambda$.

\subsubsection{Operator error.}
If an operator-level diagnostic is desired, one may estimate
\begin{equation} 
\frac{\big\|\,E_{\alpha}(-t^{\alpha}L)-F_{J}(t,L)\,\big\|_{2}}{\big\|\,E_{\alpha}(-t^{\alpha}L)\,\big\|_{2}}
\end{equation}
via power iteration on the difference operator; in practice, the probe-based
relative error \eqref{eq:relerr} is sufficient and cheaper.

Note that in exact arithmetic the operator error and the scalar spectral error are the same, since $\big\|\,E_{\alpha}(-t^{\alpha}L)\,\big\|_{2}=1$ and the 2-norm of the symmetric matrix $E_{\alpha}(-t^{\alpha}L)-F_{J}(t,L)$ is the magnitude of its largest eigenvalue, which is precisely the scalar error by the spectral theorem applied to the function $g(\lambda) =E_{\alpha}(-t^{\alpha}\lambda) - F_{J}(t,\lambda)$. However the computation of matrices in floating-point arithmetic may introduce numerical inaccuracies. This can be seen by comparing Fig.~\ref{fig:errors_ER} and Fig.~\ref{fig:error_probing}(B).

\subsection{Computational results}

In this section we consider an Erd\H{o}s--R\'{e}nyi (ER) random graph
with $250$ vertices and $1000$ edges. The graph is simple and connected.
We then compute the scalar error for three different values of $\alpha$
and a wide range of times $11\leq t\leq1001$. The results are illustrated
in Fig.~\ref{fig:errors_ER} for $\alpha=0.8$ (a), $\alpha=0.5$ (b),
and $\alpha=0.25$ (c). The convergence speed of the scalar error is
larger for smaller values of $\alpha$, that is, fewer addends are
needed to obtain a good approximation in the regime of stronger memory
effects.

\begin{figure}
\subfloat[$\alpha=0.8$]{\begin{centering}
\includegraphics[width=0.32\textwidth]{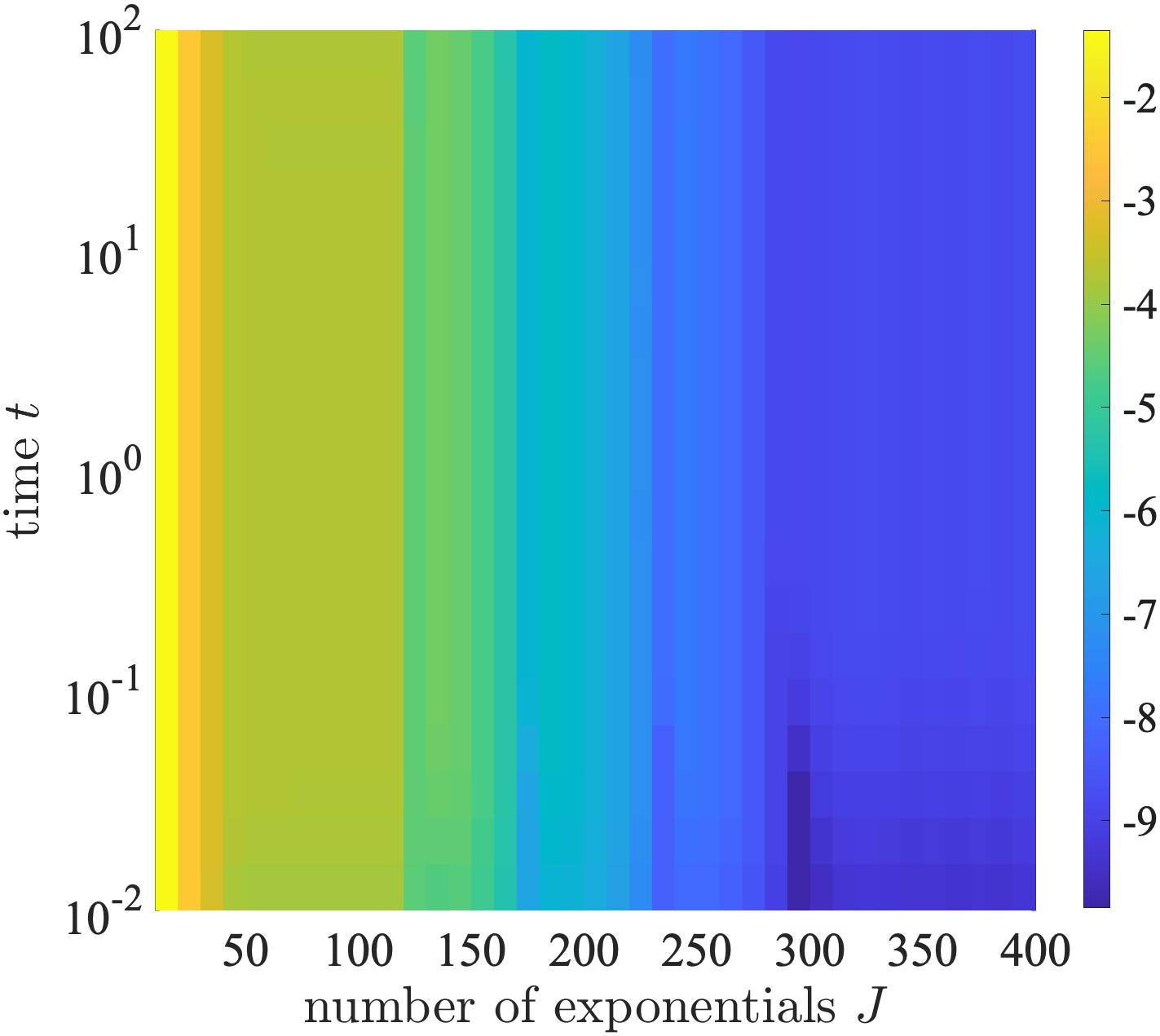}
\par\end{centering}
}\subfloat[$\alpha=0.5$]{\begin{centering}
\includegraphics[width=0.32\textwidth]{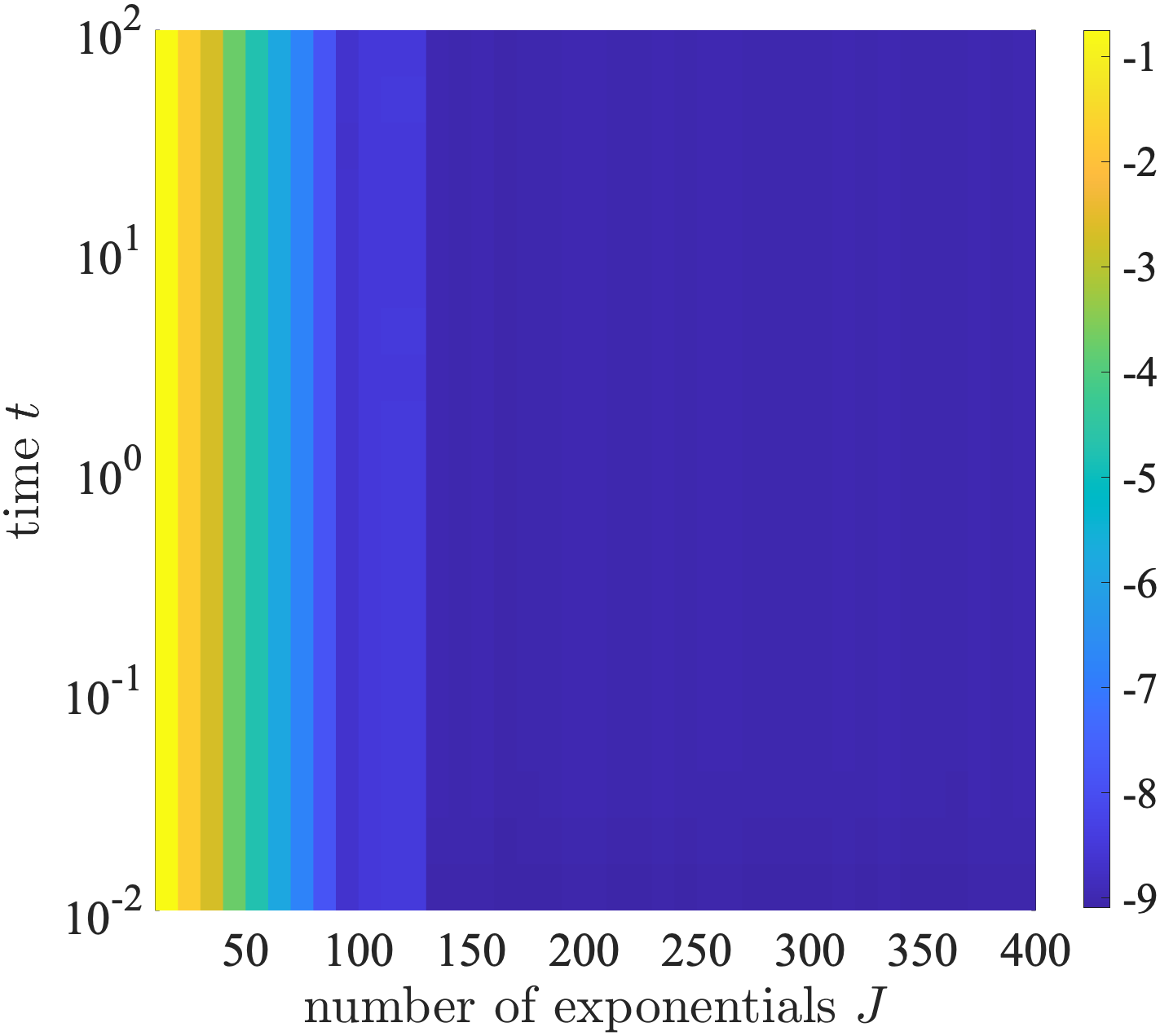}
\par\end{centering}
}\subfloat[$\alpha=0.25$]{\begin{centering}
\includegraphics[width=0.32\textwidth]{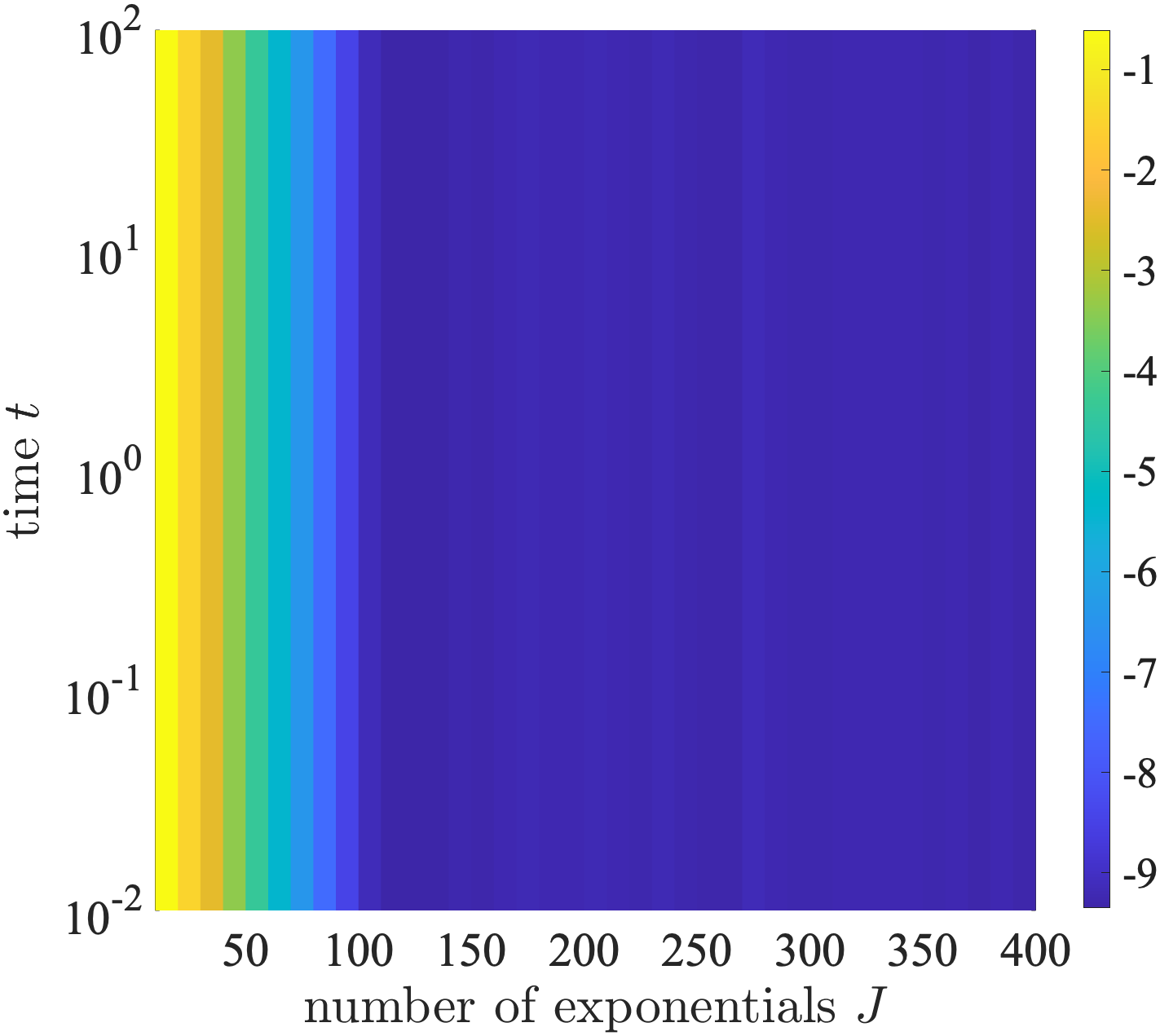}
\par\end{centering}
}

\caption{Base-10 logarithm of the scalar error of the SOE approximation for different times $t$ and number of addends $J$ for an Erd\H{o}s--R\'{e}nyi (ER) random graph with $250$ vertices and $1000$ edges. The plots use $\alpha=0.8$ (A), $\alpha=0.5$ (B), and $\alpha=0.25$ (C).}
\label{fig:errors_ER}
\end{figure}

Let us now focus on the nature of the SOE for the three different
values of $\alpha$, with $J=61$. Using our approach we have the
\emph{operator}-level approximations 
\begin{equation}
E_{0.8}(-t^{0.8}L)\approx0.242e^{-1.52t^{0.8}L}+0.220e^{-1.19t^{0.8}L}+0.148e^{-0.92t^{0.8}L}+\cdots,
\end{equation}

\begin{equation}
E_{0.5}(-t^{0.5}L)\approx0.119e^{-1.48t^{0.5}L}+0.115e^{-1.15t^{0.5}L}+0.108e^{-1.90t^{0.5}L}+\cdots,
\end{equation}
and

\begin{equation}
E_{0.25}(-t^{0.25}L)\approx0.096e^{-1.10t^{0.25}L}+0.095e^{-1.42t^{0.25}L}+0.091e^{-0.86t^{0.25}L}+\cdots.
\label{eq:example-SOE-alpha025}
\end{equation}

Here each coefficient $a_{j}>0$ may be read as the \emph{importance}
of the corresponding exponential in the operator mixture; by construction
$\sum_{j}a_{j}=1$ so these can also be seen as convex weights. The
remaining question is the physical meaning of the rates $b_{j}$. 
For larger values of $\alpha$, the weight of coefficients is skewed towards one diffusion. Indeed, in the limit $\alpha \to 1$, we recover exactly the diffusive behavior, so we expect $a_1=1$ and $a_j=0$ for $j\neq 2$. As $\alpha$ decreases, the contributions of all other diffusions is more prominent. The three leading coefficients in Equation~\eqref{eq:example-SOE-alpha025} indicate that the respective diffusion processes have (almost) the same importance.

To illustrate the importance of the strength of subdiffusion on \emph{vertex
waiting times} as previously discussed theoretically in this work let
us consider a vertex $i$ of degree one in the ER random graph studied
before. Then, for $t=1$ and  $\alpha=0.8, 0.5, 0.25$ we have the following
\emph{survival functions} respectively:

\begin{equation*}
	\begin{aligned}
		S_{i}(0.8,1)&\approx0.242e^{-1.52}+0.220e^{-1.19}+0.148e^{-0.92}+\cdots \approx 0.3869, \\
		S_{i}(0.5,1)&\approx0.119e^{-1.48}+0.115e^{-1.15}+0.108e^{-1.90}+\cdots\approx 0.4276,\\
		S_{i}(0.25,1)&\approx0.096e^{-1.10}+0.095e^{-1.42}+0.091e^{-0.86}+\cdots\approx 0.4639.
	\end{aligned}
\end{equation*}

As in the approximation of the operator $E_{\alpha}(-t^{\alpha}L)$, for large $\alpha$ the weight distribution is skewed toward the waiting time of the most influential distributions. Conversely, as $\alpha$ decreases, the subdiffusive effect is stronger. Particles experience longer residence times at each vertex, before the hop to the next vertex occurs.

What happen for vertices of larger degree? Obviously, the waiting
time decreases with the increase of the degree because the diffusive
particle has more choices to escape from the vertex. But, why does it also
drop with the increase of time? We hypothesize here that the waiting time decreases with the increase of the
physical time, because the diffusive particle remembers the paths
used previously, such that it has not much to wait
to continue a path it already knows.

\section{Subdiffusive distance and paths}
\label{sec:shortest-paths}

\subsection{Subdiffusive distance}
Let us consider the solution of the Caputo time-fractional diffusion
equation: $u(t)=E_{\alpha}(-t^{\alpha}L)\,u_{0}$ in which we consider
the initial condition $u_{0}=e_{v}$, where $e_{v}$ is the vector
with one at position $v$ and zero elsewhere. Then, let us consider
the capacity of the whole graph to diffuse mass between the vertices
$v$ and $w$ at time $t$. That is, we consider at time $t$ the
difference between the mass remaining at the origin, i.e., vertex
$v$, minus the mass diffused to vertex $w$:
\begin{equation*}
u_{v|u_{0}=e_{v}}\left(t\right)-u_{w|u_{0}=e_{v}}\left(t\right)=\left(E_{\alpha}(-t^{\alpha}L)\right)_{vv}-\left(E_{\alpha}(-t^{\alpha}L)\right)_{vw}.
\end{equation*}

Similarly, the capacity of the whole graph to move mass from $w$
to $v$ conditioned to all initial mass being allocated at $w$ is:
\begin{equation*}
u_{w|u_{0}=e_{w}}\left(t\right)-u_{v|u_{0}=e_{w}}\left(t\right)=\left(E_{\alpha}(-t^{\alpha}L)\right)_{ww}-\left(E_{\alpha}(-t^{\alpha}L)\right)_{wv}.
\end{equation*}

As we only consider undirected graphs here, the total capacity of
mass diffusion between both vertices is:
\begin{equation}
\mathcal{\mathscr{D}}_{\alpha,t}\left(v,w\right)=\left(E_{\alpha}(-t^{\alpha}L)\right)_{vv}+\left(E_{\alpha}(-t^{\alpha}L)\right)_{ww}-2\left(E_{\alpha}(-t^{\alpha}L)\right)_{vw}.
\end{equation}

Then, we have the following result. 
\begin{prop}
Let $\mathcal{\mathscr{D}}_{\alpha,t}\left(v,w\right)$ as defined
before. Then \textup{$\mathcal{\mathscr{D}}_{\alpha,t}\left(v,w\right)$}
is real and non-negative. There exist an embedding of vertices $V\to\mathbb{R}^{n}$,
$i\to x_{i}$ such that $\mathcal{\mathscr{D}}_{\alpha,t}\left(v,w\right)=\|x_{v}-x_{w}\|_{2}^{2}$. 
\end{prop}
\begin{proof}
Using positivity $E_{\alpha}(-t^{\alpha}\lambda)\geq0$
for $0<\alpha<1$ we write
\begin{align*}
E_{\alpha}(-t^{\alpha}L) & =VE_{\alpha}(-t^{\alpha}\varLambda)V^{\top}\\
 & =V\left(E_{\alpha}(-t^{\alpha}\varLambda)\right)^{1/2}\,\left(E_{\alpha}(-t^{\alpha}\varLambda)\right)^{1/2}V^{\top}\\
 & =\left(VE_{\alpha}(-t^{\alpha}\varLambda)\right)^{1/2}\left(\left(VE_{\alpha}(-t^{\alpha}\varLambda)\right)^{1/2}\right)^{\top}
\end{align*}
Denote $X_{\alpha}\left(t\right)=\left(VE_{\alpha}(-t^{\alpha}\varLambda)\right)^{1/2}$
and by $x_{\alpha,t}^{\top}\left(v\right)$ the $v$-th row of $X$.
Thus, 
\[
\left(E_{\alpha}(-t^{\alpha}L)\right)_{vw}=x_{\alpha,t}^{\top}\left(v\right)x_{\alpha,t}\left(w\right).
\]
Then, $v\mapsto x_{\alpha,t}^{\top}\left(v\right)$ is an embedding of
vertices in $\mathbb{R}^{n}$. Therefore 
\[
\mathcal{\mathscr{D}}_{\alpha,t}\left(v,w\right)=\|x_{\alpha,t}\left(v\right)-x_{\alpha,t}\left(w\right)\|_{2}^{2},
\]
so $\mathcal{\mathscr{D}}_{\alpha,t}\left(v,w\right)$ is a square
Euclidean distance between the vertices $v$ and $w$. 
\end{proof}
Hereafter we call $\mathcal{\mathscr{D}}_{\alpha,t}\left(v,w\right)$
the (squared) subdiffusive distance between the corresponding vertices in the
graph when $\alpha<1$. Notice that for $\alpha=1$, the quantity $\mathcal{\mathscr{D}}_{1,t}\left(v,w\right)$
is the diffusion distance defined by Coiffman et al. \cite{coifman2005geometric,coifman2006diffusion}.
These distances are part of a large family of diffusion-like distances
on graphs \cite{estrada2012communicability,estrada2014hyperspherical}.

\subsection{Subdiffusive shortest paths}

Instead of considering the subdiffusive distance between any pair
of vertices in the graph, which may be adjacent or not, let us instead
consider what should be the trajectory of a single subdiffusive particle
between two vertices of the graph. In order to compare such trajectories
for a standard diffusive particle and a subdiffusive one we define
the following general geometrization of the graph (see \cite{markvorsen2008minimal,bridson2013metric}). 

For that purpose we consider every edge $e=(v,w)\in E$ as a compact
$1$--dimensional manifold with boundary $\partial e=\{v,w\}$. To
each edge we assign a time--dependent length 
\[
W_{\alpha,t}(v,w)=\begin{cases}
\sqrt{\mathcal{\mathscr{D}}_{\alpha,t}\left(v,w\right)}, & (v,w)\in E,\\[4pt]
0, & (v,w)\notin E,
\end{cases}
\]
so that the geometric edge is the interval $\widetilde{e}_{\alpha,t}(v,w)=[0,W_{\alpha,t}(v,w)]$.
Equipping each edge with this length measure turns the graph into
a (time--dependent) metric length space by extending distances through
infima of curve lengths in the associated $1$--dimensional complex.
In practice we obtain this geometrization by means of the following
definition.
\begin{defn}\label{def:weighted-graph}
Let $\sqrt{\mathcal{\mathscr{D}}_{\alpha,t}}$ be the matrix
of (subdiffusive) distances, intended as the entry-wise square root of $\mathcal{\mathscr{D}}_{\alpha,t}$. 
Define a weighted graph whose adjacency matrix is: 
$W=A\odot\sqrt{\mathcal{\mathscr{D}}_{\alpha,t}}$,
which assigns to each edge $(i,j)$ the cost induced by the discrepancy
of their diffusion profiles at time $t$. The resulting weighted graph
$\widetilde{G}$ is the geometrization of the original graph $G.$
\end{defn}

The resulting shortest paths (computed via Dijkstra's algorithm) identify
chains of vertices whose subdiffusion states remain maximally coherent
at time~$t$. We dub these paths the \emph{subdiffusive shortest paths}. These paths are not, in general,
the paths that require the fewest arcs of $A$, which we call \emph{topological shortest paths} or \emph{geodesic shortest paths}. They generalize the shortest communicability paths which were introduced in \cite{silver2018tuned}
with matrix function $f(A)=\exp(\beta A)$.

Consider the subdiffusive shortest paths on $W(t)$ as time
varies. When $t$ is close to zero, the effects of diffusion are minimal.
That is, each particle has explored only a small part of the graph.
Thus, local properties of the graph are highlighted by the weighted
metric. We then have the following two results.
\begin{thm}[Shortest-path dominance for the fractional heat kernel as $t\to0$]
\label{thm:Shortest-path-dominance} Let $G=(V,E)$ be a simple
undirected graph with combinatorial Laplacian $L=D-A$, and let $0<\alpha<1$.
Denote by $d(i,j)$ the graph distance between vertices $i\neq j$.
Then, for every $i\neq j$, 
\[
\bigl(E_{\alpha}(-t^{\alpha}L)\bigr)_{ij}=\frac{(-t^{\alpha})^{d(i,j)}}{\Gamma(\alpha d(i,j)+1)}\,(L^{d(i,j)})_{ij}+O\!\bigl(t^{\alpha(d(i,j)+1)}\bigr),\qquad t\downarrow0.
\]
In particular, 
\[
\bigl(E_{\alpha}(-t^{\alpha}L)\bigr)_{ij}=O\!\bigl(t^{\alpha \,d(i,j)}\bigr),
\]
and the leading-order contribution is determined by topological shortest paths between $i$ and $j$.

Moreover, 
\[
\bigl(E_{\alpha}(-t^{\alpha}L)\bigr)_{ij}=\frac{t^{\alpha\, d(i,j)}}{\Gamma(\alpha d(i,j)+1)}\,(A^{d(i,j)})_{ij}+O\!\bigl(t^{\alpha(d(i,j)+1)}\bigr),
\]
where $(A^{d(i,j)})_{ij}$ counts the number of shortest walks of
length $d(i,j)$ from $i$ to $j$. 
\end{thm}

\begin{proof}
The Mittag-Leffler function admits the operator series representation
\[
E_{\alpha}(-t^{\alpha}L)=\sum_{m=0}^{\infty}\frac{(-t^{\alpha})^{m}}{\Gamma(\alpha m+1)}\,L^{m},
\]
which converges absolutely for all $t\ge0$ since $L$ is bounded
(see, e.g., \cite{metzler2000random}).

Taking the $(i,j)$ entry yields 
\[
\bigl(E_{\alpha}(-t^{\alpha}L)\bigr)_{ij}=\sum_{m=0}^{\infty}\frac{(-t^{\alpha})^{m}}{\Gamma(\alpha m+1)}\,(L^{m})_{ij}.
\]

As in the classical heat-kernel case, write $L=D-A$. Any term contributing
to $(L^{m})_{ij}$ corresponds to a product of $m$ factors, each
equal to either $D$ or $-A$. Diagonal factors $D$ do not change
the vertex index, while each factor $A$ corresponds to a single edge
traversal. Hence, in order for $(L^{m})_{ij}$ to be nonzero, the
word must contain at least $d(i,j)$ factors of $A$. Consequently,
\[
(L^{m})_{ij}=0\qquad\text{for all }m<d(i,j).
\]

For $m=d(i,j)$, the only contributing word is $(-A)^{d(i,j)}$, since
any appearance of $D$ would prevent reaching $j$ in exactly $d(i,j)$
steps. Therefore, 
\[
(L^{d(i,j)})_{ij}=(-1)^{d(i,j)}(A^{d(i,j)})_{ij}.
\]

Substituting into the series above, the first nonzero term occurs
at $m=d(i,j)$, which proves the stated expansion. 
\end{proof}
\begin{thm}[Short-time selection of shortest paths]
\label{thm:Short-time-selection}Let $G$ be a finite connected graph and let $X(t)$ be the
time-fractional continuous-time random walk associated with the Caputo
fractional diffusion equation on $G$ (with $0<\alpha<1$). Let $N(t)$
be the number of jumps of $X(\cdot)$ up to time $t$, and let $d(i,j)$
be the graph distance between distinct vertices $i\neq j$. Then,
for every $i\neq j$, 
\[
\lim_{t\downarrow0}\mathbb{P}\!\bigl(N(t)=d(i,j)\,\big|\,X(t)=j,\ X(0)=i\bigr)=1.
\]
In particular, 
\[
\lim_{t\downarrow0}\mathbb{P}\!\bigl(\text{the jump sequence from \ensuremath{i} to \ensuremath{j} up to time \ensuremath{t}has length }d(i,j)\ \big|\ X(t)=j,\ X(0)=i\bigr)=1.
\]
That is, conditioned on arrival at $j$ at very short times, the process
selects a topological shortest path with probability tending to $1$. 
\end{thm}

\begin{proof}
Write $d=d(i,j)$. Let $Y_{m}$ be the embedded discrete-time jump
chain of $X(t)$. The fractional CTRW representation yields 
\[
\mathbb{P}(X(t)=j\mid X(0)=i)=\sum_{m=0}^{\infty}\mathbb{P}(N(t)=m)\,\mathbb{P}(Y_{m}=j\mid Y_{0}=i),
\]
see \cite{BaeumerMeerschaert2001,MeerschaertNaneVellaisamy2011}.
Let $P$ be the one-step transition matrix of $Y_{m}$, so that $\mathbb{P}(Y_{m}=j\mid Y_{0}=i)=(P^{m})_{ij}$.
By definition of graph distance, $(P^{m})_{ij}=0$ for all $m<d$.
Hence, 
\[
\mathbb{P}(X(t)=j\mid X(0)=i)=\sum_{m=d}^{\infty}\mathbb{P}(N(t)=m)\,(P^{m})_{ij}.
\]

For the time-fractional walk, the jump-count distribution satisfies
the short-time scaling 
\[
\mathbb{P}(N(t)=m)=\frac{t^{\alpha m}}{\Gamma(\alpha m+1)}+O\bigl(t^{\alpha(m+1)}\bigr),\qquad t\downarrow0,
\]
see \cite{MeerschaertNaneVellaisamy2011}. Therefore the leading contribution
comes from $m=d$, and we obtain 
\[
\mathbb{P}(X(t)=j\mid X(0)=i)=\frac{t^{\alpha d}}{\Gamma(\alpha d+1)}(P^{d})_{ij}+O\bigl(t^{\alpha(d+1)}\bigr).
\]
Similarly, 
\[
\mathbb{P}\bigl(N(t)=d,\ X(t)=j\mid X(0)=i\bigr)=\frac{t^{\alpha d}}{\Gamma(\alpha d+1)}(P^{d})_{ij}+O\bigl(t^{\alpha(d+1)}\bigr).
\]
Dividing the two expansions yields 
\[
\mathbb{P}\bigl(N(t)=d\,\big|\,X(t)=j,\ X(0)=i\bigr)=1+O(t^{\alpha}),
\]
which proves the claim.

Finally, on the event $\{N(t)=d,\,X(t)=j,\,X(0)=i\}$ the jump sequence
$(Y_{0},\dots,Y_{d})$ has length $d$ and connects $i$ to $j$,
hence it is a topological shortest path. 
\end{proof}
\begin{rem}
[Physical interpretation of short-time path selection] The previous
results (Theorems~\ref{thm:Shortest-path-dominance} and~\ref{thm:Short-time-selection})
show that, for both classical and time-fractional diffusion on graphs,
the short-time behavior is governed by topological constraints rather
than by long-time transport mechanisms. In the fractional case, memory
effects and heavy-tailed waiting times manifest themselves through
the slower time scaling $t^{\alpha d(i,j)}$ of transition probabilities;
however, they do not alter the mechanism by which mass first propagates
across the graph. At very short times, the process has insufficient
opportunity to perform redundant or backtracking moves, and any realization
that reaches a vertex $j$ from $i$ must do so using the minimal
number of jumps permitted by the graph distance. Thus, topological shortest
paths dominate not because they are energetically or entropically
preferred, but because they are the only dynamically admissible routes
in the short-time regime. Memory affects \emph{when} such paths become
observable, but not \emph{which} paths contribute to the leading-order
behavior. This provides a precise mathematical explanation for the
observed agreement between diffusion-based distances and graph distances
at very small times, even in the presence of anomalous (fractional)
temporal dynamics. 
\end{rem}

Using Theorem~\ref{thm:Shortest-path-dominance} for two adjacent vertices $v,w$, we can compute the first-order term of the subdiffusive distance for $t \to 0$
\begin{equation}
\begin{aligned}\mathcal{\mathscr{D}}_{\alpha,t}\left(v,w\right) & =E_{\alpha}(-t^{\alpha}L)_{vv}+E_{\alpha}(-t^{\alpha}L)_{ww}-2E_{\alpha}(-t^{\alpha}L)_{vw}\\
 & =2-\frac{t^{\alpha}}{\Gamma(\alpha+1)}\left(L_{vv}+L_{ww}-2L_{vw}\right)+O(t^{2\alpha})\\
 & =2-\frac{t^{\alpha}}{\Gamma(\alpha+1)}\left(d_{v}+d_{w}+2\right)+O(t^{2\alpha}).
\end{aligned}
\end{equation}
Note that usually there are multiple topological shortest paths between
the same pair of vertices. To understand which of these shortest paths
coincides with the subdiffusive shortest path at $t\to0$ we find
\begin{equation}
W_{\alpha,t}\left(v,w\right)=\sqrt{2}-\frac{\sqrt{2}}{4}\frac{t^{\alpha}}{\Gamma(\alpha+1)}(d_{v}+d_{w}+2)+O\left(t^{2\alpha}\right).
\end{equation}
Let $P$ be a path having $m(P)$ edges. In this
case the sum of the edge weights in the path is
\[
W_{\alpha,t}\left(P\right)=\sqrt{2}\,m(P)\left(1+\frac{t^{\alpha}}{\Gamma(\alpha+1)}\right)-\frac{\sqrt{2}}{4}\frac{t^{\alpha}}{\Gamma(\alpha+1)}\sum_{k=1}^{m\left(P\right)}\delta_{e_{k}}+O\left(t^{2\alpha}\right),
\]
where $\delta_{e_{k}}=d_{i}+d_{j}-2$ is the degree of the edge $e_{k}=\left(i,j\right)$.
Consequently, among all the topological shortest paths, the one with
the largest sum of edge degrees is chosen by the subdiffusive particle,
as it has the smallest subdiffusive length. 

\subsection{Subdiffusive shortest paths on a geometric graph}

We will analyze computationally some of the previous analytical results
by considering a Gabriel graph with $n=600$ vertices and 1156 edges.
We use this graph because it is geometric in the sense that its vertices
are embedded in $\mathbb{R}^{d}$ and it can be easily visualized,
specially for illustrating paths. Gabriel graphs are defined as follows.
\begin{defn}
Let $P\subset\mathbb{R}^{d}$ be a finite set of points, referred
to as \emph{generators}. The Gabriel graph $G_{G}=(V_{G},E_{G})$
associated with $P$ has vertex set $V_{G}=P$. Two distinct vertices
$v,w\in P$ are connected by an edge $\{v,w\}\in E_{G}$ if and only
if the closed ball having the segment $[vw]$ as its diameter contains
no other points of $P$. 
\end{defn}

In this work, we restrict our attention to the planar case $d=2$
and use a rectangle with length to width proportion of 2:1 instead
of a square to embed the points. We geometrize the Gabriel graph using
both the subdiffusive communicability distance $\mathcal{\mathscr{D}}_{\alpha,t}(v,w)$
and its sum-of-exponentials (SOE) approximation. Specifically, we
use 
\[
F_{J}(t,L)=\sum_{j=1}^{J}a_{j}e^{-t^{\alpha}b_{j}L},
\]
and introduce the corresponding approximate squared distance 
\begin{equation}
\mathcal{\mathscr{\widetilde{D}}}_{J,\alpha,t}(v,w)=\left(F_{J}(t,L)\right)_{vv}+\left(F_{J}(t,L)\right)_{ww}-2\left(F_{J}(t,L)\right)_{vw}.
\end{equation}

Throughout this section, we consider the case $\alpha=0.85$. Using
the square roots of both $\mathcal{\mathscr{D}}_{\alpha,t}(v,w)$
and $\mathcal{\mathscr{\widetilde{D}}}_{J,\alpha,t}(v,w)$, we construct the weighted graph as in Definition~\ref{def:weighted-graph} and compute the subdiffusive shortest paths.

\subsubsection{Experimental results}
We begin by reporting the topological shortest paths (TSP) between
two vertices located near opposite corners of the Gabriel graph (see
Fig.~\ref{fig:SOE paths multiple J}(A)). As is typical---even for
planar graphs such as Gabriel graphs---there exist multiple TSPs
between a given pair of vertices. These TSPs are colored according
to the \emph{average edge degree} along the path, defined as 
\[
\delta_{e(v,w)}=d_{v}+d_{w}-2,
\]
where $d_{j}$ denotes the degree of vertex $j$.

\begin{figure}
\begin{centering}
\subfloat[topological shortest paths]{\begin{centering}
\includegraphics[width=0.49\textwidth]{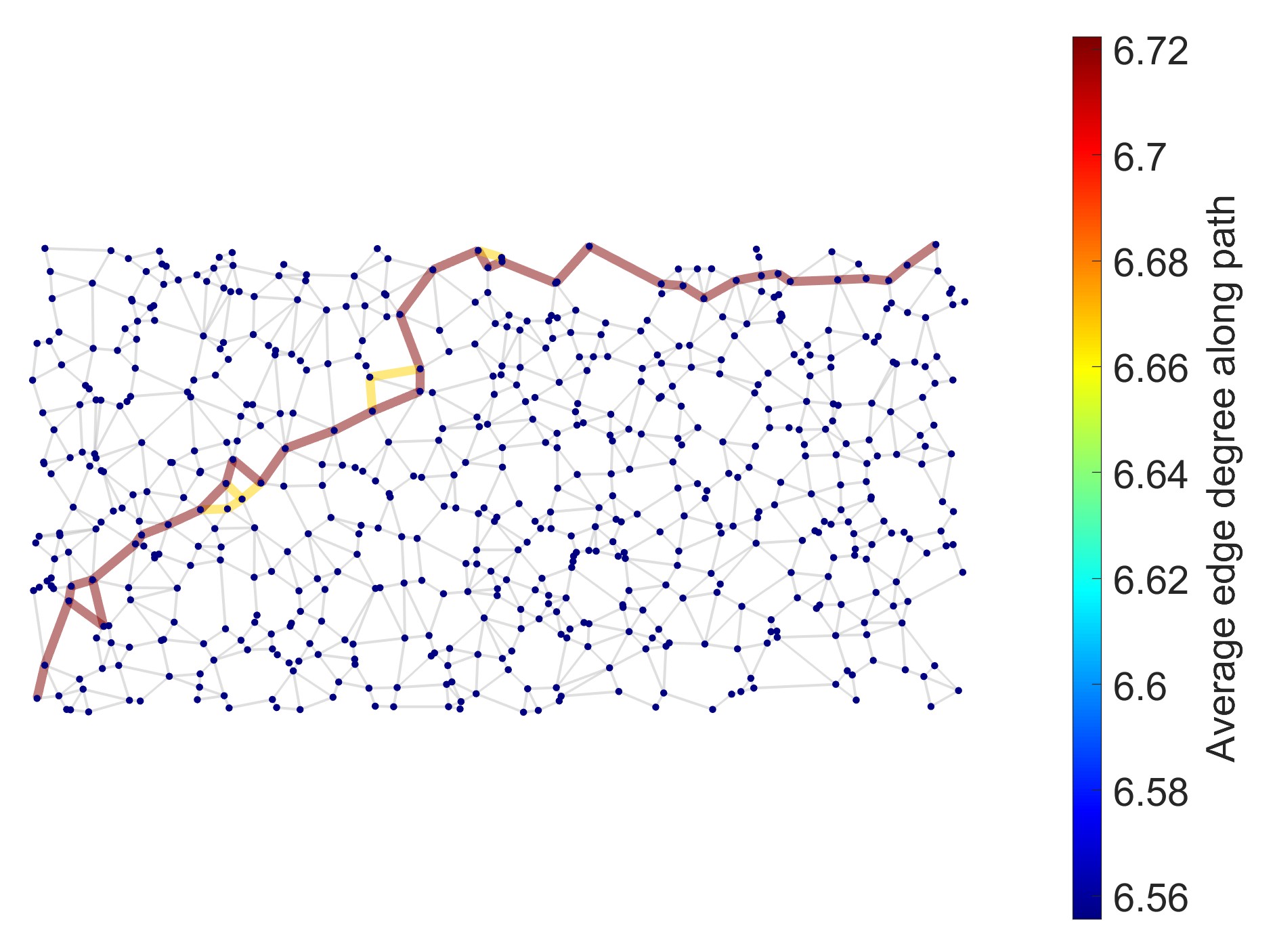}
\par\end{centering}
}\subfloat[$J=1$]{\begin{centering}
\includegraphics[width=0.49\textwidth]{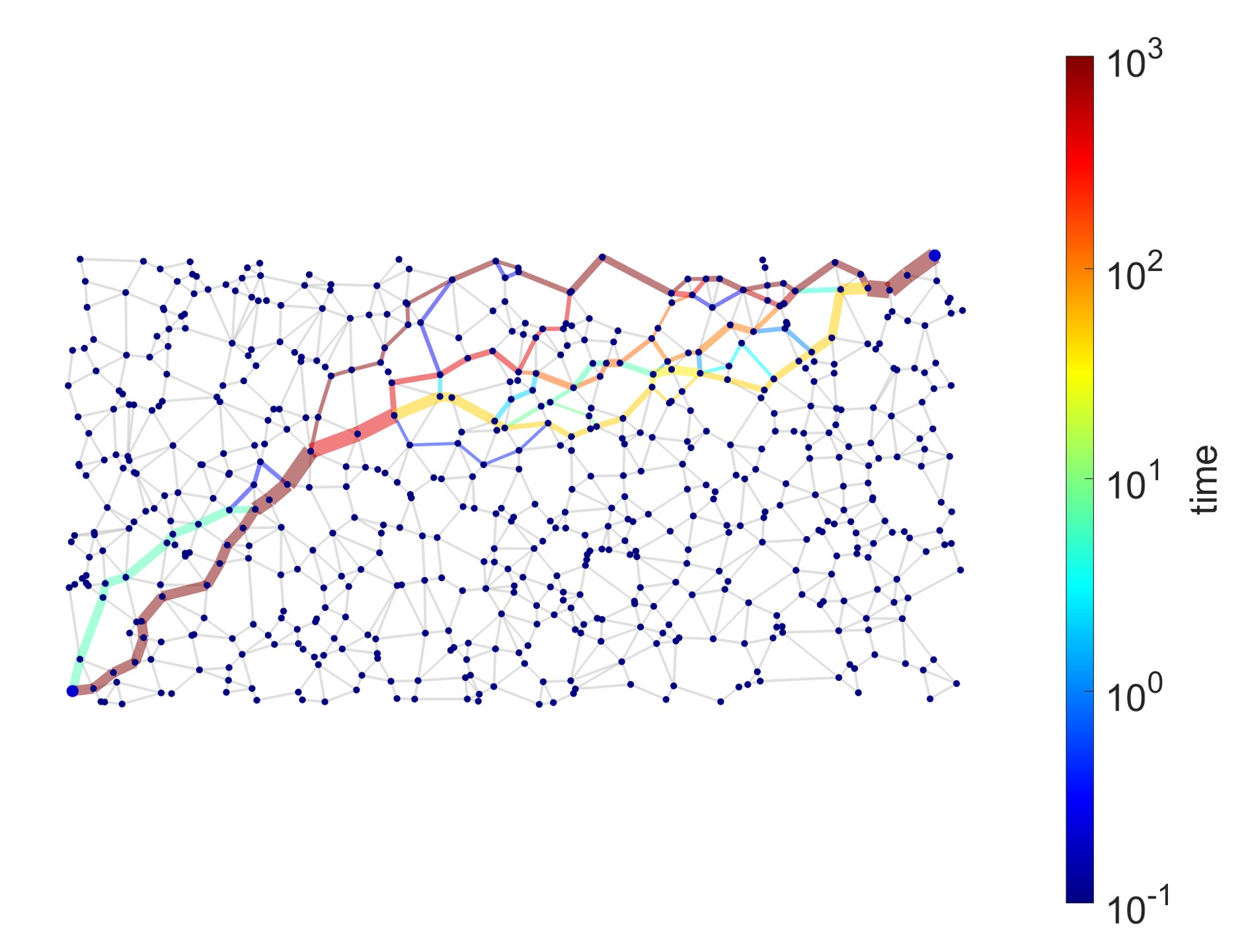}
\par\end{centering}
}
\par\end{centering}
\begin{centering}
\subfloat[$J=10$]{\begin{centering}
\includegraphics[width=0.49\textwidth]{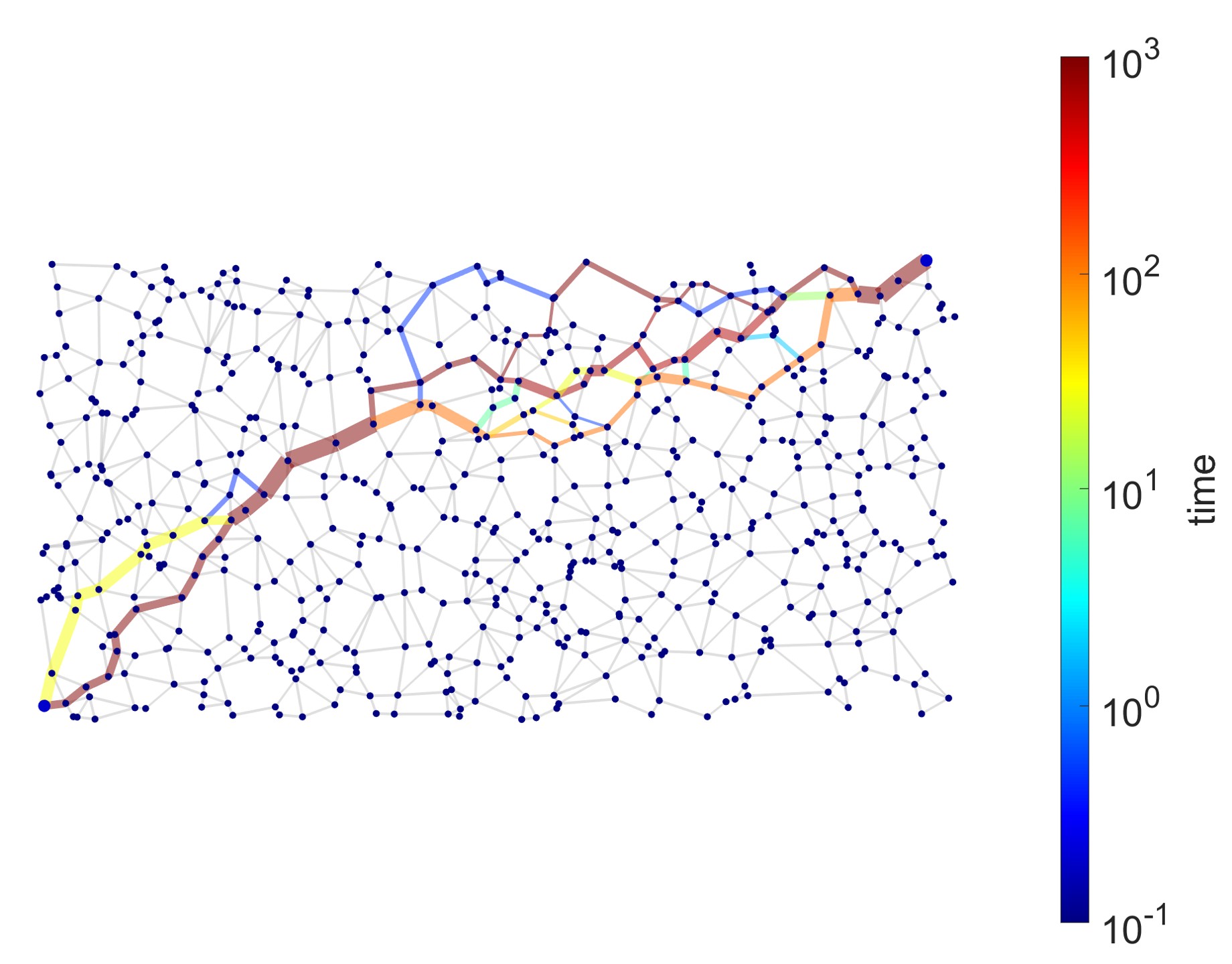}
\par\end{centering}
}\subfloat[$J=20$]{\begin{centering}
\includegraphics[width=0.49\textwidth]{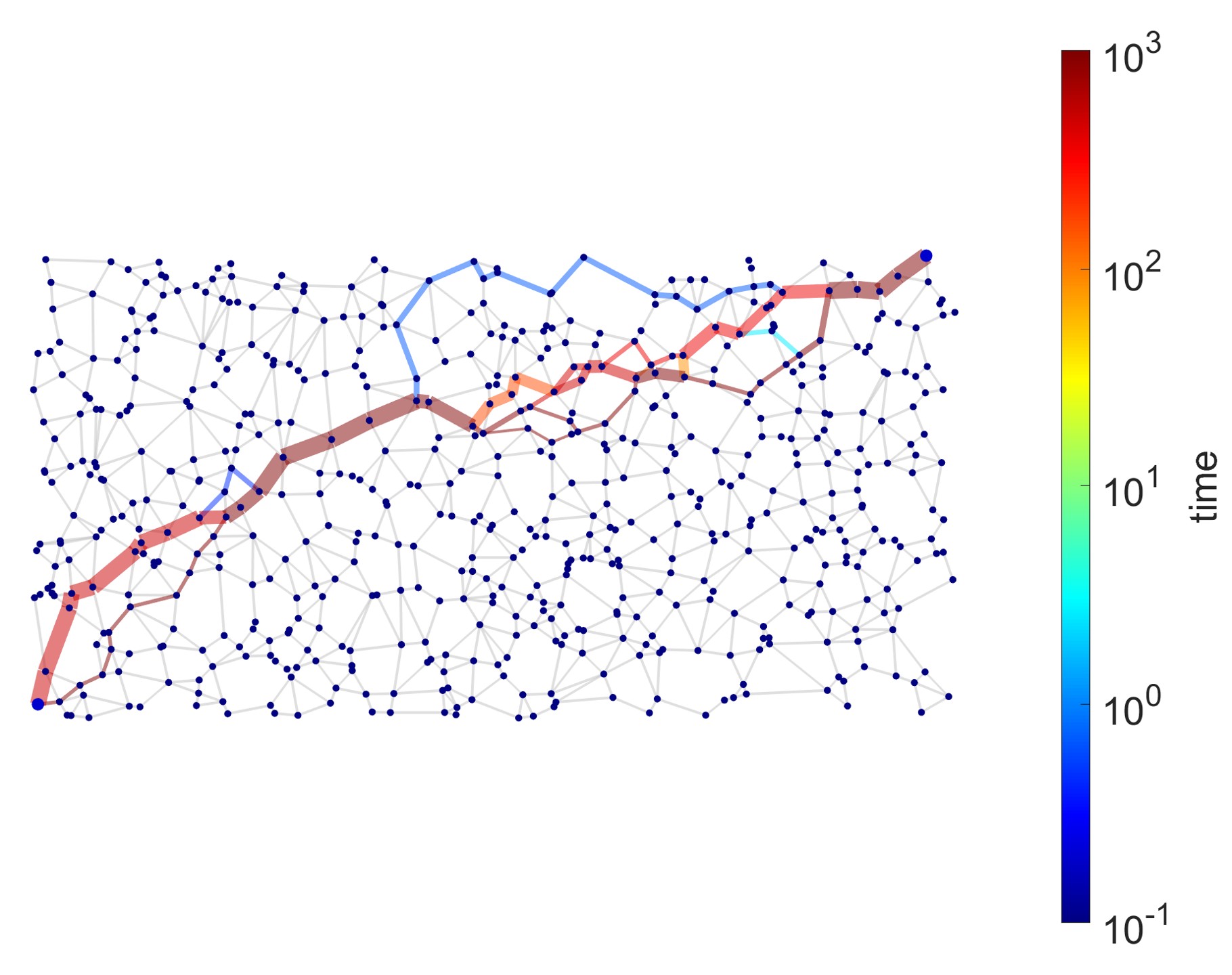}
\par\end{centering}
}
\par\end{centering}
\begin{centering}
\subfloat[$J=40$]{\begin{centering}
\includegraphics[width=0.49\textwidth]{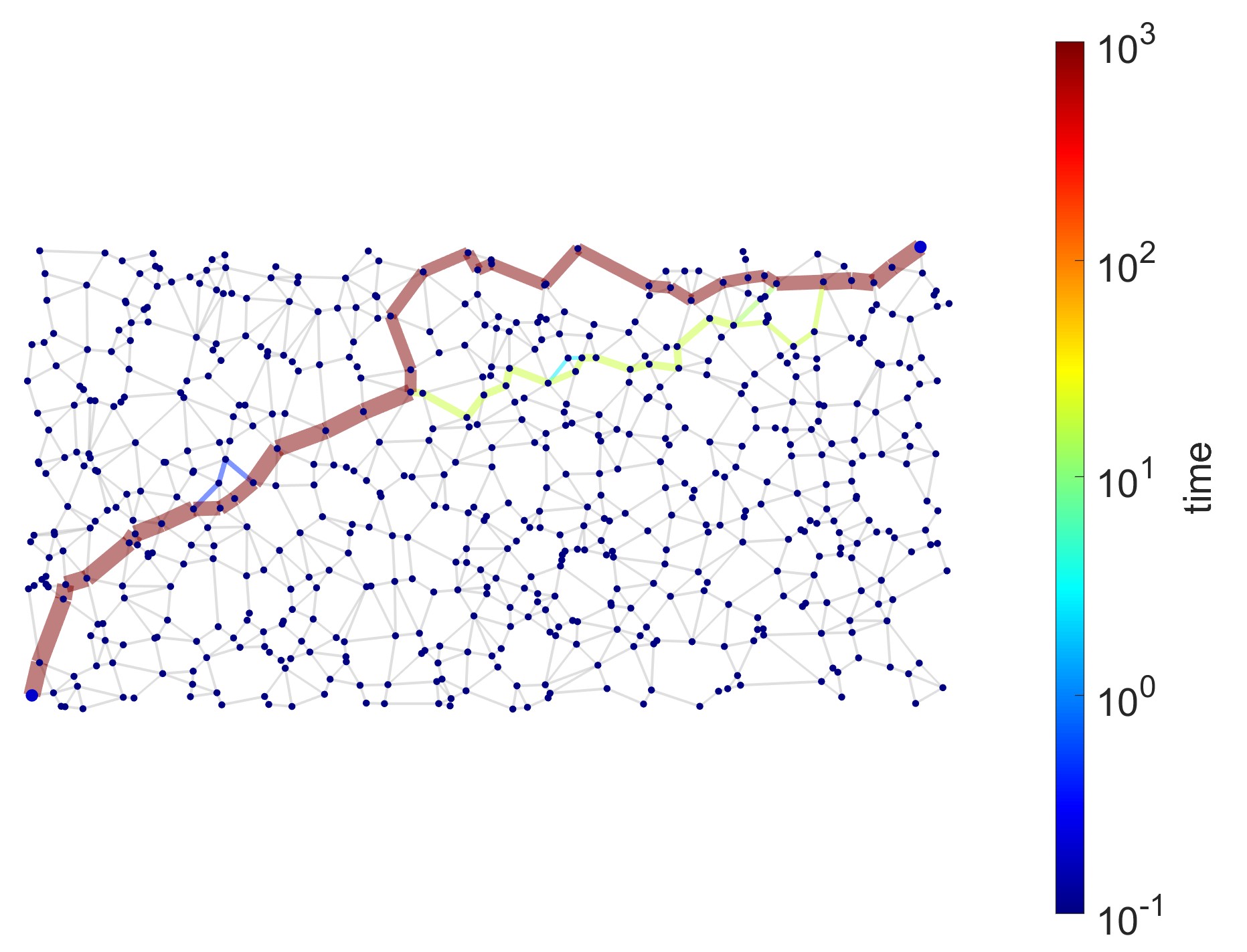}
\par\end{centering}
}\subfloat[subdiffusive shortest paths]{\begin{centering}
\includegraphics[width=0.49\textwidth]{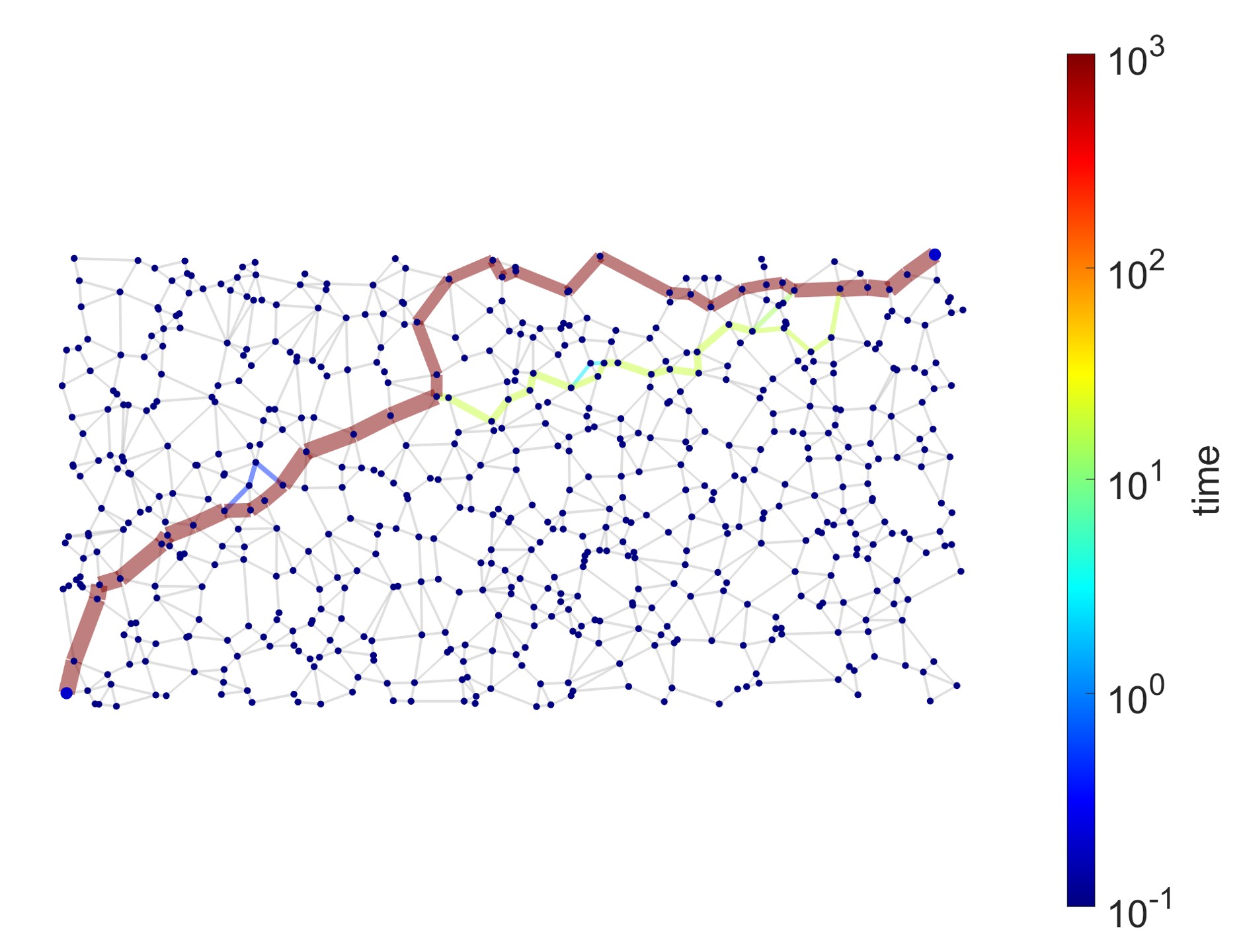}
\par\end{centering}
}
\par\end{centering}
\caption{(A) Illustration of all topological shortest paths between a given
pair of vertices in the example Gabriel graph. Coloring is by edge
degrees. Shortest paths of SOE approximation for $J=1,10,20,40$ 
(B, C, D, E respectively) and $\alpha=0.85$. (F) Subdiffusive shortest
path between the same pair of vertices as before obtained from the
Mittag-Leffler function. In (B)-(F) the paths are colored according
to the time, which is taken as $0.1\protect\leq t\protect\leq1000$
with 300 steps in the interval, at which they are observed and the
thickness of the edges in proportional to the number of times they
are by the corresponding paths.}

\label{fig:SOE paths multiple J}
\end{figure}

We next consider the shortest paths induced by $\mathcal{\mathscr{\widetilde{D}}}_{J,\alpha,t}(v,w)$
for $J=1$ (see Fig.~\ref{fig:SOE paths multiple J}(B)). This case
corresponds to a standard diffusive process, since it involves a single
exponential term, which is the solution of the classical diffusion
equation. We observe that, while at early times the diffusive paths
coincide with the TSP as predicted by Theorems~\ref{thm:Shortest-path-dominance}
and~\ref{thm:Short-time-selection}, they progressively deviate as
time increases.

As the number of exponentials $J$ in the SOE approximation increases,
the corresponding shortest paths converge toward the TSP, as illustrated
in panels (C)--(E) of Fig.~\ref{fig:SOE paths multiple J}. The
exact solution based on the Mittag-Leffler function is shown in
Fig.~\ref{fig:SOE paths multiple J}(F), and clearly demonstrates
that the subdiffusive shortest paths coincide with the TSPs. Moreover,
the limiting paths correspond to those TSPs with the largest average
edge degree, in agreement with the analytical results derived in the
previous section. Physically, the situation is as follows. At $t\rightarrow0$
the subdiffusive shortest path coincides with the  topological shortest path
as proved by Theorems~\ref{thm:Shortest-path-dominance} and~\ref{thm:Short-time-selection}.
However, in the subdiffusive case, as $t$ evolves, the particles
remember the previous path used to navigate between two vertices,
i.e., the shortest topological path, and repeat it. The result is
the repeated use of these paths as shown in Fig.~\ref{fig:SOE paths multiple J}(F).

The time evolution of the vector probing errors, $\mathrm{relerr}(u_{0}(t))$
and $\mathrm{masserr}(u_{0}(t))$, is displayed in Fig.~\ref{fig:error_probing}(A).
Fixing $J$, we can see the mass error decreases with time, while
the relative error remains above a certain threshold, which indicates
that the two solutions $u^{\star}(t)$ and $u^{\mathrm{SOE}}(t)$ are
meaningfully distinct. In Fig.~\ref{fig:error_probing}(B), the
time-averages of the vector probing errors and the operator error
are plotted as a function of $J$. We can see that all the errors
decay. The relative and mass errors reach an accuracy of $10^{-9}$,
while the operator error stagnates $10^{-4}$, due to the larger errors
in floating-point computations of the operators $F_{J}(t,L)$.
\begin{figure}
\centering %
\begin{tabular}{cc}
\includegraphics[width=0.47\textwidth,height=0.25\textheight]{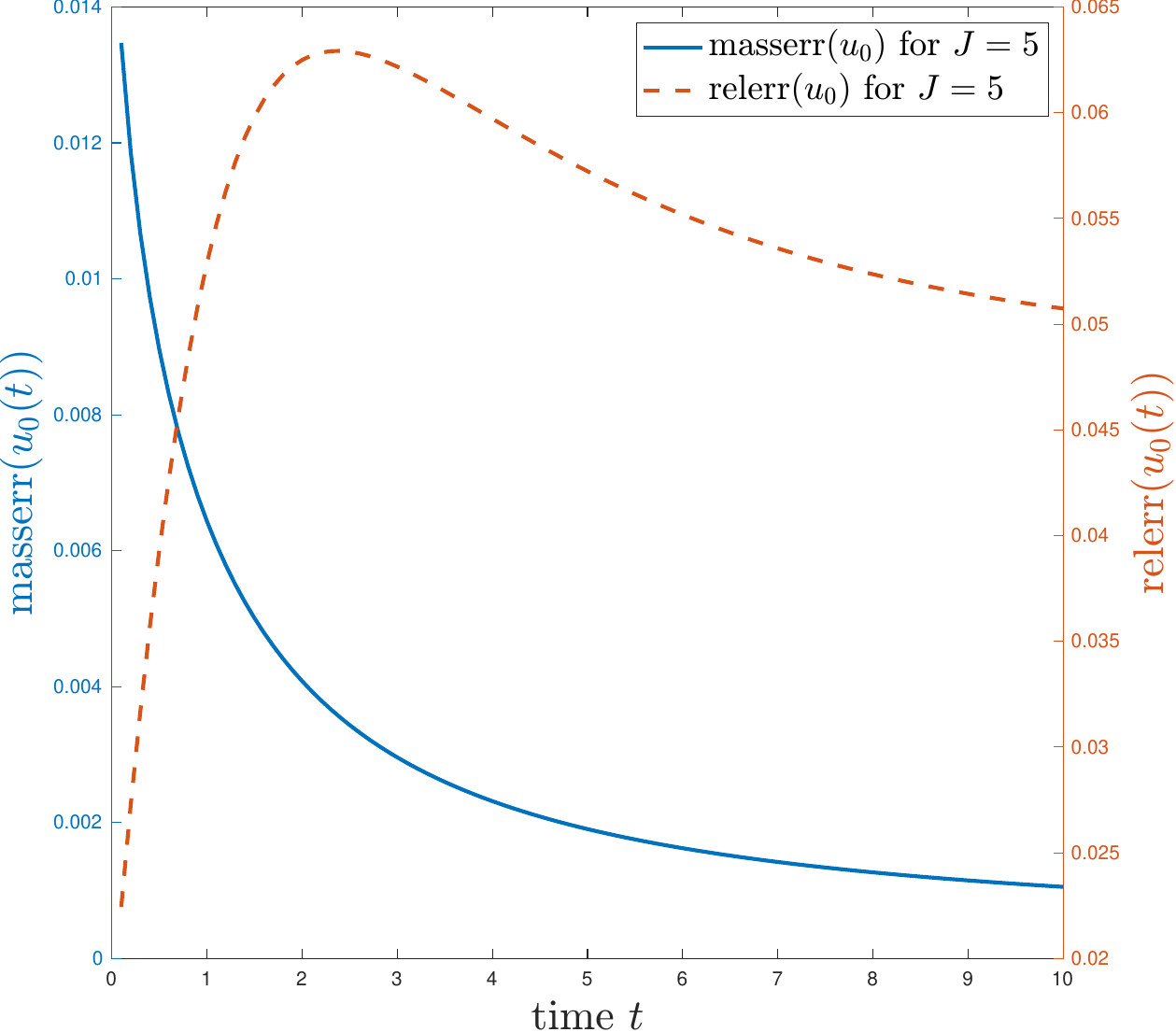} & \includegraphics[width=0.47\textwidth,height=0.25\textheight]{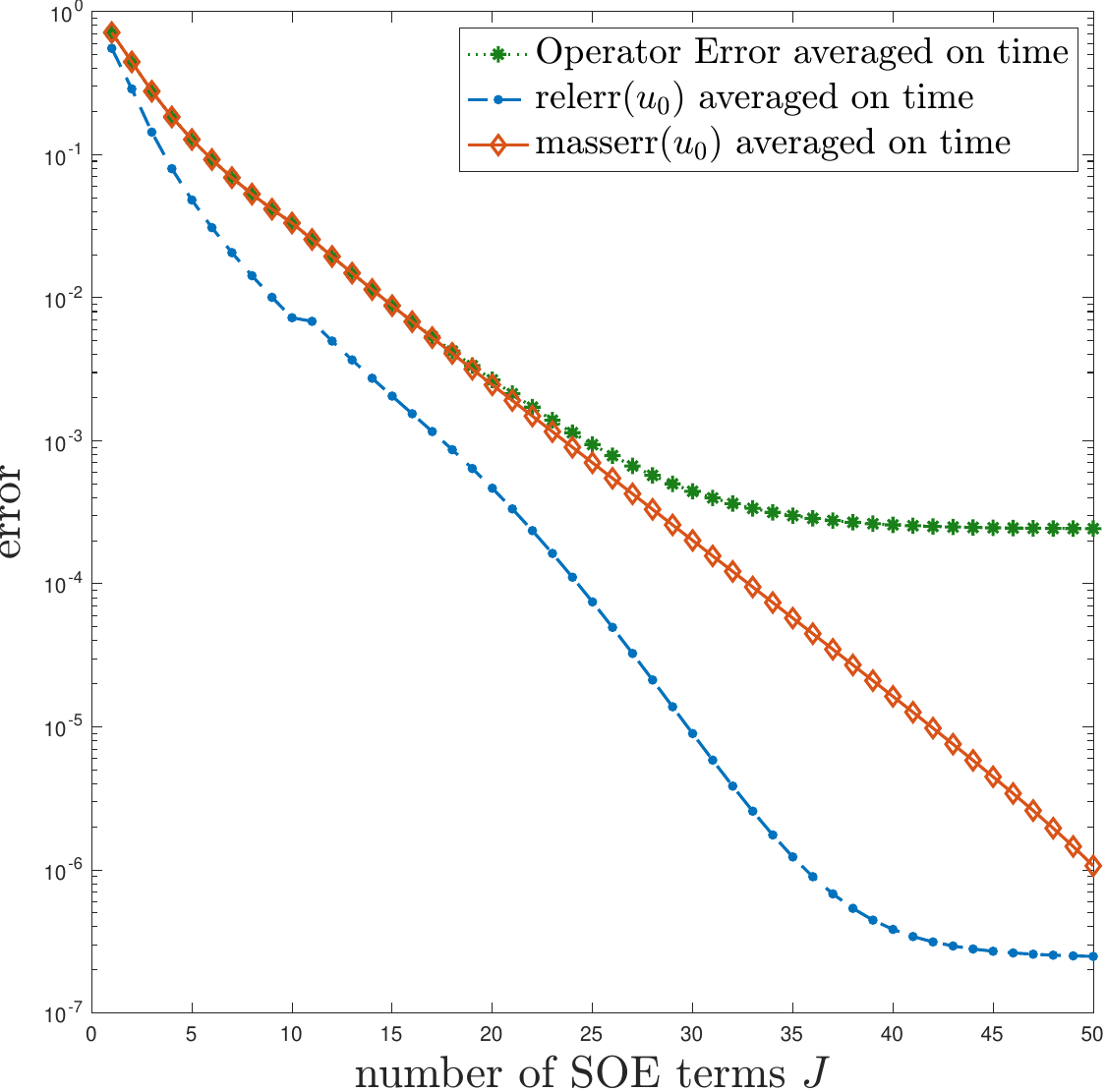}\tabularnewline
(A) & (B)\tabularnewline
\end{tabular}\caption{(A) relative (blue) and mass (orange) errors of the probing vector
$u(t)$. Comparison of the solutions dynamics with exact Mittag-Leffler
and the SOE approximation with $J=5$. (B) average on time of the
operator (green), relative (blue), and mass (orange) errors as function
of $J$. Average is taken on 300 time instants in the interval $0.1\protect\leq t\protect\leq1000$.}
\label{fig:error_probing}
\end{figure}
\subsubsection{Memory reinforcement of past states}
It is clear that increasing $J$ improves the accuracy of the SOE
approximation to the Mittag-Leffler function, leading to convergence
of the corresponding shortest paths. To interpret this convergence
from a physical perspective, recall that the SOE represents a superposition
of diffusive processes. Let $b_{p}=\max\{b_{j}:1\le j\le J\}$. Then
we can rewrite the SOE operator as
\begin{equation}
F_{J}(t,L)=\sum_{j=1}^{J}a_{j}e^{-b_{j}t^{\alpha}L}=\sum_{j=1}^{J}a_{j}e^{-b_{p}\frac{b_{j}}{b_{p}}t^{\alpha}L}=\sum_{j=1}^{J}a_{j}e^{-b_{p}t_{j}^{\alpha}L},
\end{equation}
where the effective time instants are given by 
$t_{j}=\left(\frac{b_{j}}{b_{p}}\right)^{\frac{1}{\alpha}}t$.

In the example graph, the maximum diffusion speed is $b_{p}=b_{2}=1.53$
(see Table~\ref{tab:coeff}). Each term $e^{-b_{p}t_{j}^{\alpha}L}$
therefore represents a diffusion process with speed $b_{p}$, evaluated
at a slowed time $t_{j}^{\alpha}$. Consequently, subdiffusion can be interpreted
not only as a superposition of diffusion processes with different
speeds, but also as the sampling of a single diffusion process at
multiple past time instants.

\begin{table}
\caption{The largest ten coefficients by magnitude of $a_{j}$ of the SOE approximation
as in the general case of Corollary~\ref{cor:window} for $\alpha=0.85$.
These coefficients are used for the SOE approximation for the example
Gabriel graph are displayed.}
\begin{tabular}{ccccccccccc}
\hline 
$j$ & 1 & 2 & 3 & 4 & 5 & 6 & 7 & 8 & 9 & 10\tabularnewline
\hline 
$a_{j}$ & 0.286 & 0.269 & 0.167 & 0.094 & 0.056 & 0.035 & 0.0234 & 0.016 & 0.011 & 0.008\tabularnewline
$b_{j}$ & 1.19 & 1.53 & 0.93 & 0.725 & 0.57 & 0.44 & 0.34 & 0.27 & 0.21 & 0.16\tabularnewline
\hline 
\end{tabular}\label{tab:coeff}
\end{table}

This behavior reflects the intrinsic memory property of subdiffusion.
Unlike classical diffusion, whose evolution depends solely on the
current state, subdiffusive dynamics depend on the entire history
of the process. As shown in Fig.~\ref{fig:SOE paths multiple J},
increasing $J$ incorporates a larger number of past states, allowing
the particle to remember its early-time behavior. In this sense,
both the Caputo fractional derivative and subdiffusion can be viewed
as time-averaged processes over the system's past evolution.

The fractional order $\alpha$ controls the strength of this memory
effect. As $\alpha\to1$, the fractional diffusion equation reduces
to the classical diffusion equation, and memory effects vanish. Conversely,
as $\alpha\to0$, greater weight is assigned to earlier times. This
behavior is illustrated in Fig.~\ref{fig:SOE memory}, where we plot
the dissimilarities of the subdiffusive shortest paths with respect to 
the topogical shortest path for different values of $\alpha$ and $J$. The dissimilarities are computed 
with the Levenshtein distance (also known as \emph{edit distance}), which counts the minimum
number of edits (insertions, deletions, or substitutions) needed
to transform the edge sequence of one path into the other. We observe that smaller values
of $\alpha$ require fewer SOE terms to recover the Mittag-Leffler
shortest paths. For instance, when $J=20$, only two distinct paths
are observed for $\alpha=0.25$, whereas multiple paths persist for
$\alpha=0.85$.

The contrast between the diffusive (Fig.~\ref{fig:SOE paths multiple J}(A)) and subdiffusive (Fig.~\ref{fig:SOE paths multiple J}(F)) behaviors is striking. When memory is not present, the particles are inclined to explore the graph, as the multiple possible paths show. The incipience of memory allows them to recall the topological shortest path, and gravitate towards it. Memory is strengthened either by lowering the parameter $\alpha$, or by considering more terms, i.e. by increasing $J$, as each term is a recall to a previous time instant. Stronger memory means that the Levensthein distance to the TSP is smaller, as shown in Fig.~\ref{fig:SOE paths multiple J}.

\begin{figure}
\centering{}\includegraphics[width=0.5\textwidth]{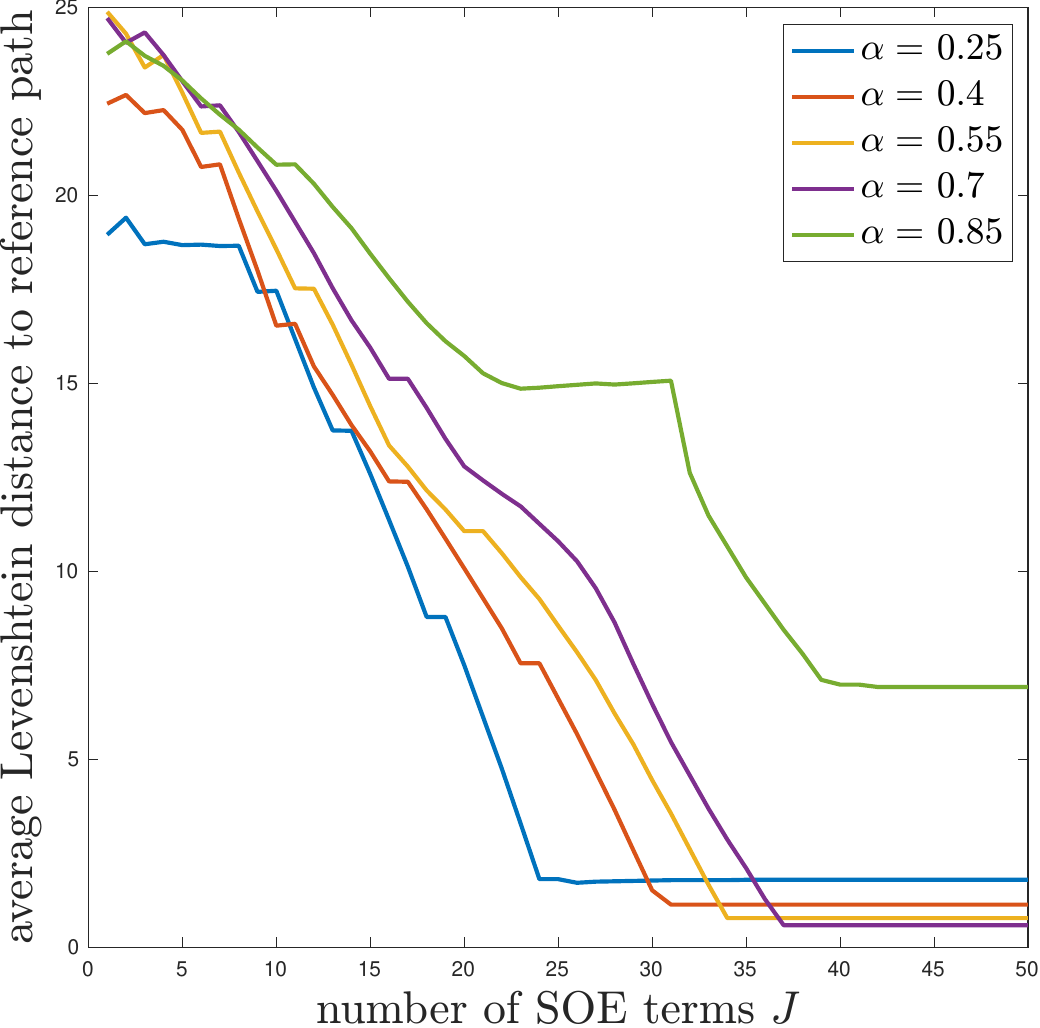}
\caption{Comparison of the average Levenshtein distance between the shortest
subdiffusive paths and the topological shortest path in the example Gabriel graph, for five
values of $\alpha$. The average is taken on 300 time instants in the
interval $0.1\protect\leq t\protect\leq1000$.}
\label{fig:SOE memory}
\end{figure}

\section{Caputo fractional derivative and memory}
\label{sec:remote-early-memory}

We have seen that memory dictates the behavior of subdiffusion and the underlying fractional-time differential equation. In this section, we analyze the memory contributions present in the Caputo derivative, distinguishing between the influence of the remote past and more recent times. We show an underlying connection between these effects and the convexity of the solution to the subdiffusion equation. 

\begin{defn}
	Consider the time-fractional Caputo derivative of a function $x(t)$ as defined in
	\S\ref{sec:prelim}. For odd $k$ let the interval $\left[0,t\right]$ be subdivided
	into $k$ subintervals $\left[t_{j},t_{j+1}\right]$ for $j=0,\ldots,k-1$
	of equal length $h=t/k$. Let us define the following specific subintervals in $\left[0,t\right]$: 
	
	\begin{itemize}
		\item Remote past: $\mathcal{R}\left(k,\alpha\right)=\dfrac{h^{1-\alpha}}{\varGamma\left(3-\alpha\right)}\left(\left(k-1\right)^{2-\alpha}-\left(k+\alpha-2\right)k^{1-\alpha}\right)x'\left(0\right);$
		\item Late past: 
		\[\mathcal{L_{R}}\left(k,\alpha\right)=\dfrac{h^{1-\alpha}}{\varGamma\left(3-\alpha\right)}\sum_{j=1}^{\left(k-1\right)/2}\left(\left(k-j+1\right)^{2-\alpha}-2\left(k-j\right)^{2-\alpha}+\left(k-j-1\right)^{2-\alpha}\right)x'\left(t_{j}\right);\]
		\item Early past: \[\mathcal{L_{P}}\left(k,\alpha\right)=\dfrac{h^{1-\alpha}}{\varGamma\left(3-\alpha\right)}\sum_{j=\left(k+1\right)/2}^{k-1}\left(\left(k-j+1\right)^{2-\alpha}-2\left(k-j\right)^{2-\alpha}+\left(k-j-1\right)^{2-\alpha}\right)x'\left(t_{j}\right);\]
		\item Present: 
		$\mathcal{P}\left(t\right)=\dfrac{h^{1-\alpha}}{\varGamma\left(3-\alpha\right)}x'\left(t\right).$
		
	\end{itemize}
\end{defn}

Then, we have the following result proved by Odibat \cite{odibat2006approximations}.
\begin{thm}
	The fractional Caputo derivative \textup{$D_{t}^{\alpha}x\left(t\right)$
		can be expressed as}
	\begin{equation}
		\begin{split}D_{t}^{\alpha}x\left(t\right) & =\mathcal{R}\left(k,\alpha\right)+\mathcal{L_{R}}\left(k,\alpha\right)+\mathcal{L_{P}}\left(k,\alpha\right)+\mathcal{P}\left(t\right)-E_{C}\left(f,h,\alpha\right)\end{split}
		,\label{eq:trapezoidal}
	\end{equation}
	where $E_{C}\left(f,h,\alpha\right)\leq\mathcal{O}\left(h^{2}\right)$
	is the error term.
\end{thm}

We have then the following result.
\begin{lem}
	Let $i\in V$ be a vertex of the graph $G.$ Then, in the limit $\alpha\to0$,
	the contributions for vertex $i$ of the different time subintervals
	take the values:
	
	\begin{align*}
		\mathcal{R}^{(i)}(k,0) & =\frac{h}{2}\,x_{i}'(0),\\[0.3em]
		\mathcal{L_{R}}^{(i)}(k,0) & =h\sum_{j=1}^{(k-1)/2}x_{i}'(t_{j}),\\
		\mathcal{L_{P}}^{(i)}(k,0) & =h\sum_{j=(k+1)/2}^{k-1}x_{i}'(t_{j}),\\
		\mathcal{P}^{(i)}(k,0) & =\frac{h}{2}\,x_{i}'(T).
	\end{align*}
	
	In the limit $\alpha\to1$, the only contribution which survives is
	$\mathcal{P}^{(i)}(k,\alpha\to1)=x_{i}'(T)$ and the rest vanish.
\end{lem}

\begin{defn}
	For a fixed vertex $i$, fixed $T>0$ and odd $k$, we say that vertex
	$i$ \emph{recalls its remote past more strongly than its recent past}
	on $[0,T]$ if 
	\[
	\mathcal{L_{R}}^{(i)}(k,0)>\mathcal{L_{P}}^{(i)}(k,0),
	\]
	and we say it \emph{recalls its recent past more strongly than its
		remote past} if 
	\[
	\mathcal{L_{R}}^{(i)}(k,0)<\mathcal{L_{P}}^{(i)}(k,0).
	\]
\end{defn}

\begin{thm}[Memory Bias]
	\label{thm:memory-bias} Let $x(t)=(x_{1}(t),\dots,x_{n}(t))^{\top}$
	be the solution of the fractional diffusion system with sufficiently
	smooth temporal evolution at each vertex. Fix a vertex $i\in V$,
	a time $t>0$, and an odd integer $k\ge3$ as above. Then,
	\begin{itemize}
		\item If the derivative $x'(t)$ is nonincreasing on $[0,t]$, then 
		\[
		\mathcal{L_{R}}^{(i)}(k,0)\;\ge\;\mathcal{L_{P}}^{(i)}(k,0)\quad\text{and}\quad\mathcal{R}^{(i)}(k,0)\;\ge\;\mathcal{P}^{(i)}(k,0),
		\]
		and both inequalities are strict if $x_{i}'(t)$ is strictly decreasing
		on $[0,t]$. In this case, vertex $i$ recalls its more remote recent
		past more strongly than its more recent past on $[0,t]$, and it also
		recalls its remote past more strongly than its present.
		
		\item If the derivative $x_{i}'(t)$ is nondecreasing on $[0,t]$, then
		\[
		\mathcal{L_{R}}^{(i)}(k,0)\;\le\;\mathcal{L_{P}}^{(i)}(k,0)\quad\text{and}\quad\mathcal{R}^{(i)}(k,0)\;\le\;\mathcal{P}^{(i)}(k,0),
		\]
		and both inequalities are strict if $x_{i}'(t)$ is strictly increasing
		on $[0,t]$. In this case, vertex $i$ recalls its more recent past
		more strongly than its more remote recent past on $[0,t]$, and it
		recalls its present more strongly than its remote past. 
	\end{itemize}
\end{thm}

\begin{proof}
	The proof is purely temporal and uses only the behavior of $x_{i}'(t)$
	on the interval $[0,t]$; the graph structure enters only through
	the fact that $x'_{i}(t)$ arises from the fractional diffusion equation.
	
	Write $k=2m+1$, so that $(k-1)/2=m$ and $k-1=2m$. For the first part, assume $x_{i}'(t)$
	is nonincreasing on $[0,t]$.  Let us compare $\mathcal{L_{R}}^{(i)}(k,0)$ and $\mathcal{L_{P}}^{(i)}(k,0)$. Define the index reflection $j'=k-j$
	for $j=1,\dots,m$. Then $j'\in\{m+1,\dots,2m\}$ and $t_{j}=jh<j'h=t_{j'}$.
	By monotonicity, $x'(t_{j})\ge x_{i}'(t_{j'})$. Reindexing the right
	sum via $j'=k-j$ yields 
	\[
	\sum_{j=m+1}^{2m}x_{i}'(t_{j})=\sum_{j=1}^{m}x_{i}'(t_{k-j})=\sum_{j=1}^{m}x_{i}'(t_{j'}),
	\]
	so that 
	\[
	\sum_{j=1}^{m}x_{i}'(t_{j})\;\ge\;\sum_{j=1}^{m}x_{i}'(t_{j'})=\sum_{j=m+1}^{2m}x_{i}'(t_{j}).
	\]
	Multiplying by $h>0$ yields $\mathcal{L_{R}}^{(i)}(k,0)\ge\mathcal{L_{P}}^{(i)}(k,0)$,
	with strict inequality if $x_{i}'$ is strictly decreasing.
	
	For the comparison between $\mathcal{R}^{(i)}(k,0)$ and $\mathcal{P}^{(i)}(k,0)$
	we use their explicit expressions 
	\[
	\mathcal{R}^{(i)}(k,0)=\frac{h}{2}x_{i}'(0),\qquad\mathcal{P}^{(i)}(k,0)=\frac{h}{2}x_{i}'(T).
	\]
	If $x_{i}'(t)$ is nonincreasing on $[0,t]$, then $x_{i}'(0)\ge x_{i}'(T)$,
	and multiplying by $h/2>0$ gives $\mathcal{R}^{(i)}(k,0)\ge\mathcal{P}^{(i)}(k,0)$,
	with strict inequality if $x_{i}'$ is strictly decreasing.
	
	For the second part, apply the first part to the function $-x_{i}(t)$. Denote by $\mathcal{L_{R}}^{(i)}_{-x_i}(k,0)$ the late past contribution, and similarly for the other time intervals.  
	If $x_{i}'(t)$ is nondecreasing on $[0,t]$, then $(-x_{i})'(t)=-x_{i}'(t)$
	is nonincreasing, hence 
	\[
	\mathcal{L_{R}}^{(i)}_{-x_i}(k,0)\;\ge\;\mathcal{L_{P}}^{(i)}_{-x_i}(k,0)\quad\text{and}\quad\mathcal{R}^{(i)}_{-x_i}(k,0)\;\ge\;\mathcal{P}^{(i)}_{-x_i}(k,0).
	\]
	By linearity of the sums defining $\mathcal{L_{R}}$, $\mathcal{L_{P}}$
	and by the explicit formulas for $\mathcal{R}$ and $\mathcal{P}$,
	\[
	\mathcal{L_{R}}^{(i)}_{-x_i}(k,0)=-\mathcal{L_{R}}^{(i)}(k,0),\quad\mathcal{L_{P}}^{(i)}_{-x_i}(k,0)=-\mathcal{L_{P}}^{(i)}(k,0),
	\]
	\[
	\mathcal{R}^{(i)}_{-x_i}(k,0)=-\mathcal{R}^{(i)}(k,0),\quad\mathcal{P}^{(i)}_{-x_i}(k,0)=-\mathcal{P}^{(i)}(k,0),
	\]
	and thus 
	\[
	-\mathcal{L_{R}}^{(i)}(k,0)\;\ge\;-\mathcal{L_{P}}^{(i)}(k,0),\qquad-\mathcal{R}^{(i)}(k,0)\;\ge\;-\mathcal{P}^{(i)}(k,0),
	\]
	which is equivalent to 
	\[
	\mathcal{L_{R}}^{(i)}(k,0)\;\le\;\mathcal{L_{P}}^{(i)}(k,0),\qquad\mathcal{R}^{(i)}(k,0)\;\le\;\mathcal{P}^{(i)}(k,0),
	\]
	with strict inequalities if $x_{i}'$ is strictly increasing. 
\end{proof}
\begin{cor}
	\label{cor:global-past-present} Under the assumptions of Theorem~\ref{thm:memory-bias},
	define the total ``past'' and ``present'' contributions at vertex
	$i$ by 
	\[
	\mathcal{T}_{\mathrm{past}}^{(i)}(k,0):=\mathcal{R}^{(i)}(k,0)+\mathcal{L_{R}}^{(i)}(k,0),\qquad\mathcal{T}_{\mathrm{pres}}^{(i)}(k,0):=\mathcal{L_{P}}^{(i)}(k,0)+\mathcal{P}^{(i)}(k,0).
	\]
	If $x_{i}'(t)$ is nonincreasing on $[0,t]$, then 
	\[
	\mathcal{T}_{\mathrm{past}}^{(i)}(k,0)\;\ge\;\mathcal{T}_{\mathrm{pres}}^{(i)}(k,0),
	\]
	with strict inequality if $x_{i}'$ is strictly decreasing. If $x_{i}'(t)$
	is nondecreasing on $[0,t]$, then 
	\[
	\mathcal{T}_{\mathrm{past}}^{(i)}(k,0)\;\le\;\mathcal{T}_{\mathrm{pres}}^{(i)}(k,0),
	\]
	with strict inequality if $x_{i}'$ is strictly increasing. 
\end{cor}

\begin{proof}
	In the nonincreasing case, Theorem~\ref{thm:memory-bias} yields
	$\mathcal{R}^{(i)}\ge\mathcal{P}^{(i)}$ and $\mathcal{L_{R}}^{(i)}\ge\mathcal{L_{P}}^{(i)}$,
	so adding the inequalities gives 
	\[
	\mathcal{R}^{(i)}+\mathcal{L_{R}}^{(i)}\;\ge\;\mathcal{L_{P}}^{(i)}+\mathcal{P}^{(i)}.
	\]
	The nondecreasing case is analogous, using the reversed inequalities
	from Theorem~\ref{thm:memory-bias}. 
\end{proof}
\begin{cor}
	\label{cor:same-vertex-two-windows} Let $x(t)$ be as in the Caputo
	fractional diffusion equation, and fix a vertex $i\in V$. Let $[t_{1},t_{2}]$
	and $[t_{3},t_{4}]$ be two disjoint time intervals with $0\le t_{1}<t_{2}\le t_{3}<t_{4}$.
	Fix an odd integer $k\ge3$ and define uniform grids on each interval:
	\[
	h_{1}=\frac{t_{2}-t_{1}}{k},\quad s_{j}^{(1)}=t_{1}+jh_{1},\quad j=0,\dots,k,
	\]
	\[
	h_{2}=\frac{t_{4}-t_{3}}{k},\quad s_{j}^{(2)}=t_{3}+jh_{2},\quad j=0,\dots,k.
	\]
	Define the corresponding left and right recent contributions: 
	\begin{align*}
		\mathcal{L_{R}}^{(i,1)}(k,0) & :=h_{1}\sum_{j=1}^{(k-1)/2}x_{i}'\bigl(s_{j}^{(1)}\bigr), & \mathcal{L_{P}}^{(i,1)}(k,0) & :=h_{1}\sum_{j=(k+1)/2}^{k-1}x_{i}'\bigl(s_{j}^{(1)}\bigr),\\[0.3em]
		\mathcal{L_{R}}^{(i,2)}(k,0) & :=h_{2}\sum_{j=1}^{(k-1)/2}x_{i}'\bigl(s_{j}^{(2)}\bigr), & \mathcal{L_{P}}^{(i,2)}(k,0) & :=h_{2}\sum_{j=(k+1)/2}^{k-1}x_{i}'\bigl(s_{j}^{(2)}\bigr).
	\end{align*}
	Assume that
	
	$x_{i}'(t)$ is strictly decreasing on $[t_{1},t_{2}]$,
	
	$x_{i}'(t)$ is strictly increasing on $[t_{3},t_{4}]$.
	
	Then 
	\[
	\mathcal{L_{R}}^{(i,1)}(k,0)\;>\;\mathcal{L_{P}}^{(i,1)}(k,0),\qquad\mathcal{L_{R}}^{(i,2)}(k,0)\;<\;\mathcal{L_{P}}^{(i,2)}(k,0).
	\]
	In particular, on $[t_{1},t_{2}]$ vertex $i$ recalls its more remote
	recent past more strongly than its more recent past, whereas on $[t_{3},t_{4}]$
	the bias is reversed. 
\end{cor}

\begin{proof}
	Define $z_{1}(s)=x_{i}(t_{1}+s)$ on $[0,t_{2}-t_{1}]$ and $z_{2}(s)=x_{i}(t_{3}+s)$
	on $[0,t_{4}-t_{3}]$. Then $z_{1}'$ is strictly decreasing and $z_{2}'$
	is strictly increasing on their respective domains, and the points
	$s_{j}^{(1)}$, $s_{j}^{(2)}$ define uniform grids. Applying Theorem~\ref{thm:memory-bias}
	to $z_{1}$ and $z_{2}$ yields the desired inequalities. 
\end{proof}
\begin{cor}
	\label{cor:two-vertices-same-interval} Let $x(t)$ be as before,
	and consider two vertices $i,j\in V$. Fix a time $t>0$ and an odd
	integer $k\ge3$, and define $h=t/k$, $t_{\ell}=\ell h$ for $\ell=0,\dots,k$.
	Define 
	\begin{align*}
		\mathcal{L_{R}}^{(i)}(k,0) & :=h\sum_{\ell=1}^{(k-1)/2}x_{i}'(t_{\ell}), & \mathcal{L_{P}}^{(i)}(k,0) & :=h\sum_{\ell=(k+1)/2}^{k-1}x_{i}'(t_{\ell}),\\[0.3em]
		\mathcal{L_{R}}^{(j)}(k,0) & :=h\sum_{\ell=1}^{(k-1)/2}x_{j}'(t_{\ell}), & \mathcal{L_{P}}^{(j)}(k,0) & :=h\sum_{\ell=(k+1)/2}^{k-1}x_{j}'(t_{\ell}).
	\end{align*}
	Assume that
	
	$x_{i}'(t)$ is strictly decreasing on $[0,t]$,
	
	$x_{j}'(t)$ is strictly increasing on $[0,t]$.
	
	Then 
	\[
	\mathcal{L_{R}}^{(i)}(k,0)\;>\;\mathcal{L_{P}}^{(i)}(k,0),\qquad\mathcal{L_{R}}^{(j)}(k,0)\;<\;\mathcal{L_{P}}^{(j)}(k,0).
	\]
	Thus, over the same time interval $[0,t]$ and for the same fractional
	model, vertex $i$ exhibits a remote-past memory bias, whereas vertex
	$j$ exhibits a recent-past memory bias. 
\end{cor}

\begin{proof}
	Apply Theorem~\ref{thm:memory-bias} to the scalar functions $t\mapsto x_{i}(t)$
	and $t\mapsto x_{j}(t)$ on $[0,t]$. 
\end{proof}

\subsection{Physical Interpretation: Memory Regimes on a Graph}

The Memory Bias Theorem shows that, within the Caputo--driven fractional
diffusion dynamics on a graph for $\alpha\rightarrow0$, the relative
importance of the ``remote'' versus ``recent'' parts of the past
is not uniform across the network. Instead, it depends sensitively
on the \emph{local temporal curvature} of the solution at each vertex.

If the temporal derivative $x_{i}'(t)$ is decreasing on a given window,
the trajectory $x_{i}(t)$ is bending downward, and older information
within that window receives a larger weight than more recent information.
In this regime, vertex $i$ is said to exhibit a \emph{remote--past
	memory bias}. Our results show that, in this case, not only the more
remote part of the recent past dominates the more recent part, but
the remote past also dominates the present when $\alpha\to0$.

Conversely, if $x_{i}'(t)$ is increasing on a given window, the trajectory
is bending upward, and the more recent information receives a larger
weight than the more remote information. The vertex then exhibits
a \emph{recent--past memory bias}, and in the $\alpha\to0$ limit
the present dominates the remote past as well.

A striking consequence is that different vertices of the same graph
may simultaneously reside in opposite memory regimes, even though
they are driven by the same fractional dynamics and are evaluated
over the same time interval. Similarly, the \emph{same} vertex may
switch memory regimes over time, depending on the evolution of its
temporal curvature. This expresses a fundamental ``heterogeneity
of memory'' in fractional diffusion on graphs.

In summary, the fractional order determines \emph{how much} of the
past is remembered globally, but the \emph{shape of the temporal evolution}
at each vertex dictates \emph{which} part of the past (remote or recent,
and past versus present) is preferentially recalled. This creates
rich, spatially distributed memory patterns that reflect both the
graph geometry and the initial configuration.

\section{How memory emerges from fractional diffusion on a graph?}
\label{sec:memory-emergence}

Once we have seen how convexity influences recall of early or late past, we now study how convexity can emerge in a network, depending on the initial mass distribution of the process. The strong influence of memory is also shown through the underlying random clock process. 

\subsection{Local Convexity and Concavity in Caputo Fractional Diffusion: A Mittag-Leffler
	and SOE-Based Analysis}

We consider the Caputo fractional diffusion equation on a finite graph
$G$ with combinatorial Laplacian $L=D-A$: 
\[
D_{t}^{\alpha}x(t)=-Lx(t),\qquad x(0)=e_{v},
\]
where $0<\alpha<1$ and $e_{v}$ is the $v$th standard basis vector.
The mild solution is given by the Mittag-Leffler matrix function
\[
x(t)=E_{\alpha}(-t^{\alpha}L)\,e_{v}.
\]

We are interested in the \textit{curvature} of the time evolution
at early times, i.e., the sign of $x_{i}''(t)$ for $t>0$ sufficiently
small. This sign determines whether the temporal trajectory at a vertex
is locally convex or concave. We show below that, regardless of the
graph, the vertex where the mass is initially placed has an early-time
\textit{convex and decreasing} profile, while every neighbor exhibits
an early-time \textit{concave and increasing} profile. Both statements
hold rigorously for every $0<\alpha<1$.
\begin{thm}
	[Local convexity/concavity from Mittag-Leffler and SOE] \label{thm:local-conv-conc}
	Let $G$ be a finite graph with Laplacian $L=D-A$. Consider the Caputo
	fractional diffusion equation 
	\[
	D_{t}^{\alpha}x(t)=-Lx(t),\qquad0<\alpha<1,\qquad x(0)=e_{v}.
	\]
	Then:
	\begin{itemize}
		\item At vertex $v$, the time evolution $x(t)$ is strictly decreasing
		and strictly convex for all $t\in(0,\varepsilon_{v})$, for some $\varepsilon_{v}>0$.
		\item If $w$ is any neighbor of $v$, i.e.\ $A_{wv}>0$, then $u_{w}(t)$
		is strictly increasing and strictly concave for all $t\in(0,\varepsilon_{w})$,
		for some $\varepsilon_{w}>0$.
	\end{itemize}
	In particular, there exists $\varepsilon>0$ such that for all $t\in(0,\varepsilon)$
	the excited vertex $v$ exhibits a convex decay, while each of its
	neighbors exhibits a concave growth. 
\end{thm}

\begin{proof}
	We use two complementary arguments: (i)~the small-time expansion
	of the Mittag-Leffler matrix function, and (ii)~a Sum-of-Exponentials
	(SOE) representation for $E_{\alpha}$.
	
	\medskip{}
	\textbf{Step 1: Exact short-time Mittag-Leffler expansion.} The
	matrix Mittag-Leffler function admits the convergent series 
	\[
	E_{\alpha}(-t^{\alpha}L)=I-\frac{t^{\alpha}}{\Gamma(\alpha+1)}L+\frac{t^{2\alpha}}{\Gamma(2\alpha+1)}L^{2}+O(t^{3\alpha}).
	\]
	Applying this to the initial condition $x(t)=E_{\alpha}(-t^{\alpha}L)e_{v}$
	gives 
	\[
	x(t)=e_{v}-\frac{t^{\alpha}}{\Gamma(\alpha+1)}Le_{v}+O(t^{2\alpha}).
	\]
	
	\noindent At the excited vertex $v$ we have 
	\[
	x_{v}(t)=1-\frac{\deg(v)}{\Gamma(\alpha+1)}\,t^{\alpha}+O(t^{2\alpha}),
	\]
	because $(Le_{v})_{v}=\deg(v)$. Differentiating yields 
	\[
	x_{v}'(t)=-\frac{\deg(v)}{\Gamma(\alpha+1)}\alpha t^{\alpha-1}+O(t^{2\alpha-1})<0,
	\]
	\[
	x_{v}''(t)=-\frac{\deg(v)}{\Gamma(\alpha+1)}\alpha(\alpha-1)t^{\alpha-2}+O(t^{2\alpha-2})>0,
	\]
	since $\alpha(\alpha-1)<0$ for $0<\alpha<1$.  Note that the exponents $\alpha-1$, $2\alpha-2$ are negative. Thus $x_{v}(t)$ is
	strictly decreasing and strictly convex on $(0,\varepsilon_{v})$.
	
	At a neighbor $w$ of $v$, 
	\[
	x_{w}(t)=-\frac{t^{\alpha}}{\Gamma(\alpha+1)}(Le_{v})_{w}+O(t^{2\alpha})=\frac{A_{wv}}{\Gamma(\alpha+1)}\,t^{\alpha}+O(t^{2\alpha}),
	\]
	because $(Le_{v})_{w}=-A_{wv}$ for $w\neq v$. Thus 
	\[
	x_{w}'(t)=\frac{A_{wv}}{\Gamma(\alpha+1)}\alpha t^{\alpha-1}+O(t^{2\alpha-1})>0,
	\]
	\[
	x_{w}''(t)=\frac{A_{wv}}{\Gamma(\alpha+1)}\alpha(\alpha-1)t^{\alpha-2}+O(t^{2\alpha-2})<0.
	\]
	Hence $x_{w}(t)$ is strictly increasing and strictly concave on $(0,\varepsilon_{w})$.
	
	\medskip{}
	\textbf{Step 2: Interpretation via SOE approximation.} We have developed
	a SOE scheme to approximate the Mittag-Leffler function by 
	\[
	E_{\alpha}(-\lambda t^{\alpha})\approx\sum_{m=1}^{M}c_{m}e^{-d_{m}t},\qquad c_{m}>0,\ d_{m}>0.
	\]
	Applying the same SOE to the matrix $L$ yields 
	\[
	x(t)=E_{\alpha}(-t^{\alpha}L)e_{v}\;\approx\;\sum_{m=1}^{M}c_{m}e^{-d_{m}tL}e_{v}.
	\]
	Each term $e^{-d_{m}tL}e_{v}$ is the solution of a \emph{classical}
	diffusion equation with rate~$d_{m}$ and is therefore strictly convex
	at the source $v$ and strictly concave at neighbors $w$ at early
	times. Because all coefficients $c_{m}$ are positive, the SOE sum
	preserves the convexity at $v$ and concavity at its neighbors. This
	matches exactly the signs obtained in the rigorous Mittag-Leffler
	expansion above.
	
	\medskip{}
	\noindent Combining Steps~1 and~2 proves (i) and (ii). 
	
	Thus, over the same time interval $t\in(0,\varepsilon_{v})$ and for
	the same fractional model, vertex $v$ exhibits a remote-past memory
	bias, i.e., it recalls more the past than the present, whereas vertex
	$w$ exhibits a recent-past memory bias, i.e., it recalls more the
	present than the past. 
\end{proof}

\subsection{Random time change induces subdiffusion and memory}

\label{sec:subdiffusion-memory}

We summarize the mechanism by which the random clock $E_{t}$ slows
down spreading (\emph{subdiffusion}) and introduces \emph{memory}
into the dynamics on a graph. Under the subordination identity, 
\[
u(t)=E_{\alpha}(-t^{\alpha}L)\,u_{0}\;=\;\mathbb{E}\!\left[e^{-E_{t}L}\right]u_{0},
\]
the state at physical time $t$ equals the baseline heat state evaluated
at the \emph{random} operational time $s=E_{t}$. For $0<\alpha<1$,
the inverse $\alpha$-stable clock satisfies 
\[
\mathbb{E}[E_{t}]=\frac{t^{\alpha}}{\Gamma(1+\alpha)},\qquad\mathrm{Var}(E_{t})=\frac{2\,t^{2\alpha}}{\Gamma(1+2\alpha)}-\frac{t^{2\alpha}}{\Gamma(1+\alpha)^{2}},
\]
so the \emph{typical} amount of diffusion time available by physical
time $t$ scales like $t^{\alpha}$ (rather than $t$). Thus any diffusive
spread measure that, at baseline, scales with $s$ (e.g., mean-square
displacement in Euclidean space, or mixing surrogates on graphs) will
scale like $t^{\alpha}$ under the random clock. This is the essence
of \emph{subdiffusion}: slower-than-classical spreading ($t^{\alpha/2}$
instead of $t^{1/2}$ in continuum settings), and on graphs a slower
homogenization than the exponential-in-$t$ decay of the standard
heat semigroup.

In the time-changed walk $Y_{t}=X_{E_{t}}$, the holding time $T_i$ at node
$i$ has survival function 
\[
\mathbb{P}\!\left(T_i>t\right)=E_{\alpha}(-d_{i}\,t^{\alpha}),
\]
with $d_{i}=\sum_{j}A_{ij}$. As $t\to\infty$, the Mittag-Leffler
tail obeys $E_{\alpha}(-d_{i}\,t^{\alpha})\sim\dfrac{1}{d_{i}\,\Gamma(1-\alpha)}\,t^{-\alpha},$
a power law (no finite mean inter-jump time when $\alpha<1$). These
\emph{rare but very long} pauses act as traps that stretch physical
time relative to operational time, producing subdiffusive spread on
the graph.

The Caputo equation 
\[
\partial_{t}^{\alpha}u(t)=-L\,u(t),\qquad\partial_{t}^{\alpha}u(t)=\frac{1}{\Gamma(1-\alpha)}\int_{0}^{t}(t-\tau)^{-\alpha}\,u'(\tau)\,d\tau,
\]
is explicitly \emph{history dependent}: the instantaneous rate $u'(t)$
depends on the full past with a power-law kernel $(t-\tau)^{-\alpha}$.
In renewal terms, the holding-time hazard for Mittag-Leffler waiting
times is \emph{decreasing in age} the longer the process has been
waiting, the less likely it is to jump immediately. Hence the future
depends on how long the current wait has lasted (``aging''), which
breaks the Markov property in physical time $t$. Conditioned on $E_{t}=s$,
the baseline path is Markov; \emph{after averaging over} the random
clock, the observed dynamics inherit memory.

On the other hand, with $L=V\Lambda V^{\top}$, each mode decays as
\[
E_{\alpha}(-t^{\alpha}\lambda_{k})\;\text{instead of}\;e^{-t\lambda_{k}}.
\]
For $t\to\infty$, $E_{\alpha}(-t^{\alpha}\lambda_{k})\sim\dfrac{1}{\lambda_{k}\,\Gamma(1-\alpha)}\,t^{-\alpha}$:
\emph{algebraic} decay replaces exponential decay. Hence mixing, return
probabilities, and any observable built from the heat kernel trace
decay more slowly, reflecting both subdiffusion (slower spread) and
long memory (long tails). Practical consequences on networks. 
\begin{itemize}
	\item \emph{Slower homogenization:} community imbalances, gradients, or
	initial heterogeneities persist longer (power-law tail). 
	\item \emph{Trapping in dense/central regions:} large $d_{i}$ increases
	the attempt rate, but the heavy-tailed clock still produces long residence
	episodes; dwell-time distributions are broad across nodes. 
	\item \emph{Non-exponential relaxation:} observables fit Mittag-Leffler
	or power-law decays rather than exponentials; log--log slopes reveal
	$\alpha$. 
\end{itemize}
In summary, the random clock $E_{t}$ \emph{compresses} operational
time from $t$ to $t^{\alpha}$ on average and introduces \emph{power-law
	waiting} between moves; together these yield \emph{subdiffusive spreading}
and \emph{memory} (aging) in the network dynamics.

\section{A generalized physico-mathematical context of subdiffusion on graphs}
\label{sec:differential-equations}

The fractional-time differential equation that give rise to subdiffusion is widely known. The SOE it is not only the approximation of a function: we show how the superposition of diffusions can emerge in a graph, providing a bridge to graph dynamics, as well as integrating the concepts in the generalized structure of kernel-based differential equations. 

\subsection{A multiplex diffusion}

Let us consider that there are $J$ parallel diffusion processes occurring
on the graph $G$, each of them having a diffusivity constant $\beta_{j}$.
Therefore, we can consider a representation of the graph as a multiplex
formed by $h$ layers \cite{kivela2014multilayer,boccaletti2014structure},
each of them representing the same graph. Every pair of layers $l_{i}$
and $l_{j}$ are interconnected by mean of extra edges from a vertex
in one layer $v\in l_{i}$ to itself in another layer $v\in l_{j}$. 

Then, we create the super-Laplacian matrix, which was introduced in
\cite{gomez2013diffusion}: 

\begin{equation}
\mathfrak{L}=\oplus_{j=1}^{J}\beta_{j}L+\omega I\otimes\left(J-I\right),
\end{equation}
where $\otimes$ is the Kronecker product, $\omega\in\mathbb{R}^{+}\cup\left\{ 0\right\} $
is the strength coupling between the layers, $J$ is an all-ones matrix
and $I$ is the identity matrix. Then, 

\begin{equation}
\mathfrak{L}=\left(\begin{array}{cccc}
\beta_{1}L & C_{12} & \cdots & C_{1J}\\
C_{21} & \beta_{2}L & \cdots & C_{2J}\\
\vdots & \vdots & \ddots & \vdots\\
C_{J1} & C_{J2} & \cdots & \beta_{J}L
\end{array}\right).
\end{equation}

Let us write the diffusion equation on the multiplex:

\begin{equation}
\partial_{t}u\left(t\right)+\mathfrak{L}u\left(t\right)=0,
\end{equation}
with initial condition $u\left(0\right)=\gamma\otimes\varphi$ where
$\gamma=\left[\gamma_{1}\cdots\gamma_{J}\right]^{T},$ $\gamma_{j}\in\mathbb{R}$
and $\varphi\in\mathbb{R}^{n\times1}$. The solution of the abstract
Cauchy problem is then

\begin{equation}
u\left(t\right)=e^{-t\mathfrak{L}}\,u(0).
\end{equation}

Let us consider a very weak coupling strength between pairs of layers,
$\omega\ll1$, such that we can consider the diffusion at each layer
almost independently of the diffusion on other layers:

\begin{equation}
u\left(t\right)\cong\oplus_{j=1}^{J}\gamma_{j}e^{-t\beta_{j}L}\varphi.
\end{equation}

We then can consider the sum of the concentrations at each vertex
of the graph in all layers:

\begin{equation}
c\left(t\right)=\sum_{j=1}^{h}u_{j}\left(t\right)=\sum_{j=1}^{J}\gamma_{j}e^{-t\beta_{j}L}\varphi,
\end{equation}
which is exactly the approximation we have found for the Mittag-Leffler
Laplacian function as a SOE. Therefore, the subdiffusive process on
the graph can be seen as the total diffusion among the vertices of
a graph occurring in parallel at different layers with their own diffusivities
and with initial conditions in one layer proportional to the others.

\subsection{Factorized high-order temporal diffusion equation}

Let us consider $h$ independent diffusive species $v_{1}(t),\dots,v_{J}(t)\in\mathbb{R}^{n}$
evolving according to 
\begin{equation}
\partial_{t}v_{j}(t)+\beta_{j}Lv_{j}(t)=0,\qquad j=1,\dots,h,\label{eq:multi-species}
\end{equation}
with distinct diffusion rates $\beta_{j}>0$. Define the observable
field 
\begin{equation}
u(t):=\sum_{j=1}^{J}v_{j}(t).\label{eq:observable}
\end{equation}

Since $L$ is diagonalizable, it suffices to work in a single Laplacian
eigenmode with eigenvalue $\lambda\ge0$ and corresponding eigenvector $z_\lambda \in \mathbb{R}^n$. 
Let $x_{j}(t)$ be the scalar component of $v_j(t)$ relative to 
the eigenmode $\lambda$; that is, $v_j(t) = \sum_{\lambda\in\mathrm{Sp}(L)} x_j(t) \, z_\lambda$. Denote $y(t)$ as $y(t)=\sum_{j=1}^{J}x_{j}(t)$. The system \eqref{eq:multi-species} reduces to 
\[
\partial_{t}x_{j}(t)+\beta_{j}\lambda x_{j}(t)=0,\qquad j=1,\dots,h,
\]
whose solutions are $x_{j}(t)=e^{-\beta_{j}\lambda t}x_{j}(0)$. Hence 
\[
y(t)=\sum_{j=1}^{J}e^{-\beta_{j}\lambda t}x_{j}(0).
\]
Each exponential $e^{-\beta_{j}\lambda t}$ is annihilated by the
operator $(\partial_{t}+\beta_{j}\lambda)$, and therefore 
\[
\prod_{m=1}^{J}(\partial_{t}+\beta_{m}\lambda)\,y(t)=0.
\]
Since this holds for every Laplacian eigenvalue $\lambda$, we obtain,
in operator form, the \emph{multiplicative diffusion equation} 
\begin{equation}
\prod_{m=1}^{J}\bigl(\partial_{t}+\beta_{m}L\bigr)\,u(t)\;=\;0,\label{eq:multiplicative}
\end{equation}
with suitable initial conditions. 

Then, we claim that 
\[
u(t)\;=\;\sum_{j=1}^{j}\gamma_{j}\,e^{-t\beta_{j}L}\,\varphi.
\]
is an exact solution of the \emph{multiplicative diffusion equation.}

\subsubsection{Verification}

It suffices to check that each term 
\[
u_{j}(t):=e^{-t\beta_{j}L}\,\varphi
\]
lies in the kernel of the operator $\prod_{m=1}^{J}(\partial_{t}+\beta_{m}L)$.
First observe that for each fixed $m$ and $j$, 
\[
(\partial_{t}+\beta_{m}L)\,u_{j}(t)=\bigl(-\beta_{j}L+\beta_{m}L\bigr)e^{-t\beta_{j}L}\,\varphi=(\beta_{m}-\beta_{j})\,L\,e^{-t\beta_{j}L}\,\varphi.
\]
Applying the full product, we get 
\[
\prod_{m=1}^{J}(\partial_{t}+\beta_{m}L)\,u_{j}(t)=\left(\prod_{m=1}^{J}(\beta_{m}-\beta_{j})\right)L^{J}e^{-t\beta_{j}L}\,\varphi.
\]
If the $\beta_{m}$ are pairwise distinct and $j\in\{1,\dots,h\}$,
then one of the factors in the product is $(\beta_{j}-\beta_{j})=0$.
Hence 
\[
\prod_{m=1}^{J}(\partial_{t}+\beta_{m}L)\,u_{j}(t)=0\quad\text{for each }j=1,\dots,h.
\]
By linearity, 
\[
\prod_{m=1}^{J}(\partial_{t}+\beta_{m}L)\,u(t)=\sum_{j=1}^{J}\gamma_{j}\prod_{m=1}^{J}(\partial_{t}+\beta_{m}L)\,u_{j}(t)=0.
\]

Thus, the function 
\[
u(t)\;=\;\sum_{j=1}^{J}\gamma_{j}\,e^{-t\beta_{j}L}\,\varphi
\]
is an exact solution of \eqref{eq:multiplicative}.

\subsubsection{Initial conditions}

Equation \eqref{eq:multiplicative} is of order $h$ in time, so one
can prescribe $u(0),\partial_{t}u(0),\dots,\partial_{t}^{h-1}u(0)$.
The representation 
\[
u(t)=\sum_{j=1}^{J}e^{-t\beta_{j}L}\,\psi_{j}
\]
is the general solution of \eqref{eq:multiplicative}, and the vectors
$\psi_{j}$ are uniquely determined from the initial data via a Vandermonde-type
linear system (mode by mode in the eigenbasis of $L$). Choosing 
\[
\psi_{j}=\gamma_{j}\varphi
\]
gives precisely 
\[
u(t)=\sum_{j=1}^{J}\gamma_{j}e^{-t\beta_{j}L}\,\varphi.
\]

Therefore, for any prescribed coefficients $\beta_{j}$, $\gamma_{j}$
and profile $\varphi$, there exist initial conditions for \eqref{eq:multiplicative}
such that the unique solution is exactly 
\[
u(t)=\sum_{j=1}^{J}\gamma_{j}e^{-t\beta_{j}L}\,\varphi.
\]

\subsubsection{Physical interpretation.}

Equation \eqref{eq:multiplicative} does not describe diffusion of
a single species. Rather, it is the \emph{effective evolution equation}
satisfied by the total concentration $u$ obtained by summing $k$
independent diffusive species with distinct diffusion rates. The product
structure arises from eliminating the hidden fields $v_{j}$ and encodes
the presence of multiple diffusive time scales. Thus, the multiplicative
diffusion equation \eqref{eq:multiplicative} represents the minimal
closed description of a multi-rate diffusion process when only the
aggregate observable is accessible. 

\subsection{Operator-valued memory and diffusion on graphs}

\subsubsection{Operator-valued Volterra kernels}

Let $X$ be a finite-dimensional Hilbert space and denote by $\mathcal{L}(X)$
the space of bounded linear operators on $X$. An \emph{operator-valued
Volterra kernel} is a strongly measurable mapping 
\[
K:\mathbb{R}_{+}\to\mathcal{L}(X)
\]
such that 
\[
\int_{0}^{T}\|K(t)\|_{\mathcal{L}(X)}\,dt<\infty\quad\text{for all }T>0.
\]
Given such a kernel, the associated Volterra convolution operator
acts on sufficiently regular functions $u:\mathbb{R}_{+}\to X$ by
\[
(\mathcal{K}u)(t):=\int_{0}^{t}K(t-s)\,u(s)\,ds,\qquad t\ge0.
\]
An evolution equation of the form 
\begin{equation}
u'(t)+\mathcal{K}u(t)=0\label{eq:Volterra-1}
\end{equation}
is called a \emph{Volterra evolution equation with operator-valued
memory}. The Volterra structure enforces causality, while the operator-valued
nature of $K$ allows for mode-dependent memory effects.

\medskip{}

\begin{rem}
Equation \eqref{eq:Volterra-1} describes memory acting directly on
the state $u$. Caputo-type formulations, in which the kernel acts
on $u'$, will arise later as singular limits of this general framework.
\end{rem}

\subsubsection{Multiplicative diffusion on graphs}

\noindent Let $L\in\mathbb{R}^{n\times n}$ be a symmetric graph Laplacian
($L\succeq0$ and $\ker L=\operatorname{span}\{\mathbf{1}\}$), and
let $\beta_{1},\dots,\beta_{k}>0$ be pairwise distinct. We consider
the \emph{multiplicative diffusion equation} 
\begin{equation}
\prod_{j=1}^{J}\bigl(\partial_{t}+\beta_{j}L\bigr)\,u(t)=0,\qquad t>0,\label{eq:mult-1}
\end{equation}
supplemented with $k$ initial conditions $u(0),u'(0),\dots,u^{(k-1)}(0)$.

Solutions of \eqref{eq:mult-1} admit the representation 
\begin{equation}
u(t)=\sum_{j=1}^{J}e^{-\beta_{j}tL}\,\psi_{j},\label{eq:sumsemigroup-1}
\end{equation}
where the vectors $\psi_{j}\in\mathbb{R}^{n}$ are uniquely determined
by the initial conditions. Consequently, the Laplace transform $\widehat{u}(s)=\mathcal{L}\{u\}(s)$
satisfies 
\begin{equation}
\widehat{u}(s)=\sum_{j=1}^{J}(sI+\beta_{j}L)^{-1}\,\psi_{j},\qquad u(0)=\sum_{j=1}^{J}\psi_{j}.\label{eq:uLap-1}
\end{equation}

\subsubsection{Equivalence with an operator-valued memory equation}

We now relate \eqref{eq:mult-1} to a Volterra evolution equation
of the form 
\begin{equation}
u'(t)+\int_{0}^{t}K(t-s)\,L\,u(s)\,ds=0,\qquad t>0,\label{eq:memory-1}
\end{equation}
where $K:\mathbb{R}_{+}\to\mathcal{L}(\mathbb{R}^{n})$ is an operator-valued
Volterra kernel.

Taking the Laplace transform of \eqref{eq:memory-1} yields 
\begin{equation}
\bigl(sI+\widehat{K}(s)L\bigr)\,\widehat{u}(s)=u(0),\label{eq:memoryLap-1}
\end{equation}
for $\Re(s)$ sufficiently large.

Assume for simplicity that $u(0)\in\operatorname{Ran}(L)$, so that
$L$ is invertible on the dynamically relevant subspace.\footnote{More generally, $L^{-1}$ may be replaced by the Moore--Penrose inverse
$L^{\dagger}$.} Define 
\begin{equation}
G^{(J)}(s):=\sum_{j=1}^{J}(sI+\beta_{j}L)^{-1}P_{j},\label{eq:G-1}
\end{equation}
where $P_{j}:\mathbb{R}^{n}\to\mathbb{R}^{n}$ are linear maps such
that $\psi_{j}=P_{j}u(0)$. If $G^{(J)}(s)$ is invertible for $\Re(s)$
sufficiently large, define 
\begin{equation}
\widehat{K}(s):=L^{-1}\bigl(G^{(J)}(s)^{-1}-sI\bigr).\label{eq:KhatDef-1}
\end{equation}
Then \eqref{eq:uLap-1} implies 
\[
\bigl(sI+\widehat{K}(s)L\bigr)\widehat{u}(s)=u(0),
\]
which coincides with \eqref{eq:memoryLap-1}. Hence, solutions of
the multiplicative diffusion equation \eqref{eq:mult-1} also satisfy
the Volterra equation \eqref{eq:memory-1} with operator-valued kernel
$K$.

\subsubsection{Spectral representation}

Let $L=U\Lambda U^{\top}$ with $\Lambda=\operatorname{diag}(\lambda_{\ell})$.
On each eigenmode $\lambda>0$, the kernel is characterized by 
\begin{equation}
\widehat{K}_{\lambda}(s)=\frac{1}{\lambda}\left(\frac{1}{\sum_{j=1}^{J}\frac{p_{j}(\lambda)}{s+\beta_{j}\lambda}}-s\right),\label{eq:KhatMode-1}
\end{equation}
where $p_{j}(\lambda)$ are the scalar multipliers induced by $P_{j}$.
Thus, the memory kernel is generally mode-dependent and cannot be
reduced to a single scalar kernel.

\subsubsection{Fractional diffusion as a singular memory limit}
In the next proposition we show that the fractional diffusion equation is equivalent to a Volterra equation, for a suitably chosen subdiffusion kernel $K_S$.
\begin{prop}
\label{lem:Caputo_from_Volterra} Let $L\in\mathbb{R}^{n\times n}$
be a symmetric, positive semidefinite operator (e.g.\ a graph Laplacian),
and let $0<\alpha<1$. Consider the Volterra evolution equation 
\begin{equation}
u'(t)+\int_{0}^{t}K(t-s)\,L\,u(s)\,ds=0,\qquad t>0,\label{eq:VolterraLemma}
\end{equation}
with initial condition $u(0)=u_{0}$. If the kernel is chosen as 
\begin{equation}
K_{S}(t)=\frac{1}{\Gamma(\alpha-1)}\,t^{\alpha-2},\qquad t>0,\label{eq:kernelS}
\end{equation}
then \eqref{eq:VolterraLemma} is equivalent to the Caputo fractional
diffusion equation 
\begin{equation}
D_{t}^{\alpha}u(t)=-L\,u(t),\qquad t>0,\label{eq:CaputoLemma}
\end{equation}
where $D_{t}^{\alpha}$ denotes the Caputo derivative of order $\alpha$.
Moreover, the unique solution is given by 
\begin{equation}
u(t)=E_{\alpha}(-t^{\alpha}L)\,u_{0},\label{eq:MLsolution}
\end{equation}
where $E_{\alpha}$ is the Mittag-Leffler function. 
\end{prop}

\begin{proof}
Let $\widehat{u}(s)=\mathcal{L}\{u\}(s)$ denote the Laplace transform
of $u$. Using the identities 
\[
\mathcal{L}\{u'\}(s)=s\widehat{u}(s)-u_{0},\qquad\mathcal{L}\!\left\{ \int_{0}^{t}K(t-s)Lu(s)\,ds\right\} =\widehat{K}(s)\,L\,\widehat{u}(s),
\]
the Laplace transform of \eqref{eq:VolterraLemma} yields 
\begin{equation}
\bigl(sI+\widehat{K}(s)L\bigr)\,\widehat{u}(s)=u_{0}.\label{eq:LaplaceVolterra}
\end{equation}
For the kernel \eqref{eq:kernelS}, the classical Laplace-transform
formula gives 
\[
\widehat{K}_{S}(s)=s^{1-\alpha},\qquad\Re(s)>0.
\]
Substituting this expression into \eqref{eq:LaplaceVolterra} gives
\[
\bigl(sI+s^{1-\alpha}L\bigr)\,\widehat{u}(s)=u_{0}.
\]
Multiplying both sides by $s^{\alpha-1}$ yields 
\begin{equation}
\bigl(s^{\alpha}I+L\bigr)\widehat{u}(s)=s^{\alpha-1}u_{0}\label{eq:LaplaceStep}
\end{equation}
On the other hand, the Laplace transform of the Caputo derivative
$D_{t}^{\alpha}u$ is 
\[
\mathcal{L}\{D_{t}^{\alpha}u\}(s)=s^{\alpha}\widehat{u}(s)-s^{\alpha-1}u_{0}.
\]
Taking the Laplace transform of \eqref{eq:CaputoLemma} therefore
yields 
\[
s^{\alpha}\widehat{u}(s)-s^{\alpha-1}u_{0}=-L\,\widehat{u}(s),
\]
which rearranges to \eqref{eq:LaplaceStep}. Thus, equations \eqref{eq:VolterraLemma}
and \eqref{eq:CaputoLemma} are equivalent.

Finally, the expression of $\widehat{u}(s)$ in terms of $u_{0}$
can be obtained from \eqref{eq:LaplaceStep} 
\begin{equation}
\widehat{u}(s)=s^{\alpha-1}\bigl(s^{\alpha}I+L\bigr)^{-1}u_{0}.\label{eq:LaplaceSolution}
\end{equation}
The inverse Laplace transform of $s^{\alpha-1}(s^{\alpha}I+L)^{-1}$
is the operator-valued Mittag-Leffler function $E_{\alpha}(-t^{\alpha}L)$,
which proves \eqref{eq:MLsolution}. 
\end{proof}
\begin{rem}
\label{rem:resolvent_and_limit} We note that the Laplace-domain representation
is written in a \emph{resolvent form}. More generally, finite
superpositions of Laplacian semigroups of the form 
\[
u(t)=\sum_{j=1}^{J}e^{-\beta_{j}tL}\,\psi_{j}
\]
lead (via the construction in \S\ref{eq:KhatMode-1}) to operator-valued
memory kernels whose Laplace symbols are \emph{rational} in $s$ (mode-by-mode
in the spectrum of $L$). When the number of terms increases and the
associated time scales become dense, these rational symbols may converge,
in an appropriate sense, to the fractional symbol $s^{\alpha-1}$.
In this resolvent sense, time-fractional diffusion can be viewed as
a singular, scale-free limit of diffusion with operator-valued Volterra
memory. 
\end{rem}

Note that the kernel $K_{S}(t)$ is negative, since $0<\alpha<1$
and thus $\Gamma(\alpha-1)<0$. The stability of the heat Volterra
equation has been studied in \cite{LiZhouGao2018}, for a continuous
spatial setting, with particular attention to positive and negative
kernels. An example of application with a negative kernel in a stochastic setting is given in \cite{Goychuk2018}.
\begin{prop}
\label{thm:SOE_to_fractional_resolvent} Let $L\in\mathbb{R}^{n\times n}$
be a symmetric graph Laplacian with spectrum $\sigma(L)\subset[0,\lambda_{\max}]$
and let $0<\alpha<1$. For $s$ with $\Re(s)>0$, define the \emph{fractional
resolvent family} 
\begin{equation}
\widehat{G}_{S}(s):=s^{\alpha-1}(s^{\alpha}I+L)^{-1}.\label{eq:FS_symbol}
\end{equation}
Assume that for each $h\in\mathbb{N}$ we are given coefficients $a_{j},b_{j}>0$
($j=1,\dots,h$) and define the rational operator-valued symbol 
\begin{equation}
\widehat{G}^{(J)}(s):=\sum_{j=1}^{J}a_{j}\,(sI+b_{j}L)^{-1}.\label{eq:FJ_symbol}
\end{equation}
Suppose there exists a domain $\Omega\subset\{s\in\mathbb{C}:\Re(s)>0\}$
and constants $M,\delta>0$ such that for all $s\in\Omega$: 
\begin{align}
\|\widehat{G}_{S}(s)\| & \le M,\label{eq:bound_FS}\\
\|\widehat{G}_{S}(s)^{-1}\|_{\operatorname{Ran}(L)\to\operatorname{Ran}(L)} & \le M,\label{eq:bound_inv_FS}\\
\|\widehat{G}^{(J)}(s)-\widehat{G}_{S}(s)\| & \le\delta,\qquad\text{with }\ \delta M<1.\label{eq:close_FJ}
\end{align}
(Here $\widehat{G}_{S}(s)^{-1}$ denotes the inverse on $\operatorname{Ran}(L)$,
and the norm is the operator norm on $\mathbb{R}^{n}$.)

Then for every $s\in\Omega$:

$\widehat{G}^{(J)}(s)$ is invertible on $\operatorname{Ran}(L)$ and
\begin{equation}
\|\widehat{G}^{(J)}(s)^{-1}-\widehat{G}_{S}(s)^{-1}\|_{\operatorname{Ran}(L)\to\operatorname{Ran}(L)}\le\frac{M^{2}}{1-\delta M}\,\|\widehat{G}^{(J)}(s)-\widehat{G}_{S}(s)\|.\label{eq:inv_stability}
\end{equation}

Define the \emph{SOE-induced operator-valued memory symbol} by 
\begin{equation}
\widehat{K}^{(J)}(s):=\bigl(\widehat{G}^{(J)}(s)^{-1}-sI\bigr)L^{\dagger},\qquad s\in\Omega,\label{eq:KJ_def}
\end{equation}
where $L^{\dagger}$ is the Moore--Penrose inverse. Define analogously
\begin{equation}
\widehat{K}_{S}(s):=\bigl(\widehat{G}_{S}(s)^{-1}-sI\bigr)L^{\dagger}=s^{1-\alpha}I\quad\text{on }\operatorname{Ran}(L).\label{eq:KS_from_FS}
\end{equation}
Then 
\begin{equation}
(sI+\widehat{K}^{(J)}(s)L)^{-1}=\widehat{G}^{(J)}(s),\qquad(sI+\widehat{K}_{S}(s)L)^{-1}=\widehat{G}_{S}(s),\label{eq:resolvent_identities}
\end{equation}
and moreover the resolvents converge in operator norm: 
\begin{equation}
\|(sI+\widehat{K}^{(J)}(s)L)^{-1}-(sI+s^{\alpha-1}L)^{-1}\|=\|\widehat{G}^{(J)}(s)-\widehat{G}_{S}(s)\|\xrightarrow[J\to\infty]{}0,\qquad s\in\Omega.\label{eq:resolvent_convergence}
\end{equation}
\end{prop}

\begin{proof}
Fix $s\in\Omega$.

\smallskip{}
 \emph{Step 1: stability of invertibility on $\operatorname{Ran}(L)$.}
Let $A:=\widehat{G}_{S}(s)$ and $B:=\widehat{G}^{(J)}(s)$, viewed
as operators on $\operatorname{Ran}(L)$. By assumption $A$ is invertible
there and $\|A^{-1}\|\le M$. Moreover, 
\[
\|A^{-1}(B-A)\|\le\|A^{-1}\|\,\|B-A\|\le M\delta<1.
\]
Hence $B$ is invertible on $\operatorname{Ran}(L)$ and the standard
resolvent identity gives 
\[
B^{-1}-A^{-1}=B^{-1}(A-B)A^{-1}.
\]
Using $\|B^{-1}\|\le\|A^{-1}\|/(1-\|A^{-1}(B-A)\|)\le M/(1-\delta M)$,
we obtain 
\[
\|B^{-1}-A^{-1}\|\le\frac{M}{1-\delta M}\,\|A-B\|\,M=\frac{M^{2}}{1-\delta M}\,\|B-A\|,
\]
which proves \eqref{eq:inv_stability}.

\smallskip{}
 \emph{Step 2: definition of the memory symbols and resolvent identities.}
By definition \eqref{eq:KJ_def}, 
\[
\widehat{G}^{(J)}(s)^{-1}=sI+\widehat{K}^{(J)}(s)L\quad\text{on }\operatorname{Ran}(L),
\]
because $L^{\dagger}L$ is the identity on $\operatorname{Ran}(L)$.
Taking inverses on $\operatorname{Ran}(L)$ yields 
\[
(sI+\widehat{K}^{(J)}(s)L)^{-1}=\widehat{G}^{(J)}(s)\quad\text{on }\operatorname{Ran}(L),
\]
and the same argument gives $(sI+\widehat{K}_{S}(s)L)^{-1}=\widehat{G}_{S}(s)$.
This proves \eqref{eq:resolvent_identities}.

\smallskip{}
\noindent{} \emph{Step 3: identifying $\widehat{K}_{S}(s)$ and resolvent convergence.}
From \eqref{eq:FS_symbol} we have 
\[
\widehat{G}_{S}(s)^{-1}=s^{1-\alpha}(s^{\alpha}I+L)=sI+s^{1-\alpha}L,
\]
hence $\widehat{K}_{S}(s)=s^{\alpha-1}I$ as in \eqref{eq:KS_from_FS}
on $\operatorname{Ran}(L)$. Therefore, 
\[
(sI+\widehat{K}_{S}(s)L)^{-1}=(sI+s^{1-\alpha}L)^{-1}=\widehat{G}_{S}(s).
\]
Finally, \eqref{eq:resolvent_convergence} follows immediately from
\eqref{eq:resolvent_identities} and the assumed closeness \eqref{eq:close_FJ},
and in particular tends to $0$ if $\|\widehat{G}^{(J)}(s)-\widehat{G}_{S}(s)\|\to0$
as $J\to\infty$. 
\end{proof}
An analysis of the subdiffusion equation with memory with relation
to rational kernels in the Laplace-domain can be found in \cite{Ponce2021}.
\begin{rem}
\label{rem:why_prop_matters} Theorem~\ref{thm:SOE_to_fractional_resolvent}
provides a precise way to say that an SOE approximation does not merely
approximate a solution curve, but induces an \emph{approximate evolution
law} in the same analytic class as the limiting fractional dynamics,
namely through the Laplace-domain resolvent. It makes explicit the
structural pathway 
\[
\boxed{\text{\parbox{3.3cm}{\centering SOE }}\;\Longleftrightarrow\;\text{\parbox{3.3cm}{\centering operator-valued }}\;\Longrightarrow\;\text{\parbox{3.3cm}{\centering fractional limit}}}
\]

where the double arrow refers to the exact algebraic construction
$\widehat{G}^{(J)}\leftrightarrow\widehat{K}^{(J)}$ via \eqref{eq:KJ_def}
and \eqref{eq:resolvent_identities}, while the single arrow denotes
the limiting process $J\to\infty$ in the resolvent sense \eqref{eq:resolvent_convergence}.
This is the natural notion of convergence for nonlocal-in-time equations,
since the resolvent family uniquely characterizes the corresponding
linear evolution. 
\end{rem}

\medskip{}

\begin{rem}
\label{rem:what_not_claimed} The conclusions of Theorem~\ref{thm:SOE_to_fractional_resolvent}
are restricted to \emph{Laplace-domain} (resolvent) statements on
a domain $\Omega\subset\{\Re(s)>0\}$. In particular, the proposition
does \emph{not} assert any of the following:

\textbf{Time-domain kernel convergence.} It does not claim that the
inverse Laplace transforms satisfy $K^{(J)}(t)\to K_{S}(t)$ pointwise,
in $L_{\mathrm{loc}}^{1}(\mathbb{R}_{+})$, or in any distributional
sense.

\textbf{Strong convergence of semigroup families in physical time.}
If a physical-time representation involves nonlinear reparametrizations
(e.g.\ $\tau=t^{\alpha}$), the proposition does not claim that uniform
approximation of $G^{(J)}(\tau)$ implies convergence of Laplace transforms
in the physical time variable.

\textbf{Equivalence to a Caputo equation for each finite $J$.} For
finite $J$, the induced kernel is regular (typically a finite exponential
mixture in time), and the proposition does not claim that the corresponding
memory equation is a Caputo fractional differential equation.

\textbf{Global-in-$s$ convergence.} No statement is made about convergence
uniformly for all $\Re(s)>0$, nor about behavior near $s=0$ or $|s|\to\infty$
unless such regions are included in $\Omega$ and the hypotheses are
verified there. 
\end{rem}

\section*{Conclusions}

In this work we have investigated subdiffusion on graphs from a structural,
dynamical, and mechanistic perspective, emphasizing the role of memory
rather than modifying spatial connectivity or transition rules. By
grounding time-fractional diffusion in a random time-change framework,
we have shown that subdiffusion on networks is not merely a slower
version of classical diffusion, but a fundamentally non-Markovian
process in which past states actively shape future evolution. Importantly,
this loss of Markovianity occurs without sacrificing linearity or
mass conservation, making time-fractional models both analytically
tractable and physically meaningful for networked systems.

A central outcome of our analysis is the demonstration that Mittag-Leffler
graph dynamics admit a convex, mass-preserving decomposition into
a finite or infinite superposition of classical heat semigroups evaluated
at rescaled times. This representation provides a transparent interpretation
of subdiffusion as a multi-time-scale diffusion process and forms
the basis for efficient numerical approximations via sum-of-exponentials
schemes. Beyond computational convenience, this viewpoint reveals
how memory operates at the operator level and clarifies the relationship
between fractional dynamics, diffusion with memory kernels, and multi-rate
transport processes.

At the vertex level, we have shown that memory in fractional diffusion
is inherently heterogeneous. Despite being governed by a single fractional
order, different nodes---and even the same node at different times---exhibit
distinct memory biases, favoring remote past, recent past, or present
contributions depending on the local temporal curvature of the evolving
state. This heterogeneity is absent in classical diffusion and highlights
how fractional dynamics encode rich, spatially distributed memory
patterns that reflect both network structure and initial conditions.

These memory effects have concrete consequences for transport and
geometry on graphs. Fractional diffusion induces qualitative biases
such as early-time convex decay at sources, concave growth at neighbors,
algebraic relaxation of modes, and degree-dependent waiting times.
Together, these features reshape how mass propagates through the network,
altering trapping behavior and path selection in ways that cannot
be captured by Markovian models.

Perhaps most strikingly, we have shown that subdiffusive dynamics
enable a particle to discover shortest topological paths between vertices
using only local information. Memory acts as an implicit reinforcement
mechanism that suppresses exploration of suboptimal routes and stabilizes
optimal trajectories, leading shortest subdiffusive paths to coincide
with shortest topological paths. Moreover, among these, subdiffusion
preferentially selects paths traversing high-degree regions, reversing
the hub-avoidance tendency of classical diffusion and revealing a
form of memory-assisted navigation on networks.

Finally, by connecting fractional diffusion to finite superpositions
of classical diffusions and to operator-valued Volterra memory equations,
we have placed time-fractional graph dynamics within a broader hierarchy
of diffusion models with memory. In this view, fractional equations
arise as singular, scale-free limits of multi-rate processes, providing
a unifying framework that links analytical theory, numerical approximation,
and physical interpretation. Taken together, our results clarify what
subdiffusion means on a graph, how memory reshapes network transport,
and why fractional dynamics constitute a natural and structurally
rich extension of classical diffusion on networks. 

\section{Appendix}

\subsubsection{Numerical window tables and usage}

We summarize practical window choices for the log--trapezoidal SOE,
based on the tail bounds previously found. For the \emph{right tail},
the stretched--exponential decay of $M_{\alpha}$ implies that 
\[
\theta_{\max}\;\ge\;\Big(\tfrac{\log(2/\varepsilon)}{c_{\alpha}}\Big)^{\!1/q_{\alpha}},\qquad c_{\alpha}=(1-\alpha)\,\alpha^{\alpha/(1-\alpha)},\quad q_{\alpha}=\frac{1}{1-\alpha},
\]
is sufficient to make ${\displaystyle \sup_{\lambda\in[0,\lambda_{\max}]}\int_{\theta_{\max}}^{\infty}M_{\alpha}(\theta)\,e^{-\theta t\lambda}\,d\theta\lesssim\varepsilon/2}$
\emph{uniformly} in $\lambda$ (i.e., including the zero mode); see
\eqref{eq:right-tail-general}. For the \emph{left tail}, taking $\theta_{\min}=\frac{\varepsilon}{2}\Gamma(1-\alpha)$
ensures ${\displaystyle \sup_{\lambda\in[0,\lambda_{\max}]}\int_{0}^{\theta_{\min}}M_{\alpha}(\theta)\,e^{-\theta t\lambda}\,d\theta\le\varepsilon/2}$
(Prop.~\ref{prop:left-tail}).

\medskip{}
 Tables~\ref{tab:theta-max}--\ref{tab:theta-min} list these cutoffs
for typical $\alpha$ and $\varepsilon$ values, together with $y_{\max}=\log\theta_{\max}$
and $y_{\min}=\log\theta_{\min}$, which are the actual endpoints
used by the SOE in the variable $y=\log\theta$.

\begin{table}[ht]
\centering {\small\setlength{\tabcolsep}{6pt} {}{}{}}{\small{}%
\begin{tabular}{lcccc}
\toprule 
{\small$\alpha$} & {\small$\varepsilon=10^{-6}$} & {\small$\varepsilon=10^{-8}$} & {\small$\varepsilon=10^{-10}$} & {\small$\varepsilon=10^{-12}$}\tabularnewline
\midrule 
{\small 0.20} & {\small 14.016 {[}2.640{]}} & {\small 17.474 {[}2.861{]}} & {\small 20.768 {[}3.033{]}} & {\small 23.936 {[}3.175{]}}\tabularnewline
{\small 0.30} & {\small 11.979 {[}2.483{]}} & {\small 14.529 {[}2.676{]}} & {\small 16.899 {[}2.827{]}} & {\small 19.134 {[}2.951{]}}\tabularnewline
{\small 0.50} & {\small\medspace{}7.618 {[}2.031{]}} & {\small\medspace{}8.744 {[}2.168{]}} & {\small\medspace{}9.740 {[}2.276{]}} & {\small 10.644 {[}2.365{]}}\tabularnewline
{\small 0.70} & {\small\medspace{}4.109 {[}1.413{]}} & {\small\medspace{}4.464 {[}1.496{]}} & {\small\medspace{}4.762 {[}1.561{]}} & {\small\medspace{}5.023 {[}1.614{]}}\tabularnewline
{\small 0.80} & {\small\medspace{}2.816 {[}1.035{]}} & {\small\medspace{}2.976 {[}1.090{]}} & {\small\medspace{}3.107 {[}1.134{]}} & {\small\medspace{}3.219 {[}1.169{]}}\tabularnewline
\bottomrule
\end{tabular}}{\small\caption{Conservative right--tail cutoff $\theta_{\max}$ ensuring ${\displaystyle \sup_{\lambda\in[0,\lambda_{\max}]}\protect\int_{\theta_{\max}}^{\infty}M_{\alpha}(\theta)\,e^{-\theta t\lambda}\,d\theta\lesssim\varepsilon/2}$.
Brackets show $y_{\max}=\log\theta_{\max}$ for direct use in the
log--trapezoidal SOE. Values are \emph{independent} of $t$ and of
the spectrum; they come from the uniform bound \eqref{eq:right-tail-general}
using the stretched--exponential tail of $M_{\alpha}$ with constants
$c_{\alpha},q_{\alpha}$ \cite{MainardiBook,GorenfloKilbasMainardiRogosin}.}
\label{tab:theta-max}}
\end{table}

\begin{table}[ht]
\centering {\small\setlength{\tabcolsep}{6pt} {}{}{}}{\small{}%
\begin{tabular}{lcccc}
\toprule 
{\small$\alpha$} & {\small$\varepsilon=10^{-6}$} & {\small$\varepsilon=10^{-8}$} & {\small$\varepsilon=10^{-10}$} & {\small$\varepsilon=10^{-12}$}\tabularnewline
\midrule 
{\small 0.20} & {\small$5.82\times10^{-7}$ {[}$-14.357${]}} & {\small$5.82\times10^{-9}$ {[}$-18.962${]}} & {\small$5.82\times10^{-11}$ {[}$-23.567${]}} & {\small$5.82\times10^{-13}$ {[}$-28.172${]}}\tabularnewline
{\small 0.30} & {\small$6.49\times10^{-7}$ {[}$-14.248${]}} & {\small$6.49\times10^{-9}$ {[}$-18.853${]}} & {\small$6.49\times10^{-11}$ {[}$-23.458${]}} & {\small$6.49\times10^{-13}$ {[}$-28.063${]}}\tabularnewline
{\small 0.50} & {\small$8.86\times10^{-7}$ {[}$-13.936${]}} & {\small$8.86\times10^{-9}$ {[}$-18.541${]}} & {\small$8.86\times10^{-11}$ {[}$-23.147${]}} & {\small$8.86\times10^{-13}$ {[}$-27.752${]}}\tabularnewline
{\small 0.70} & {\small$1.50\times10^{-6}$ {[}$-13.413${]}} & {\small$1.50\times10^{-8}$ {[}$-18.018${]}} & {\small$1.50\times10^{-10}$ {[}$-22.623${]}} & {\small$1.50\times10^{-12}$ {[}$-27.228${]}}\tabularnewline
{\small 0.80} & {\small$2.30\times10^{-6}$ {[}$-12.985${]}} & {\small$2.30\times10^{-8}$ {[}$-17.590${]}} & {\small$2.30\times10^{-10}$ {[}$-22.195${]}} & {\small$2.30\times10^{-12}$ {[}$-26.800${]}}\tabularnewline
\bottomrule
\end{tabular}}{\small\caption{Left--tail cutoff $\theta_{\min}=\tfrac{\varepsilon}{2}\Gamma(1-\alpha)$
guaranteeing ${\displaystyle \sup_{\lambda\in[0,\lambda_{\max}]}\protect\int_{0}^{\theta_{\min}}M_{\alpha}(\theta)\,e^{-\theta t\lambda}\,d\theta\le\varepsilon/2}$.
Brackets show $y_{\min}=\log\theta_{\min}$. (Rounded to three significant
digits.)}
\label{tab:theta-min}}
\end{table}

Let us now explain how to use these Tables.
\begin{enumerate}
\item Fix a target tolerance $\varepsilon$ for the \emph{tails} (e.g.,
$\varepsilon=10^{-12}$ in double precision). 
\item Read off $(\theta_{\min},y_{\min})$ from Table~\ref{tab:theta-min}
and $(\theta_{\max},y_{\max})$ from Table~\ref{tab:theta-max} for
your $\alpha$ and $\varepsilon$. 
\item Set the SOE window in the log variable as $[y_{\min},\,y_{\max}]$;
then increase $J$ until the \emph{in-window discretization} error
decays to your overall tolerance (Theorem~\ref{thm:geo} shows geometric
decay with $J$). 
\end{enumerate}
Variants and remarks. 
\begin{itemize}
\item (\emph{Mean-zero subspace.}) If you only act on $\mathbf{1}^{\perp}$
and know the spectral gap $\lambda_{2}>0$, you may choose $\theta_{\max}=\frac{1}{t\lambda_{2}}\log(2/\varepsilon)$
(for reuse over $t\in[t_{\min},t_{\max}]$, use $t_{\min}$ here to
keep the nodes $t$-independent) (Prop.~\ref{prop:right-tail-gap}),
which can be \emph{much smaller} than the uniform value in Table~\ref{tab:theta-max}. 
\item (\emph{Heuristic rule of thumb.}) For high-frequency modes it is common
to take $\theta_{\max}\approx\frac{32}{t\,\lambda}$ so that $e^{-\theta_{\max}t\lambda}\approx10^{-14}$
(machine precision). The uniform bound above is conservative and includes
the zero mode. 
\item (\emph{Scaling with $\alpha$.}) As $\alpha\uparrow1$ the tail constant
$c_{\alpha}$ grows and the required $\theta_{\max}$ shrinks; as
$\alpha\downarrow0$ the reverse happens, hence the larger values
in the top row of Table~\ref{tab:theta-max}. 
\end{itemize}
\textbf{Acknowledgments:} E.E. thanks financial support from Agencia
Estatal de Investigación (AEI, MCI, Spain) MCIN/AEI/ 10.13039/501100011033
under grant PID2023-149473NB-I00 and by Agencia Estatal de Investigación
(AEI, MCI, Spain) MCIN/AEI/ 10.13039/501100011033 and Fondo Europeo
de Desarrollo Regional (FEDER, UE) under the Marí­a de Maeztu Program
for units of Excellence in R\&D, grant CEX2021-001164-M) are also
acknowledged.

\bibliographystyle{plain}
\bibliography{References}

@book{SchillingSongVondracekBernstein,
  author    = {Schilling, Ren{\'e} L. and Song, Renming and Vondra\v{c}ek, Zoran},
  title     = {Bernstein Functions: Theory and Applications},
  series    = {De Gruyter Studies in Mathematics},
  volume    = {37},
  edition   = {2},
  publisher = {de Gruyter},
  address   = {Berlin},
  year      = {2012},
  isbn      = {978-3-11-025229-3}
}

@book{MainardiBook,
  author    = {Mainardi, Francesco},
  title     = {Fractional Calculus and Waves in Linear Viscoelasticity: An Introduction to Mathematical Models},
  publisher = {Imperial College Press},
  address   = {London},
  year      = {2010},
  isbn      = {978-1-84816-329-4}
}

@book{GorenfloKilbasMainardiRogosin,
  author    = {Gorenflo, Rudolf and Kilbas, Anatoly A. and Mainardi, Francesco and Rogosin, Sergei V.},
  title     = {Mittag-Leffler Functions, Related Topics and Applications},
  series    = {Springer Monographs in Mathematics},
  publisher = {Springer},
  address   = {Berlin},
  year      = {2014},
  isbn      = {978-3-662-43550-2}
}

@book{MeerschaertSikorskii,
  author    = {Meerschaert, Mark M. and Sikorskii, Alla},
  title     = {Stochastic Models for Fractional Calculus},
  series    = {De Gruyter Studies in Mathematics},
  volume    = {43},
  publisher = {de Gruyter},
  address   = {Berlin},
  year      = {2011},
  isbn      = {978-3-11-025525-6}
}

@book{ApplebaumLevy,
  author    = {Applebaum, David},
  title     = {L{\'e}vy Processes and Stochastic Calculus},
  series    = {Cambridge Studies in Advanced Mathematics},
  volume    = {116},
  edition   = {2},
  publisher = {Cambridge University Press},
  address   = {Cambridge},
  year      = {2009},
  isbn      = {978-0-521-73865-1},
  doi       = {10.1017/CBO9780511809781}
}

@article{TrefethenWeideman2014,
  author    = {Trefethen, Lloyd N. and Weideman, J. A. C.},
  title     = {The Exponentially Convergent Trapezoidal Rule},
  journal   = {SIAM Review},
  year      = {2014},
  volume    = {56},
  number    = {3},
  pages     = {385--458},
  doi       = {10.1137/130932132}
}

@article{BaeumerMeerschaert2001,
  author    = {Baeumer, Boris and Meerschaert, Mark M.},
  title     = {Stochastic solutions for fractional Cauchy problems},
  journal   = {Fractional Calculus and Applied Analysis},
  year      = {2001},
  volume    = {4},
  number    = {4},
  pages     = {481--500}
}

@misc{MainardiMuraPagniniSurvey,
  author    = {Mainardi, Francesco and Mura, Antonio and Pagnini, Gianni},
  title     = {The M-Wright function in time-fractional diffusion processes: a tutorial survey},
  howpublished = {arXiv preprint},
  eprint    = {arXiv:1001.3477},
  year      = {2010},
  url       = {https://arxiv.org/abs/1001.3477}
}

@article{silver2018tuned,
  title={Tuned communicability metrics in networks. The case of alternative routes for urban traffic},
  author={Silver, Grant and Akbarzadeh, Meisam and Estrada, Ernesto},
  journal={Chaos, Solitons \& Fractals},
  volume={116},
  pages={402--413},
  year={2018},
  publisher={Elsevier}
}

@article{estrada2012communicability,
  title={The communicability distance in graphs},
  author={Estrada, Ernesto},
  journal={Linear Algebra and its Applications},
  volume={436},
  number={11},
  pages={4317--4328},
  year={2012},
  publisher={Elsevier}
}

@article{estrada2014hyperspherical,
  title={Hyperspherical embedding of graphs and networks in communicability spaces},
  author={Estrada, Ernesto and Sanchez-Lirola, MG and De La Pe{\~n}a, Jos{\'e} Antonio},
  journal={Discrete Applied Mathematics},
  volume={176},
  pages={53--77},
  year={2014},
  publisher={Elsevier}
}

@article{coifman2006diffusion,
  title={Diffusion maps},
  author={Coifman, Ronald R and Lafon, St{\'e}phane},
  journal={Applied and computational harmonic analysis},
  volume={21},
  number={1},
  pages={5--30},
  year={2006},
  publisher={Elsevier}
}

@article{coifman2005geometric,
  title={Geometric diffusions as a tool for harmonic analysis and structure definition of data: Diffusion maps},
  author={Coifman, Ronald R and Lafon, Stephane and Lee, Ann B and Maggioni, Mauro and Nadler, Boaz and Warner, Frederick and Zucker, Steven W},
  journal={Proceedings of the national academy of sciences},
  volume={102},
  number={21},
  pages={7426--7431},
  year={2005},
  publisher={National Academy of Sciences}
}

@article{Goychuk2018,
	author = {Goychuk, Igor},
	title  = {Viscoelastic subdiffusion in a random Gaussian environment},
	journal  = {Phys. Chem. Chem. Phys.},
	year  = {2018},
	volume  = {20},
	issue  = {37},
	pages  = {24140-24155},
	publisher  = {The Royal Society of Chemistry},
	doi  = {10.1039/C8CP05238G},
	url  = {http://dx.doi.org/10.1039/C8CP05238G},
}

@article{LiZhouGao2018,
	author = {Lingfei, Li 
		and Xiuxiang,Zhou 
		and Hang Gao},
	title = {The stability and exponential stabilization of the heat equation with memory},
	journal = {Journal of Mathematical Analysis and Applications},
	volume = {466},
	number = {1},
	pages = {199-214},
	year = {2018},
	issn = {0022-247X},
	doi = {https://doi.org/10.1016/j.jmaa.2018.05.078},
	url = {https://www.sciencedirect.com/science/article/pii/S0022247X18304797},
}

@article{Ponce2021,
	author = {Ponce, Rodrigo},
	title = {Discrete Subdiffusion Equations with Memory},
	year = {2021},
	issue_date = {Dec 2021},
	publisher = {Springer-Verlag},
	address = {Berlin, Heidelberg},
	volume = {84},
	number = {3},
	issn = {0095-4616},
	url = {https://doi.org/10.1007/s00245-021-09753-z},
	doi = {10.1007/s00245-021-09753-z},
	journal = {Appl. Math. Optim.},
	month = dec,
	pages = {3475?3497},
	numpages = {23},
}

@article{combinido2012crowding,
  title={Crowding effects in vehicular traffic},
  author={Combinido, Jay Samuel L and Lim, May T},
  journal={Plos one},
  volume={7},
  number={11},
  pages={e48151},
  year={2012},
  publisher={Public Library of Science San Francisco, USA}
}

@article{foroozani2019anomalous,
  title={Anomalous information diffusion in social networks: Twitter and Digg},
  author={Foroozani, Ahmad and Ebrahimi, Morteza},
  journal={Expert Systems with Applications},
  volume={134},
  pages={249--266},
  year={2019},
  publisher={Elsevier}
}

@article{kim2024cover,
  title={Cover times of many diffusive or subdiffusive searchers},
  author={Kim, Hyunjoong and Lawley, Sean D},
  journal={SIAM Journal on Applied Mathematics},
  volume={84},
  number={2},
  pages={602--620},
  year={2024},
  publisher={SIAM}
}

@article{jung2024memory,
  title={Memory effects of transcription regulator- DNA interactions in bacteria},
  author={Jung, Won and Chen, Tai-Yen and Santiago, Ace George and Chen, Peng},
  journal={Proceedings of the National Academy of Sciences},
  volume={121},
  number={41},
  pages={e2407647121},
  year={2024},
  publisher={National Academy of Sciences}
}

@article{abadias2020fractional,
  title={Fractional-order susceptible-infected model: definition and applications to the study of COVID-19 main protease},
  author={Abadias, Luciano and Estrada-Rodriguez, Gissell and Estrada, Ernesto},
  journal={Fractional Calculus and Applied Analysis},
  volume={23},
  number={3},
  pages={635--655},
  year={2020},
  publisher={De Gruyter}
}

@article{d2025fractional,
  title={Fractional derivative in continuous-time Markov processes and applications to epidemics in networks},
  author={D'Alessandro, Matteo and Van Mieghem, Piet},
  journal={Physical Review Research},
  volume={7},
  number={1},
  pages={013017},
  year={2025},
  publisher={APS}
}

@article{estrada2020fractional,
  title={Fractional diffusion on the human proteome as an alternative to the multi-organ damage of SARS-CoV-2},
  author={Estrada, Ernesto},
  journal={Chaos: An Interdisciplinary Journal of Nonlinear Science},
  volume={30},
  number={8},
  year={2020},
  publisher={AIP Publishing}
}

@incollection{hoffmann2013random,
  title={Random walks on stochastic temporal networks},
  author={Hoffmann, Till and Porter, Mason A and Lambiotte, Renaud},
  booktitle={Temporal Networks},
  pages={295--313},
  year={2013},
  publisher={Springer}
}

@article{scholtes2014causality,
  title={Causality-driven slow-down and speed-up of diffusion in non-Markovian temporal networks},
  author={Scholtes, Ingo and Wider, Nicolas and Pfitzner, Ren{\'e} and Garas, Antonios and Tessone, Claudio J and Schweitzer, Frank},
  journal={Nature communications},
  volume={5},
  number={1},
  pages={5024},
  year={2014},
  publisher={Nature Publishing Group UK London}
}

@article{lambiotte2015effect,
  title={Effect of memory on the dynamics of random walks on networks},
  author={Lambiotte, Renaud and Salnikov, Vsevolod and Rosvall, Martin},
  journal={Journal of complex networks},
  volume={3},
  number={2},
  pages={177--188},
  year={2015},
  publisher={Oxford University Press}
}

@article{sokolov2012models,
  title={Models of anomalous diffusion in crowded environments},
  author={Sokolov, Igor M},
  journal={Soft Matter},
  volume={8},
  number={35},
  pages={9043--9052},
  year={2012},
  publisher={Royal Society of Chemistry}
}

@article{kosztolowicz2005measuring,
  title={Measuring subdiffusion parameters},
  author={Koszto{\l}owicz, T and Dworecki, K and Mr{\'o}wczy{\'n}ski, St},
  journal={Physical Review E--Statistical, Nonlinear, and Soft Matter Physics},
  volume={71},
  number={4},
  pages={041105},
  year={2005},
  publisher={APS}
}

@article{gallos2007scaling,
  title={Scaling theory of transport in complex biological networks},
  author={Gallos, Lazaros K and Song, Chaoming and Havlin, Shlomo and Makse, Hern{\'a}n A},
  journal={Proceedings of the National Academy of Sciences},
  volume={104},
  number={19},
  pages={7746--7751},
  year={2007},
  publisher={National Academy of Sciences}
}

@article{kepten2015guidelines,
  title={Guidelines for the fitting of anomalous diffusion mean square displacement graphs from single particle tracking experiments},
  author={Kepten, Eldad and Weron, Aleksander and Sikora, Grzegorz and Burnecki, Krzysztof and Garini, Yuval},
  journal={PLoS One},
  volume={10},
  number={2},
  pages={e0117722},
  year={2015},
  publisher={Public Library of Science San Francisco, CA USA}
}

@article{sandev2018generalized,
  title={Generalized diffusion-wave equation with memory kernel},
  author={Sandev, Trifce and Tomovski, Zivorad and Dubbeldam, Johan LA and Chechkin, Aleksei},
  journal={Journal of Physics A: Mathematical and Theoretical},
  volume={52},
  number={1},
  pages={015201},
  year={2018},
  publisher={IOP Publishing}
}

@article{toan2022nonclassical,
  title={The nonclassical diffusion equations with time-dependent memory kernels and a new class of nonlinearities},
  author={Toan, Nguyen Duong and others},
  journal={Glasgow Mathematical Journal},
  volume={64},
  number={3},
  pages={716--733},
  year={2022},
  publisher={Cambridge University Press}
}

@article{trimper2004memory,
  title={Memory-controlled diffusion},
  author={Trimper, Steffen and Zabrocki, Knud and Schulz, Michael},
  journal={Physical Review E--Statistical, Nonlinear, and Soft Matter Physics},
  volume={70},
  number={5},
  pages={056133},
  year={2004},
  publisher={APS}
}

@article{saif2023inverse,
  title={An inverse problem for a two-dimensional diffusion equation with arbitrary memory kernel},
  author={Saif, Summaya and Malik, Salman},
  journal={Mathematical Methods in the Applied Sciences},
  volume={46},
  number={9},
  pages={11007--11020},
  year={2023},
  publisher={Wiley Online Library}
}

@article{metzler2000random,
  title={The random walk's guide to anomalous diffusion: a fractional dynamics approach},
  author={Metzler, Ralf and Klafter, Joseph},
  journal={Physics reports},
  volume={339},
  number={1},
  pages={1--77},
  year={2000},
  publisher={Elsevier}
}

@article{bouchaud1990anomalous,
  title={Anomalous diffusion in disordered media: statistical mechanisms, models and physical applications},
  author={Bouchaud, Jean-Philippe and Georges, Antoine},
  journal={Physics reports},
  volume={195},
  number={4-5},
  pages={127--293},
  year={1990},
  publisher={Elsevier}
}

@book{evangelista2018fractional,
  title={Fractional diffusion equations and anomalous diffusion},
  author={Evangelista, Luiz Roberto and Lenzi, Ervin Kaminski},
  year={2018},
  publisher={Cambridge University Press}
}

@article{sokolov2002fractional,
  title={Fractional kinetics},
  author={Sokolov, Igor M and Klafter, Joseph and Blumen, Alexander},
  journal={Physics Today},
  volume={55},
  number={11},
  pages={48--54},
  year={2002},
  publisher={AIP Publishing}
}

@article{polanowski2014simulation,
  title={Simulation of diffusion in a crowded environment},
  author={Polanowski, Piotr and Sikorski, Andrzej},
  journal={Soft Matter},
  volume={10},
  number={20},
  pages={3597--3607},
  year={2014},
  publisher={Royal Society of Chemistry}
}

@article{fanelli2010diffusion,
  title={Diffusion in a crowded environment},
  author={Fanelli, Duccio and McKane, Alan J},
  journal={Physical Review E--Statistical, Nonlinear, and Soft Matter Physics},
  volume={82},
  number={2},
  pages={021113},
  year={2010},
  publisher={APS}
}

@article{weiss2004anomalous,
  title={Anomalous subdiffusion is a measure for cytoplasmic crowding in living cells},
  author={Weiss, Matthias and Elsner, Markus and Kartberg, Fredrik and Nilsson, Tommy},
  journal={Biophysical journal},
  volume={87},
  number={5},
  pages={3518--3524},
  year={2004},
  publisher={Elsevier}
}

@article{meinecke2017multiscale,
  title={Multiscale modeling of diffusion in a crowded environment},
  author={Meinecke, Lina},
  journal={Bulletin of mathematical biology},
  volume={79},
  number={11},
  pages={2672--2695},
  year={2017},
  publisher={Springer}
}

@article{goychuk2021fingerprints,
  title={Fingerprints of viscoelastic subdiffusion in random environments: Revisiting some experimental data and their interpretations},
  author={Goychuk, Igor and P{\"o}schel, Thorsten},
  journal={Physical Review E},
  volume={104},
  number={3},
  pages={034125},
  year={2021},
  publisher={American Physical Society}
}

@article{chauhan2024quantifying,
  title={Quantifying macrostructures in viscoelastic sub-diffusive flows},
  author={Chauhan, Tanisha and Kalyanaraman, Kaushik and Sircar, Sarthok},
  journal={Journal of Mathematical Physics},
  volume={65},
  number={7},
  year={2024},
  publisher={AIP Publishing}
}

@article{saxton2001anomalous,
  title={Anomalous subdiffusion in fluorescence photobleaching recovery: a Monte Carlo study},
  author={Saxton, Michael J},
  journal={Biophysical journal},
  volume={81},
  number={4},
  pages={2226--2240},
  year={2001},
  publisher={Elsevier}
}

@article{lim2002self,
  title={Self-similar Gaussian processes for modeling anomalous diffusion},
  author={Lim, Soonchieh C and Muniandy, Sithi Vinayakam},
  journal={Physical Review E},
  volume={66},
  number={2},
  pages={021114},
  year={2002},
  publisher={APS}
}

@book{estrada2012structure,
  title={The structure of complex networks: theory and applications},
  author={Estrada, Ernesto},
  year={2012},
  publisher={Oxford University Press}
}

@article{nicolaides2010anomalous,
  title={Anomalous physical transport in complex networks},
  author={Nicolaides, Christos and Cueto-Felgueroso, Luis and Juanes, Ruben},
  journal={Physical Review E--Statistical, Nonlinear, and Soft Matter Physics},
  volume={82},
  number={5},
  pages={055101},
  year={2010},
  publisher={APS}
}

@article{medina2022nontrivial,
  title={Nontrivial and anomalous transport on weighted complex networks},
  author={Medina, Pablo and Carrasco, Sebasti{\'a}n and Correa-Burrows, Paulina and Rogan, Jos{\'e} and Valdivia, Juan Alejandro},
  journal={Communications in Nonlinear Science and Numerical Simulation},
  volume={114},
  pages={106684},
  year={2022},
  publisher={Elsevier}
}

@article{lopez2016overview,
  title={An overview of diffusion in complex networks},
  author={L{\'o}pez-Pintado, Dunia},
  journal={Complex Networks and Dynamics: Social and Economic Interactions},
  pages={27--48},
  year={2016},
  publisher={Springer}
}

@article{masuda2017random,
  title={Random walks and diffusion on networks},
  author={Masuda, Naoki and Porter, Mason A and Lambiotte, Renaud},
  journal={Physics reports},
  volume={716},
  pages={1--58},
  year={2017},
  publisher={Elsevier}
}

@article{basak2019understanding,
  title={Understanding biochemical processes in the presence of sub-diffusive behavior of biomolecules in solution and living cells},
  author={Basak, Sujit and Sengupta, Sombuddha and Chattopadhyay, Krishnananda},
  journal={Biophysical Reviews},
  volume={11},
  number={6},
  pages={851--872},
  year={2019},
  publisher={Springer}
}

@article{grimaldo2019dynamics,
  title={Dynamics of proteins in solution},
  author={Grimaldo, Marco and Roosen--Runge, Felix and Zhang, Fajun and Schreiber, Frank and Seydel, Tilo},
  journal={Quarterly Reviews of Biophysics},
  volume={52},
  pages={e7},
  year={2019},
  publisher={Cambridge University Press}
}

@article{gupta2016protein,
  title={Protein entrapment in polymeric mesh: Diffusion in crowded environment with fast process on short scales},
  author={Gupta, Sudipta and Biehl, Ralf and Sill, Clemens and Allgaier, Jurgen and Sharp, Melissa and Ohl, Michael and Richter, Dieter},
  journal={Macromolecules},
  volume={49},
  number={5},
  pages={1941--1949},
  year={2016},
  publisher={ACS Publications}
}

@article{mainardi2010m,
  title={The M-Wright Function in Time-Fractional Diffusion Processes: A Tutorial Survey},
  author={Mainardi, Francesco and Mura, Antonio and Pagnini, Gianni},
  journal={International Journal of Differential Equations},
  volume={2010},
  number={1},
  pages={104505},
  year={2010},
  publisher={Wiley Online Library}
}

@article{sun2011convergence,
  title={Convergence speed of a fractional order consensus algorithm over undirected scale-free networks},
  author={Sun, Wei and Li, Yan and Li, Changpin and Chen, YangQuan},
  journal={Asian Journal of Control},
  volume={13},
  number={6},
  pages={936--946},
  year={2011},
  publisher={Wiley Online Library}
}

@incollection{lu2012consensus,
  title={Consensus of networked multi-agent systems with delays and fractional-order dynamics},
  author={Lu, Jianquan and Shen, Jun and Cao, Jinde and Kurths, J{\"u}rgen},
  booktitle={Consensus and synchronization in complex networks},
  pages={69--110},
  year={2012},
  publisher={Springer}
}

@article{cao2009distributed,
  title={Distributed coordination of networked fractional-order systems},
  author={Cao, Yongcan and Li, Yan and Ren, Wei and Chen, YangQuan},
  journal={IEEE Transactions on Systems, Man, and Cybernetics, Part B (Cybernetics)},
  volume={40},
  number={2},
  pages={362--370},
  year={2009},
  publisher={IEEE}
}

@article{huang2024distributed,
  title={Distributed Consensus Tracking of Incommensurate Heterogeneous Fractional-Order Multi-Agent Systems Based on Vector Lyapunov Function Method},
  author={Huang, Conggui and Wang, Fei},
  journal={Fractal and Fractional},
  volume={8},
  number={10},
  pages={575},
  year={2024},
  publisher={MDPI}
}

@article{sun2024group,
  title={Group consensus for fractional-order heterogeneous multi-agent systems under cooperation-competition networks with time delays},
  author={Sun, Fenglan and Han, Yunpeng and Zhu, Wei and Kurths, J{\"u}rgen},
  journal={Communications in Nonlinear Science and Numerical Simulation},
  volume={133},
  pages={107951},
  year={2024},
  publisher={Elsevier}
}

@article{yan2024consensus,
  title={Consensus of fractional-order multi-agent systems via observer-based boundary control},
  author={Yan, Xu and Li, Kun and Yang, Chengdong and Zhuang, Jiaojiao and Cao, Jinde},
  journal={IEEE Transactions on Network Science and Engineering},
  volume={11},
  number={4},
  pages={3370--3382},
  year={2024},
  publisher={IEEE}
}

@article{caputo1966linear,
  title={Linear models of dissipation whose Q is almost frequency independent},
  author={Caputo, Michele},
  journal={Annals of Geophysics},
  volume={19},
  number={4},
  pages={383--393},
  year={1966}
}

@article{odibat2006approximations,
  title={Approximations of fractional integrals and Caputo fractional derivatives},
  author={Odibat, Zaid},
  journal={Applied Mathematics and Computation},
  volume={178},
  number={2},
  pages={527--533},
  year={2006},
  publisher={Elsevier}
}

@article{gurtin1968general,
  title={A general theory of heat conduction with finite wave speeds},
  author={Gurtin, Morton E and Pipkin, Allen C},
  journal={Archive for Rational Mechanics and Analysis},
  volume={31},
  number={2},
  pages={113--126},
  year={1968},
  publisher={Springer}
}

@article{kivela2014multilayer,
  title={Multilayer networks},
  author={Kivel{\"a}, Mikko and Arenas, Alex and Barthelemy, Marc and Gleeson, James P and Moreno, Yamir and Porter, Mason A},
  journal={Journal of complex networks},
  volume={2},
  number={3},
  pages={203--271},
  year={2014},
  publisher={Oxford University Press}
}

@article{boccaletti2014structure,
  title={The structure and dynamics of multilayer networks},
  author={Boccaletti, Stefano and Bianconi, Ginestra and Criado, Regino and Del Genio, Charo I and G{\'o}mez-Gardenes, Jes{\'u}s and Romance, Miguel and Sendina-Nadal, Irene and Wang, Zhen and Zanin, Massimiliano},
  journal={Physics reports},
  volume={544},
  number={1},
  pages={1--122},
  year={2014},
  publisher={Elsevier}
}

@article{gomez2013diffusion,
  title={Diffusion dynamics on multiplex networks},
  author={Gomez, Sergio and Diaz-Guilera, Albert and Gomez-Gardenes, Jesus and Perez-Vicente, Conrad J and Moreno, Yamir and Arenas, Alex},
  journal={Physical review letters},
  volume={110},
  number={2},
  pages={028701},
  year={2013},
  publisher={APS}
}

@article{markvorsen2008minimal,
  title={Minimal webs in Riemannian manifolds},
  author={Markvorsen, Steen},
  journal={Geometriae Dedicata},
  volume={133},
  number={1},
  pages={7--34},
  year={2008},
  publisher={Springer}
}

@book{bridson2013metric,
  title={Metric spaces of non-positive curvature},
  author={Bridson, Martin R and Haefliger, Andr{\'e}},
  volume={319},
  year={2013},
  publisher={Springer Science \& Business Media}
}

@article{pang2019fpinns,
  title={fPINNs: Fractional physics-informed neural networks},
  author={Pang, Guofei and Lu, Lu and Karniadakis, George Em},
  journal={SIAM Journal on Scientific Computing},
  volume={41},
  number={4},
  pages={A2603--A2626},
  year={2019},
  publisher={SIAM}
}

@article{joshi2023survey,
  title={A survey of fractional calculus applications in artificial neural networks},
  author={Joshi, Manisha and Bhosale, Savita and Vyawahare, Vishwesh A},
  journal={Artificial Intelligence Review},
  volume={56},
  number={11},
  pages={13897--13950},
  year={2023},
  publisher={Springer}
}

@article{MeerschaertNaneVellaisamy2011,
  author  = {Meerschaert, Mark M. and Nane, Erkan and Vellaisamy, P.},
  title   = {The fractional Poisson process and the inverse stable subordinator},
  journal = {Electronic Journal of Probability},
  volume  = {16},
  pages   = {1600--1620},
  year    = {2011}
}

\end{document}